\documentclass[11pt]{article}

\usepackage{latexsym}
\usepackage{amsmath}
\usepackage{amssymb}
\usepackage{amsthm}
\usepackage{hyperref}
\usepackage{cite}
\usepackage{graphicx}
\usepackage{color}
\usepackage{enumitem}
\usepackage{bm}
\usepackage[left=1in,right=1in,top=1in,bottom=1in]{geometry}
\usepackage{xspace}
\setlist[description]{font=\normalfont\itshape\textbullet\space}
\usepackage[all]{xy}
\usepackage{doi, url}

\renewcommand{\paragraph}[1]{\vspace{6pt} \noindent \textbf{#1}\xspace}

\theoremstyle{plain}
\newtheorem{theorem}{Theorem}[section]
\newtheorem{maintheorem}{Theorem}

\newtheorem*{thmsearch}{Theorem~\ref{thm:search_decision}}
\newtheorem*{corsearch}{Corollary~\ref{cor:search_decision}}
\newtheorem*{thmdto3}{Theorem~\ref{thm:d_to_3}}
\newtheorem*{corp}{Corollary~P}
\newcommand{\refcorp}{Cor.~\hyperref{}{}{cor:p}{P}}
\newtheorem{maincorollary}{Corollary}

\newtheorem{corollary}[theorem]{Corollary}
\newtheorem{lemma}[theorem]{Lemma}
\newtheorem{observation}[theorem]{Observation}
\newtheorem{proposition}[theorem]{Proposition}

\theoremstyle{definition}

\newtheorem{remark}[theorem]{Remark}

\newtheorem{definition}[theorem]{Definition}
\newtheorem{example}[theorem]{Example}
\newtheorem{question}[theorem]{Open Question}

\newcommand{\SL}{\mathrm{SL}}
\newcommand{\GL}{\mathrm{GL}}
\newcommand{\F}{\mathbb{F}}
\newcommand{\Z}{\mathbb{Z}}
\newcommand{\Q}{\mathbb{Q}}
\newcommand{\C}{\mathbb{C}}
\newcommand{\R}{\mathbb{R}}

\renewcommand{\hom}{\mathrm{Hom}}

\newcommand{\rk}{\mathrm{rk}}

\newcommand{\poly}{\mathrm{poly}}

\newcommand{\M}{\mathrm{M}}
\newcommand{\T}{\mathrm{T}}
\renewcommand{\S}{\mathrm{S}}
\newcommand{\Mon}{\mathrm{Mon}}

\newcommand{\tuple}[1]{\mathbf{#1}}
\newcommand{\tens}[1]{\mathtt{#1}}
\newcommand{\spa}[1]{\mathcal{#1}}

\newcommand{\cA}{\spa{A}}
\newcommand{\cB}{\spa{B}}

\newcommand{\tA}{\tens{A}}
\newcommand{\tB}{\tens{B}}
\newcommand{\tE}{\tens{E}}
\newcommand{\tF}{\tens{F}}

\newcommand{\vA}{\tuple{A}}
\newcommand{\vB}{\tuple{B}}

\newcommand{\vzero}{\tuple{0}}
\newcommand{\tzero}{\tens{0}}

\newcommand{\oracle}{\mathcal{O}}

\newcommand{\linspan}{\mathrm{span}}

\newcommand{\algprobm}[1]{{\sc #1}\xspace}
\newcommand{\GI}{\algprobm{GI}}
\newcommand{\GIlong}{\algprobm{Graph Isomorphism}}
\newcommand{\GpI}{\algprobm{GpI}}
\newcommand{\GpIlong}{\algprobm{Group Isomorphism}}
\newcommand{\CubicFormlong}{\algprobm{Cubic Form Equivalence}}
\newcommand{\DFormlong}{\algprobm{Degree-$d$ Form Equivalence}}
\newcommand{\NcCubicFormlong}{\algprobm{Trilinear} \algprobm{Form} \algprobm{Equivalence}} 
\newcommand{\AlgIsolong}{\algprobm{Algebra Isomorphism}}
\newcommand{\AltMatSpIsomlong}{\algprobm{Alternating Matrix Space Isometry}}
\newcommand{\AltMatSpIsomWords}{\algprobm{Alternating} \algprobm{Matrix} \algprobm{Space} \algprobm{Isometry}} 
\newcommand{\SymMatSpIsomlong}{\algprobm{Symmetric Matrix Space Isometry}}
\newcommand{\MatSpIsomlong}{\algprobm{Matrix Space Isometry}}
\newcommand{\MatSpConjlong}{\algprobm{Matrix Space Conjugacy}}
\newcommand{\MatSpIsomlongWords}{\algprobm{Matrix} \algprobm{Space} \algprobm{Isometry}} 
\newcommand{\MatSpConjlongWords}{\algprobm{Matrix} \algprobm{Space} \algprobm{Conjugacy}} 
\newcommand{\AltMatSpMonIsomlong}{\algprobm{Alternating Matrix Space Monomial Isometry}}
\newcommand{\MatSpEquivlong}{\algprobm{Matrix Space Equivalence}}
\newcommand{\ThreeTI}{\algprobm{3TI}}
\newcommand{\DeeTI}{$d$\algprobm{TI}}
\newcommand{\ThreeTIlong}{\algprobm{3-Tensor Isomorphism}}

\newcommand{\TI}{\algprobm{TI}}
\newcommand{\TIlong}{\algprobm{Tensor Isomorphism}}
\newcommand{\DeeTIlong}{\algprobm{$d$-Tensor Isomorphism}}
\newcommand{\CodeEq}{\algprobm{CodeEq}}
\newcommand{\CodeEqlong}{\algprobm{Code Equivalence}}
\newcommand{\MonCodeEqlong}{\algprobm{Monomial Code Equivalence}}
\newcommand{\MatLieConjlong}{\algprobm{Matrix Lie Algebra Conjugacy}}

\newcommand{\cc}[1]{\mathsf{#1}}

\DeclareMathOperator{\id}{id}
\DeclareMathOperator{\rank}{rank}
\DeclareMathOperator{\diag}{diag}
\DeclareMathOperator{\chr}{char}

\newcommand{\too}%
{\xrightarrow{\text{\raisebox{-3pt}{$\sim$}}\,}}

\newcommand{\logbf}{\mathbf{log}}
\newcommand{\expbf}{\mathbf{exp}}

\usepackage[T1]{fontenc}
\def\DJ{{\hbox{D\kern-.8em\raise.15ex\hbox{--}\kern.35em}}}

\usepackage{arydshln}
\hypersetup{
    pdftitle={Isomorphism problems for groups, 3-tensors, and cubic forms: 
completeness and reductions},
    pdfauthor={Joshua A. Grochow and Youming Qiao},
    bookmarksnumbered=true,     
    bookmarksopen=true,         
    bookmarksopenlevel=1,       
    colorlinks=true,            
    pdfstartview=Fit,           
    pdfpagemode=UseOutlines,    
    pdfpagelayout=TwoPageRight
}

\title{
Isomorphism problems for tensors, groups, 
and cubic forms: 
completeness and reductions
}

\author{
Joshua A. Grochow
\footnote{Departments of Computer Science and Mathematics, University of Colorado, Boulder. \tt{jgrochow@colorado.edu}}
\and
Youming Qiao
\footnote{Centre for Quantum Software and Information, University of 
Technology Sydney. \tt{youming.qiao@uts.edu.au}}
}

\date{\today}

\begin{document}

\pagenumbering{gobble}  

\maketitle

\begin{abstract}
In this paper we consider the problems of testing isomorphism of tensors, 
$p$-groups, cubic 
forms, algebras, and more, which arise from a variety of areas, including machine learning, group theory, and cryptography. These problems can all be cast as orbit problems on multi-way 
arrays under different group actions. Our first two main results are:
 \begin{enumerate}
\item \label{abs:TI-complete}All the aforementioned isomorphism problems are equivalent under polynomial-time reductions, in conjunction with the recent results of Futorny--Grochow--Sergeichuk 
(\emph{Lin. Alg. Appl.}, 2019).

\item \label{abs:d-to-3} Isomorphism of $d$-tensors reduces to isomorphism of $3$-tensors, for any $d \geq 3$. 
\end{enumerate}
All but one of the reductions for the preceding contributions work over arbitrary fields. Together they
suggest that the aforementioned isomorphism problems form a rich and robust equivalence class,
which we call \TIlong-complete, or \TI-complete for short.
Furthermore, this provides a unified viewpoint on these hard isomorphism testing 
problems arising from a variety of areas.

We then leverage the techniques used in the above results to prove two first-of-their-kind results for \GpIlong (\GpI):
\begin{enumerate}
\setcounter{enumi}{2}
\item \label{abs:class} We give a reduction from testing isomorphism of $p$-groups of exponent $p$ and small class ($c < p$) to isomorphism of $p$-groups of exponent $p$ and class 2. The latter are widely believed to be the hardest cases of \GpI, but as far as we know, this is the first reduction from any more general class of groups to this class.

\item \label{abs:search-to-decision} We give a search-to-decision reduction for isomorphism of $p$-groups of exponent $p$ and class $2$ in time $|G|^{O(\log \log |G|)}$. While search-to-decision reductions for \GIlong (\GI) have been known for more than 40 years, as far as we know this is the first non-trivial search-to-decision reduction in the context of \GpI.
\end{enumerate}

Our main technique for (\ref{abs:TI-complete}), (\ref{abs:class}), and (\ref{abs:search-to-decision}) is a 
linear-algebraic analogue of the classical graph coloring gadget, which was used 
to obtain the search-to-decision reduction for \algprobm{GI}. This gadget 
construction may be of independent interest and utility.
The technique for (\ref{abs:d-to-3}) gives a method for encoding an arbitrary tensor into an algebra.
\end{abstract}

\newpage
\pagenumbering{arabic}     

\section{Introduction}\label{sec:intro}


\newcommand{\Sec}[1]{Sec.~\ref{#1}}
\newcommand{\Thm}[1]{Thm.~\ref{#1}}
\newcommand{\Obs}[1]{Obs.~\ref{#1}}
\newcommand{\Prop}[1]{Prop.~\ref{#1}}
\newcommand{\Cor}[1]{Cor.~\ref{#1}}
\newcommand{\Lem}[1]{Lem.~\ref{#1}}

\paragraph{Isomorphism problems in light of Babai's breakthrough on Graph 
Isomorphism.}
In late 2015, Babai presented a quasipolynomial-time algorithm for  
\GIlong (\GI) \cite{Bab16}. This is widely regarded as one of the major 
breakthroughs in 
theoretical computer science of the past decade. Indeed, \GI
 has been at the heart of complexity theory 
nearly since its inception: both Cook and Levin were thinking about \GI when they 
defined $\cc{NP}$ \cite[Sec. 1]{AllenderDas}, \algprobm{Graph (Non-)Isomorphism} 
played a special role in the creation of the 
class $\cc{AM}$ \cite{babai85, GMR85, BM88}, 
and it still stands today 
as one of the few 
natural candidates for a problem that 
is ``$\cc{NP}$-intermediate,'' that is, in $\cc{NP}$, but neither in $\cc{P}$ nor 
$\cc{NP}$-complete \cite{Ladner} (see \cite{StackExchangeIntermediate} for 
additional 
candidates). Beyond its practical applications (e.\,g., \cite{SV17, irniger} and references therein) and its 
naturality, part of its fascination comes from its universal property: \GI is 
universal for isomorphism problems for ``explicitly given'' structures \cite[Sec.~15]{ZKT}, that is, first-order structures on a set $V$ 
where, e.\,g., a $k$-ary relation on $V$ is given by listing out a subset $R 
\subseteq V^k$. 

In light of Babai's breakthrough
on \GI \cite{Bab16}, it is natural to consider ``what's next?'' 
for isomorphism problems. That is, what isomorphism problems stand as crucial 
bottlenecks to further improvements on \GI, and what isomorphism problems should 
naturally draw our attention for further exploration? Of course, 
one of the main open questions in the area remains whether or not \GI is in $\cc{P}$.
Babai \cite[arXiv version, Sec.~13.2 and 13.4]{Bab16} already lists several 
isomorphism problems for further study, including \GpIlong, 
\algprobm{Linear Code Equivalence}, and \algprobm{Permutation Group Conjugacy}. In 
this paper we expand this list in what we argue is a very natural direction, 
namely to \emph{isomorphism problems for multi-way arrays}, also known as 
tensors.\footnote{There have been some disputes on the terminologies; see the 
preface of \cite{Lan12}. Our approach is to use 
multi-way arrays as the basic 
underlying object, and to use tensors as the multi-way arrays under a certain group 
action.}

\paragraph{Group actions on 3-way arrays.} 3-way arrays 
are simply arrays with 3 indices, generalizing the case of matrices (=2-way arrays). In this paper we consider entries of the arrays being from a field 
$\F$, so a 3-way array is just $\tA=(a_{i,j,k})$, $i\in[\ell]$, $j\in[n]$, 
$k\in[m]$, and $a_{i,j,k}\in\F$. 

Let $\GL(n, \F)$ be the general linear group of degree $n$ over $\F$, and let 
$\M(n,\F)$ denote the set of $n \times n$ matrices. There are three natural group 
actions on $\M(n,\F)$: for $A \in \M(n,\F)$, (1) $(P,Q) \in \GL(n,\F) \times 
\GL(n,\F)$ sends $A$ to $P^t A Q$, (2) $P \in \GL(n,\F)$ sends $A$ to $P^{-1} A 
P$, and (3) $P \in \GL(n,\F)$ sends $A$ to $P^t A P$.
These three actions 
then endow $A$ with different algebraic/geometric interpretations: (1) a linear map 
from a vector space $V$ to another vector space $W$, (2) a linear map from $V$ to 
itself, and (3) a bilinear map from $V\times V$ to $\F$. 

Likewise, 3-way arrays $\tA=(a_{i,j,k})$, $i, j, k\in[n]$, can be naturally acted 
by $\GL(n, \F)\times \GL(n, \F)\times \GL(n, \F)$ in one way, by $\GL(n, \F)\times 
\GL(n, \F)$ in two different ways, and by $\GL(n, \F)$ in two different ways. 
These five actions endow various families of 3-way arrays with 
different algebraic/geometric 
meanings, including 3-tensors, bilinear maps, matrix (associative or Lie) 
algebras, and trilinear forms (a.k.a. 
non-commutative cubic forms). (See \Sec{sec:problems} for detailed explanations.) 
Over finite fields, the associated isomorphism problems are in $\cc{NP}\cap\cc{coAM}$, following the essentially same 
$\cc{coAM}$ 
protocol as for \GI.

With these group actions in mind, 3-way arrays capture a variety of important structures in several mathematical and computational disciplines. 
They arise naturally in quantum mechanics (states are described by 
tensors), the complexity of matrix multiplication (matrix multiplication is
described by a tensor, and its algebraic complexity is essentially its tensor 
rank), the Geometric Complexity Theory approach \cite{Mul11} to the Permanent versus 
Determinant Conjecture \cite{Val79} (tensors describe the boundary of the 
determinant orbit 
closure, e.\,g., \cite[Sec.~13.6.3]{Lan12} and
\cite[Sec.~3.5.1]{grochowPhD} for introductions,
and \cite{HL16, H17} for applications), data analysis \cite{KB09},  
machine learning 
\cite{PSS18}, computational group theory \cite{LQ17, BMW18}, and cryptography 
\cite{Pat96,JQSY19}.

\paragraph{Main results.} 
The five natural actions on 3-way arrays mentioned above each lead to a different 
isomorphism problem on 3-way arrays; we discuss these problems and their 
interpretations in \Sec{sec:problems}. 
Our first main result, \Thm{thm:main}, shows that these 
isomorphism problems for 3-way arrays are all equivalent under 
polynomial-time reductions. Due to the 
algebraic or geometric interpretations, these problems are further 
equivalent to isomorphism problems on certain classes of groups, 
cubic forms, trilinear forms (a.k.a. non-commutative cubic forms), 
associative algebras, and Lie algebras. 
One consequence of these results (\refcorp), along with those of \cite{FGS19}, is a reduction 
from \GpI for $p$-groups of exponent $p$ and class $< p$ to \GpI for $p$-groups of 
exponent $p$ and class 2. Although the latter have long been believed to be the 
hardest cases of \GpI, as far as we are aware, this is the first reduction from a 
more general class of groups to this class.

Although these equivalences may have been expected by some experts, it had not been 
immediately clear to us for some time during this project. 
To get a sense for the non-obviousness, let us postulate  the following hypothetical question. 
Recall that two matrices $A, B\in \M(n, \F)$ are called \emph{equivalent} if 
there exists $P, Q\in\GL(n, \F)$ such that $P^{-1}AQ=B$, and they are \emph{conjugate} if 
there 
exists $P\in \GL(n,\F)$ such that $P^{-1}AP=B$. Can we reduce testing \algprobm{Matrix Conjugacy}
to testing \algprobm{Matrix Equivalence}? Of course since they are both in 
$\cc{P}$ there is a trivial reduction; to avoid this, let us consider only 
reductions $r$ which send a matrix $A$ to a matrix $r(A)$ such that $A$ and $B$ 
are conjugate iff $r(A)$ and $r(B)$ are equivalent. Nearly all reductions between 
isomorphism problems that we are aware of have this form (so-called ``kernel 
reductions'' \cite{FortnowGrochowPEq}; 
cf. functorial reductions \cite{BabaiSR}). 
After some thought, we realize that this is essentially 
impossible. The reason is that the equivalence class of a matrix is completely determined by its 
rank, while the conjugacy class of a matrix is determined by its rational canonical form. Among $n \times n$ matrices there are only $n+1$ equivalence classes, but there are at least $|\F|^n$ rational canonical forms (coming from the choice of minimal polynomial/companion matrix). Even when $\F$ is a finite field, such a reduction would thus require an exponential increase in dimension, and when $\F$ is infinite, such a reduction is impossible (regardless of running time).

Nonetheless, one of our results is that for \emph{spaces} of matrices (one form of 3-way arrays), conjugacy testing does indeed reduce to equivalence testing!
This is in sharp contrast to the case of single matrices. In 
the above setting, it means that there exists a polynomial-time computable map 
$\phi$ from $\M(n, \F)$ to \emph{subspaces of} $\M(s, \F)$, such that $A, B$ are 
conjugate up to a scalar if and only if $\phi(A), \phi(B)\leq \M(s, \F)$ are 
equivalent as matrix spaces. Such a reduction may not be clear 
at first sight.

Our second main result reduces \DeeTI to \ThreeTI, for any fixed $d \geq 3$. 
From one viewpoint, this can be seen as a linear algebraic analogue of the 
now-classical reduction from $d$-uniform \algprobm{Hypergraph Isomorphism} to \GI (e.\,g., \cite{ZKT}). 
However, as the reader will see, the reduction here is quite a bit more involved, 
using quiver algebras and the Wedderburn--Mal'cev Theorem on complements of the 
Jacobson radical in associative algebras. From another viewpoint, this can be seen 
as a step towards showing that \ThreeTI is not only universal among isomorphism 
problems on 3-way arrays \cite{FGS19}, but perhaps \ThreeTI is already universal 
for isomorphism problems on $d$-way arrays for any $d$; see \Sec{sec:universality}.
These first two results indicate the robustness and naturality of the notion of 
$\cc{TI}$-completeness.

Our next set of results reduce \GIlong and \algprobm{Linear Code Equivalence} to these 
isomorphism problems for 3-way arrays (\Sec{sec:GI_code}). This 
shows that these isomorphism problems for 3-way arrays form a set of potentially 
harder problems than these two problems, as 
also supported by the current difference in their practical 
difficulties.\footnote{There is a heuristic 
algorithm for \algprobm{Linear Code Equivalence} by Sendrier \cite{Sen00}, which 
is practically effective in many cases, 
though for self-dual codes it reverts to an exponential search.}  It currently 
seems unlikely to us that either 
\GIlong or \CodeEqlong is $\cc{TI}$-complete. 

Finally, our third main contribution is to 
show a search-to-decision reduction for these 
tensor problems (\Thm{thm:search_decision}), which may be of independent interest,  leveraging our technique from above.
While such a reduction has long been known for \GI, for \GpIlong in general this 
remains a long-standing open question. Our techniques allow us to give a  
search-to-decision reduction for isomorphism of $p$-groups of class 2 and exponent $p$ in time 
$|G|^{O(\log \log |G|)}$ in the model of matrix groups over finite fields.
This group class is widely regarded to be the hardest cases of 
\GpIlong. As far as we know, this is the first 
non-trivial search-to-decision reduction for testing isomorphism of a class of finite groups. 

\paragraph{Implications of main results for practical 
computations.} Our first 
main result may partly help to explain the difficulties from various areas when dealing 
with these isomorphism problems. There is currently a significant difference between isomorphism 
problems for 3-way arrays and that for graphs. Namely, in sharp contrast to 
\GIlong---for 
which very effective practical algorithms have existed for some
time \cite{McK80,MP14}---the 
problems we consider here all still pose great difficulty even on relatively small 
examples in practice. Indeed, such problems have been proposed to be difficult 
enough 
for cryptographic purposes \cite{Pat96,JQSY19}. As further evidence of their 
practical 
difficulty, current 
algorithms implemented for \AltMatSpIsomlong\footnote{An $n\times n$ matrix $A$ 
over $\F$ is alternating if for every $v\in \F^n$, $v^tAv=0$. When $\F$ is not of 
characteristic $2$, this is equivalent to the skew-symmetry condition.}---a 
problem we show is $\cc{TI}$-complete---can handle the 
cases when the 3-way array is of size $10\times 10\times 10$ over $\F_{13}$, but 
absolutely  not for 3-way arrays of size $100\times 100\times 100$, even though in 
this case the input can still be stored  in only a 
few megabytes.\footnote{We thank James B. Wilson, who maintains a suite of 
algorithms for $p$-group isomorphism testing, for communicating this insight to us 
from his hands-on 
experience. We of course maintain responsibility for any possible 
misunderstanding, or lack of knowledge regarding the performance of other 
implemented algorithms.}
In \cite{PSS18}, motivated by machine learning applications, 
computations on one $\cc{TI}$-complete problem were performed in Macaulay2 \cite{M2}, but these could 
not go beyond small examples either. 
Our results imply that the complexities of these problems arising in many fields%
---from computational group theory to cryptography to machine learning---are all 
equivalent.

\paragraph{Isomorphism problems for 3-way arrays as a bottleneck for graph 
isomorphism.}
In addition to their many 
incarnations and practical uses mentioned above, the isomorphism problems we 
consider on 3-way arrays can be further motivated by their relationship to \GI. 
Specifically, these problems both form a key bottleneck to putting \GI into 
$\cc{P}$, and pose a great challenge for extending techniques used to solve \GI.

Isomorphism problems for 3-way arrays 
stand as a key bottleneck to put \GI in $\cc{P}$. 
This is because, as Babai 
pointed out 
\cite{Bab16}, \GpIlong is a key bottleneck to putting \GI into $\cc{P}$. 
Indeed, the current-best upper bounds on these two problems are now quite 
close: $n^{O(\log n)}$ for \GpIlong (originally due to \cite{FN70, 
Mil78}\footnote{Miller attributes this to 
Tarjan.}, with improved constants   
\cite{Wil14, Ros13a, Ros13b}), 
and 
$n^{O(\log^2 n)}$ for \GI \cite{Bab16} (see \cite{HBD17} for calculation of the exponent). 
Within \GpIlong, it is 
widely regarded, for several reasons 
(e.\,g., \cite{Bae38, HigmanEnum, SergeichukPgpWild, WilsonWildSlides}), that the 
bottleneck is the class of $p$-groups 
of class 2 and exponent $p$ (i.e., $G/Z(G)$ is abelian and $g^p=1$ for all $g$, 
$p$ odd). 
Then 3-way arrays enter the picture by Baer's Correspondence 
\cite{Bae38}, which shows that the isomorphism problem for these groups is 
equivalent to 
telling whether two linear spaces of skew-symmetric matrices over $\F_p$ are 
equivalent up to transformations of the form $A \mapsto P^t A P$. This is the 
\AltMatSpIsomlong problem, which we show in this paper is 
$\cc{TI}$-complete.\footnote{Because of the difference in verbosity of inputs, 
solving \GpIlong for this class of groups in time $\poly(|G|)$ 
is equivalent  to solving \AltMatSpIsomlong in time $p^{O(n+m)}$ for $n\times n$ matrix 
spaces of dimension $m$ over $\F_p$. The current state of the art is 
$p^{O(n^2)}$, which corresponds to the nearly-trivial upper bound of $|G|^{O(\log 
|G|)}$ on \GpIlong.}

To see why the techniques for \GI face great difficulty when dealing with 
isomorphism problems for multi-way arrays, recall that
most algorithms for \GI, including Babai's  
\cite{Bab16}, 
are built on 
two families of techniques: group-theoretic, and combinatorial. 
One of the main differences is 
that the underlying group 
action for \GI is a permutation group acting on a combinatorial structure, whereas 
the underlying group actions for isomorphism problems for 3-way arrays are matrix 
groups acting on (multi)linear structures. 

Already in moving from permutation groups to matrix groups, we find many new
computational difficulties 
 that arise naturally in basic subroutines used in isomorphism testing. For 
 example, the membership problem for 
permutation groups is well-known to be efficiently solvable by Sims's algorithm 
\cite{Sim78} (see, e.\,g., \cite{Ser03} for a textbook treatment),
while for matrix groups this was only recently 
shown to be 
solvable with a number-theoretic 
oracle over finite fields of odd characteristic 
\cite{BBS09}. 
Correspondingly, 
when moving from combinatorial structures to (multi)linear 
algebraic structures, we also find severe limitation on 
the use of most combinatorial techniques, like individualizing a vertex. For example,  
it is quite expensive to 
enumerate all 
vectors in a vector space, while it is usually considered efficient to go through 
all 
elements in a set. 
Similarly, within a set, any subset has a unique complement, whereas within 
$\F_q^n$, a subspace can have up to $q^{\Theta(n^2)}$ complements. 

Given all the differences between the combinatorial and linear-algebraic 
worlds, it may be surprising that combinatorial techniques for \GIlong can nonetheless be useful for \GpIlong. Indeed, guided by the 
postulate that alternating matrix spaces can be viewed as a linear algebraic 
analogue of graphs, Li and the second author \cite{LQ17} adapted the individualisation and 
refinement technique, as used by Babai, Erd\H{o}s and Selkow \cite{BES80}, to 
tackle \AltMatSpIsomlong over $\F_q$. This algorithm was recently 
improved \cite{BGL+19}. However, this technique, though helpful to improve from 
the brute-force
$q^{n^2}\cdot \poly(n, \log q)$ time, seems still limited to getting $q^{O(n)}$-time algorithms.

\paragraph{New techniques.} 
Our first new 
technique for the above results on 3-way arrays is to develop a linear-algebraic analogue 
of the coloring gadget used in the 
context of \GIlong (see, e.\,g., \cite{KST93}). 
These gadgets
help us to restrict to various subgroups of the general linear group. 
Recall that, in relating \GI with other isomorphism problems, coloring is a very 
useful idea. Given a graph $G=(V, E)$, a coloring of vertices is 
 a function $c:V \to C$ where $C$ is a set of ``colors.'' Colored isomorphism between two 
vertex-colored graphs asks only for isomorphisms that send vertices of one color 
to vertices of that same color.
If we are interested in 
making a specific vertex $v\in V$ special (``individualizing'' that vertex), we can assign this 
vertex a unique color. To reduce \algprobm{Colored Graph Isomorphism} to 
ordinary \GIlong uses certain gadgets, and we adapt this idea to the context of 
3-way arrays.  We note that 
\cite{FGS19} construct a related such gadget. In this paper, we develop a new gadget which we use both by itself, and in combination with the gadget from \cite{FGS19} (albeit in a new context), 
see \Sec{sec:related} and \Sec{sec:reduction_gadget}.

Our second new technique, used to show the reduction from \DeeTI to \ThreeTI, is a 
simultaneous generalization of our reduction from \ThreeTI to \AlgIsolong and the 
technique Grigoriev used \cite{Grigoriev83} to show that isomorphism in a certain 
restricted class of algebras is equivalent to \GI. In brief 
outline: a 3-way 
array $\tA$ specifies the structure constants of an algebra with basis $x_1, 
\dotsc, x_n$ via $x_i \cdot x_j := \sum_{k} \tA(i,j,k) x_k$, and this is 
essentially how we use it in the reduction from \ThreeTI to \AlgIsolong. For 
arbitrary $d \geq 3$, we would like to similarly use a $d$-way array $\tA$ to 
specify how $d$-tuples of elements in some algebra $\cA$ multiply. The issue is 
that for $\cA$ to be an algebra, our construction must still specify how 
\emph{pairs} of elements multiply. The basic idea is to let pairs (and triples, 
and so on, up to $(d-2)$-tuples) multiply ``freely'' (that is, without additional 
relations), and then to use $\tA$ to rewrite any product of $d-1$ generators as a 
linear combination of the original generators. While this construction as 
described already gives one direction of the reduction (if $\tA \cong \tB$, then 
$\cA \cong \cB$), the other direction is trickier. For that, we modify the 
construction to an algebra in which short products (less than $d-2$ generators) do 
not quite multiply freely, but almost. After the fact, we found out that this 
construction generalizes the one used by Grigoriev \cite{Grigoriev83} to show that 
\GI was equivalent \AlgIsolong for a certain class of algebras (see 
\Sec{sec:related} for a comparison).

\paragraph{Organization.} We aim to reach as wide an audience as 
possible, so we start with a detailed introduction to the 
various isomorphism problems on 
3-way arrays, and their algebraic and geometric 
interpretations in \Sec{sec:problems}. We then describe our results in 
detail in \Sec{sec:result} and consider related work in \Sec{sec:related}. An illustration of the key technique is in 
\Sec{sec:technique}. These sections may be viewed as 
an extended abstract. 

The remainder of the paper gives detailed proofs of all results. \Sec{sec:prel} contains additional
preliminaries. In \Sec{sec:reduction_gadget}, we present those 
reductions which use the linear-algebraic coloring technique, thus proving
\Thm{thm:main}(\ref{thm:main:isom}) and 
\Thm{thm:search_decision}. We then finish the proof of 
\Thm{thm:main} 
by presenting the remaining reductions in \Sec{sec:reduction_other}. 
\Thm{thm:d_to_3} is proved in \Sec{sec:dto3}.
In \Sec{sec:conclusion}, we put forward a theory of 
universality for basis-explicit linear structures, in analogy with \cite{ZKT}. 
While not yet complete, this seems to provide another justification for studying 
\TIlong and related problems, and it motivates some interesting open 
questions. In Appendix~\ref{app:cubic} we give a reduction from \CubicFormlong to \DFormlong for any $d \geq 3$ (for $d > 6$ this is easy; for $d=4$ it requires some work).

\section{Preliminaries: Group actions on 3-way arrays} 
\label{sec:problems}
The formulas for most natural group actions on 3-way arrays are somewhat unwieldy; our experience suggests that they are more easily digested when presented in the context of some of the natural interpretations of 3-way arrays as mathematical objects. To connect the interpretations with the formulas themselves, one technical tool is very useful, namely, given a 3-way array $\tA(i,j,k)$, we define its \emph{frontal slices} to be the matrices $A_k$ defined by $A_k(i,j) := \tA(i,j,k)$; that is, we think of the box of $\tA$ as arranged so that the $i$ and $j$ axes lie in the page, while the $k$-axis is perpendicular to the page. Similarly, its \emph{lateral slices} (viewing the 3D box of $\tA$ ``from the side'') are defined by $L_j(i,k) := \tA(i,j,k)$. An $\ell \times n \times m$ 3-way array thus has $m$ frontal slices and $n$ lateral slices.

A natural action on arrays of size $\ell \times n \times m$ is that of $\GL(\ell, \F) \times \GL(n,\F) \times \GL(m,\F)$ by change of 
basis in each of the 3 directions, namely $((P,Q,R) \cdot \tA)(i',j',k') = 
\sum_{i,j,k} \tA(i,j,k) P_{ii'} Q_{jj'} R_{kk'}$. We will see several 
interpretations of this action below.

\paragraph{3-tensors.} A 3-way array $\tA(i,j,k)$, 
where $i\in[\ell]$, 
$j\in[n]$, and $k\in[m]$, is naturally identified as a vector in 
$\F^\ell\otimes\F^n\otimes\F^m$. Letting $\vec{e_i}$ denote the $i$th 
standard basis vector of $\F^n$,  a standard basis of 
$\F^\ell\otimes\F^n\otimes\F^m$ is 
$\{\vec{e_i}\otimes\vec{e_j}\otimes\vec{e_k}\}$. Then $\tA$ represents the vector 
$\sum_{i,j,k}\tA(i,j,k)\vec{e_i}\otimes\vec{e_j}\otimes\vec{e_j}$ in 
$\F^\ell\otimes\F^n\otimes\F^m$. The natural action by 
$\GL(\ell, \F)\times\GL(n, \F)\times\GL(m, \F)$ above corresponds to changes of 
basis of the three vector spaces in the tensor product. The problem of deciding 
whether 
two 3-way arrays are the same under this action is called 
\ThreeTIlong.\footnote{Some authors call this \algprobm{Tensor Equivalence}; we use 
``\algprobm{Isomorphism}'' both because this is the natural notion of isomorphism 
for such objects, and because we will be considering many different equivalence 
relations on essentially the same underlying objects.}

\paragraph{Matrix spaces.} 
Given a 3-way array $\tA$, it is natural to consider the linear span of its frontal slices, $\cA = \langle A_1, \dotsc, A_m \rangle$, also called a \emph{matrix space}. 
One convenience of this viewpoint is that the action of $\GL(m,\F)$ becomes implicit: it corresponds to 
change of basis \emph{within} the matrix space $\cA$. 
This allows us to generalize the three natural equivalence relations on matrices to matrix 
spaces: (1) two $\ell \times n$ matrix spaces $\cA$ 
and $\cB$ are \emph{equivalent} if there exists $(P, Q) \in \GL(\ell, \F) \times 
\GL(n, \F)$ such 
that $P\cA Q = \cB$, where $P\cA Q := \{PA Q : A \in \cA\}$; (2) two 
$n \times n$ matrix spaces $\cA, \cB$ are \emph{conjugate} if there exists $P \in 
\GL(n, \F)$ such that $P \cA P^{-1} = \cB$; and (3) they are \emph{isometric} if 
$P \cA 
P^t = \cB$. The corresponding decision problems, when $\cA$ is 
given by a basis $A_1, \dotsc, A_d$, are \MatSpEquivlong, \MatSpConjlong, 
and \MatSpIsomlong, respectively. 

\hyperdef{}{sec:problems:nilpotent}{}
\paragraph{Nilpotent groups.} If $A,B$ are two subsets of a group $G$, then 
$[A,B]$ 
denotes the sub\emph{group} generated by all elements of the form $[a,b] = 
aba^{-1}b^{-1}$, for $a \in A, b \in B$. The \emph{lower central series} of a 
group $G$ is defined as follows: $\gamma_1(G) = G$, $\gamma_{k+1}(G) = 
[\gamma_k(G), G]$. A group is \emph{nilpotent} if there is some $c$ such that 
$\gamma_{c+1}(G) = 1$; the smallest such $c$ is called the \emph{nilpotency class} 
of $G$, or sometimes just ``class'' when it is understood from context. A finite 
group is nilpotent if and only if it is the product of its Sylow subgroups; in 
particular, all groups of prime power order are nilpotent.

\paragraph{Bilinear maps, finite groups, and systems of polynomials.} 
While the 
matrix 
space viewpoint has the merit of 
drawing an analogy with the more familiar object of matrices, other interpretations 
lead to standard complexity problems that may be more familiar to some readers. 
For example, from an $\ell \times n \times m$ 3-way array $\tA$, we can construct 
a bilinear map (=system of $m$ bilinear forms) $f_\tA:\F^\ell\times\F^n\to\F^m$, 
sending $(u, v)\in \F^\ell\times 
\F^n$ to $(u^t A_1 v, \dots, u^tA_m v)^t$, where the $A_k$ are the frontal slices of $\tA$.\footnote{In this paper elements in $\F^n$ are column vectors.} 
The group action defining \MatSpEquivlong 
is equivalent to the action 
of $\GL(\ell, \F)\times\GL(n, \F)\times \GL(m, \F)$ on such bilinear maps.  

When $\ell=n$, the action in \MatSpIsomlong is equivalent to the natural action of 
$\GL(n, \F)\times \GL(m, \F)$ on such bilinear maps.
Two bilinear maps that are essentially the same up to such basis changes 
are sometimes called pseudo-isometric \cite{BW12}.

Bilinear maps of the form $V\times V\to W$ turn out to arise naturally in group theory and 
algebraic geometry. When $A_k$ are skew-symmetric over $\F_p$, $p$ an 
odd prime, Baer's correspondence \cite{Bae38} gives a bijection between finite 
$p$-groups 
of class 2 
and exponent $p$, that is, in which $g^p = 1$ for all $g$ and in which $[G, G] 
\leq Z(G)$, and their corresponding bilinear maps $G/Z(G) \times G/Z(G) \to 
[G,G]$, given by $(gZ(G), hZ(G)) \mapsto [g,h]=ghg^{-1}h^{-1}$. Two such groups 
are isomorphic if and only if their corresponding bilinear maps are 
pseudo-isometric, if and only if, using the matrix space terminology, the 
matrix spaces they span are isometric. When $A_k$ are symmetric, by the classical 
correspondences between symmetric matrices and homogeneous quadratic forms, a 
symmetric bilinear map naturally yields a quadratic map from $\F^n$ to $\F^m$. The 
two quadratic maps are isomorphic if and only if the corresponding bilinear 
maps are pseudo-isometric.

\paragraph{Cubic forms \& trilinear forms.}
From a 3-way array 
$\tA$ 
we  can also construct a cubic form (=homogeneous degree 3 polynomial) $\sum_{i,j,k} 
\tA(i,j,k) x_i x_j x_k$, where $x_i$ are formal variables. 
If we consider the variables as commuting---or, equivalently, if $\tA$ is 
symmetric, meaning it is unchanged by permuting its three indices---we get an 
ordinary cubic form; if we consider them as non-commuting, we get a trilinear form 
(or ``non-commutative cubic form''). In either case, 
the natural notion of isomorphism here comes from the 
action of $\GL(n,\F)$ on the $n$ variables $x_i$, in which $P \in \GL(n,\F)$ transforms 
the preceding form into $\sum_{ijk} \tA(i,j,k) (\sum_{i'} P_{ii'} x_{i'})(\sum_{j'} 
P_{jj'} x_{j'})(\sum_{k'} P_{kk'} x_{k'})$. In terms of 3-way arrays, we get $(P 
\cdot \tA)(i', j', k') = \sum_{ijk} \tA(i,j,k) P_{ii'} P_{jj'} P_{kk'}$. The 
corresponding isomorphism 
problems are called \CubicFormlong (in the commutative case) and \NcCubicFormlong.

\paragraph{Algebras.} We may also consider a 3-way array 
$\tA(i,j,k)$, $i, j, 
k\in[n]$, as the structure 
constants of an algebra (which need not be associative, commutative, nor unital), 
say with basis $x_1, \dotsc, x_n$, and with multiplication given by $x_i \cdot x_j 
= \sum_k \tA(i,j,k) x_k$, and then extended (bi)linearly. Here the natural notion 
equivalence comes from the action of $\GL(n,\F)$ by change of basis on the $x_i$. 
Despite the seeming similarity of this action to that on cubic forms, it turns out 
to be quite different: given $P \in \GL(n,\F)$, let $\vec{x}' = P\vec{x}$; then we 
have $x_i' \cdot x_j' = (\sum_{i} P_{i' i} x_i)\cdot (\sum_{j} P_{j' j} x_j) 
 = \sum_{i,j} P_{i' i} P_{j' j} x_i \cdot x_j$ 
 $= \sum_{i,j,k} P_{i' i} P_{j' j} \tA(i,j,k) x_k = \sum_{i,j,k} P_{i' i} P_{j' j} 
 \tA(i,j,k) \sum_{k'} (P^{-1})_{kk'} x_{k'}$.
Thus $\tA$ becomes $(P \cdot \tA)(i',j',k') = \sum_{ijk} \tA(i,j,k) P_{i' i} P_{j' j} 
(P^{-1})_{k k'}$. The inverse in the third factor here is the crucial difference 
between this case and that of cubic or trilinear forms above, 
similar to the difference between matrix conjugacy and matrix isometry. The 
corresponding isomorphism problem is called \AlgIsolong.

\paragraph{Summary.} 
The isomorphism problems of the above structures all have 3-way arrays as the 
underlying object, but are determined by different group actions. It is not hard to 
see that there are 
essentially five group actions in total: \ThreeTIlong, \MatSpConjlong, 
\MatSpIsomlong, \NcCubicFormlong, and \algprobm{Algebra} \algprobm{Isomorphism}.
It turns out that these cover all the natural isomorphism problems on 3-way arrays 
in which the group acting is a product of $\GL(n,\F)$ (where $n$ is the side 
length of the arrays); see 
\Sec{sec:prelim:tensor} for discussion.

\section{Main results}\label{sec:result}

\subsection{Equivalence of isomorphism problems for 3-way arrays} 

\begin{definition}[{$d\cc{TI}, \cc{TI}$}]
For any field $\F$, $d\cc{TI}_\F$ denotes the class of problems that are 
polynomial-time Turing (Cook) reducible to \DeeTIlong over 
$\F$.\footnote{We follow a natural convention: when $\F$ is finite, a fixed 
algebraic extension of a finite field such as $\overline{\F}_p$, the rationals, or 
a fixed algebraic extension of the rationals such as $\overline{\Q}$, we consider 
the usual model of Turing machines; when $\F$ is $\mathbb{R}$, $\mathbb{C}$, the 
$p$-adic rationals $\Q_p$, or other more ``exotic'' fields, we consider this in 
the Blum--Shub--Smale model over $\F$.} When we write $d\cc{TI}$ without 
mentioning the field, the result holds for any field. 
$\cc{TI}_{\F} = \bigcup_{d \geq 1} d\cc{TI}_{\F}$. 
\end{definition}

We now state our first main theorem. 
\begin{maintheorem}\label{thm:main}
\ThreeTIlong reduces to each of the following problems in polynomial time.

\begin{enumerate}
\item \GpIlong for $p$-groups exponent $p$ ($g^p=1$ for all $g$) and class 2 
($G/Z(G)$ is abelian) given by generating matrices over $\F_{p^e}$. Here we consider only $\cc{3TI}_{\F_{p^e}}$ where $p$ is 
an odd prime.

\item \label{thm:main:isom} 
\MatSpIsomlong, even for alternating or symmetric matrix spaces.

\item 
\MatSpConjlong, and even the special cases: 
\begin{enumerate}
\item \algprobm{Matrix Lie Algebra Conjugacy}, for solvable Lie algebras $L$ of 
derived length 2.\footnote{And even further,  where $L/ [L, L] \cong \F$.} 
\item \algprobm{Associative Matrix Algebra Conjugacy}.\footnote{Even for algebras 
$A$ whose Jacobson radical $J(A)$ squares to zero and $A/J(A) \cong \F$.} 
\end{enumerate}

\item 
\algprobm{Algebra Isomorphism}, and even the special cases:
\begin{enumerate}
\item \algprobm{Associative Algebra Isomorphism}, for algebras that are 
commutative and unital, and for algebras that are commutative and 3-nilpotent 
($abc=0$ for all $a,b,c, \in A$)
\item \algprobm{Lie Algebra Isomorphism}, for 2-step nilpotent Lie algebras ($[u,[v,w]] = 0$ $\forall u,v,w$) 
\end{enumerate}
\item 
\CubicFormlong and \NcCubicFormlong.
\end{enumerate}
The algebras in (3) are given by a set of matrices which linearly span the 
algebra, while in (4) they are given by structure constants (see ``Algebras'' in 
\Sec{sec:problems}).
\end{maintheorem}

\begin{remark}
Agrawal \& Saxena \cite[Thm.~5]{AS05} gave a reduction from \CubicFormlong over $\F$ to \algprobm{Ring Isomorphism} for commutative, unital, associative algebras over $\F$, when every element of $\F$ has a cube root. 
For finite fields $\F_q$, the 
only such fields are those for which $q=p^{2e+1}$ and $p\equiv 2 \pmod{3}$, which 
is asymptotically half of all primes. As explained after the proof of 
\cite[Thm.~5]{AS05}, the use of cube roots
seems inherent in their reduction.

Using our results in conjunction with \cite{FGS19}, we get a new reduction from 
\CubicFormlong to \algprobm{Ring Isomorphism} 
(for the same class of rings) which works over 
any field of characteristic 0 or $p > 3$. 
Note that our reduction is very different from the one in 
\cite{AS05}.
\end{remark}

Figure~\ref{fig:main} below summarizes where the various parts of \Thm{thm:main} are proven.

We then resolve an open question well-known to the experts:\footnote{We asked several experts who knew of the question, but we were unable to find a written reference. Interestingly, Oldenburger \cite{oldenburger} worked on what we would call \DeeTIlong as far back as the 1930s. We would be grateful for any prior written reference to the question of whether \DeeTI reduces to \ThreeTI.}

\setcounter{maincorollary}{\arabic{maintheorem}}
\begin{maintheorem} \label{thm:d_to_3}
\DeeTIlong reduces to \AlgIsolong.
\end{maintheorem}

Since the main result of \cite{FGS19} reduces the problems in Theorem~\ref{thm:main} to 
\ThreeTIlong (cf. \cite[Rmk.~1.1]{FGS19}), we have:

\begin{maincorollary} \label{cor:main}
Each of the problems listed in Theorem~\ref{thm:main} is $\cc{TI}$-complete.\footnote{For \CubicFormlong, we only show that it is in $\cc{TI}_\F$ when $\chr 
 \F > 3$ or $\chr \F = 0$.}
In particular, $d\TI$ and $\ThreeTI$ are equivalent.
\end{maincorollary}

\begin{remark}
This phenomenon is reminiscent of the transition in hardness from 2 to 3 in 
$k$-\algprobm{SAT}, $k$-\algprobm{Coloring}, $k$-\algprobm{Matching}, and many 
other $\cc{NP}$-complete problems. It is interesting that an analogous 
phenomenon---a transition to some sort of ``universality'' from 2 to 3---occurs in 
the setting of isomorphism problems, which we believe are not $\cc{NP}$-complete 
 over finite fields.
\end{remark}

\newif\iffgs
\fgsfalse

\begin{figure}[!htbp]
\[
\hspace{-0.35in}
\xymatrix{
 & & \parbox{1in}{\centering $d$\algprobm{-Tensor Iso.} \\ $U_1 \otimes \dotsb \otimes U_d$} \ar@/^1pc/[dddr]^(0.7){\text{\Thm{thm:d_to_3}}}\\
 & & \ar@/_/[dll]_{\parbox{0.5in}{\scriptsize \text{\Prop{prop:3-tensor_isometry},} \\ \text{\Cor{cor:pgp}}}} \parbox{1in}{\centering \algprobm{3-Tensor Iso.} \\ $U \otimes V \otimes W$}\ar@/^/[drr]^{\text{Prop.~\ref{prop:3-tensor_conjugacy}}} & & \\
\parbox{1.3in}{\centering \algprobm{Bilinear Map Iso.} \\ p-\algprobm{Group Iso.} \\ $V \otimes V \otimes W$} \iffgs \ar@/_/[urr]_{\text{\cite{FGS19}}}\fi \ar[dr]_{\text{Prop.~\ref{prop:isometry_algebra}}} \ar@/^/[drrr]^(.65){\text{Prop.~\ref{prop:isometry_algebra}}} \ar@/^1.4pc/[ddrrr]^(.65){\text{\Cor{cor:pseudo_special}}} & & & &  \iffgs \ar@/^/[ull]^{\text{\cite{FGS19}}}\fi \parbox{1in}{\centering \algprobm{Matrix Space} \\ \algprobm{Conjugacy} \\ $V \otimes V^* \otimes W$} \\
&  \parbox{1in}{\centering \algprobm{Trilinear} \\ \algprobm{Form} \algprobm{Equiv.} \\ $V \otimes V \otimes V$} \iffgs \ar[uur]_(.65){\text{\cite{FGS19}}}\fi & & \iffgs \ar[uul]_(.35){\text{\cite{FGS19}}}\fi \parbox{1in}{\centering \algprobm{Algebra Iso.} \\ $V \otimes V \otimes V^*$} &  \\
 & \parbox{1in}{\centering \algprobm{Cubic Form} \ar[u]^(.45){\text{Special case, when $6$ is a unit}} \\ \algprobm{Equiv.}} & & \ar[ll]^{\text{\cite{AS05, AS06}}} \ar[u]_{\text{Special case}} \parbox{1in}{\centering \algprobm{Commutative} \\ \algprobm{Algebra Iso.}}
 }
\]

\caption{\label{fig:main} Reductions for Thm.~\ref{thm:main}. An arrow $A \to B$ indicates that $A$ reduces to $B$, i.\,e., $A \leq_m^p B$. For \Cor{cor:main}, the five tensor problems in the center circle all reduce to \ThreeTI via \cite{FGS19}. For the ``$V \otimes V \otimes W$'' notation, see \Sec{sec:prelim:tensor}. } 
\end{figure}
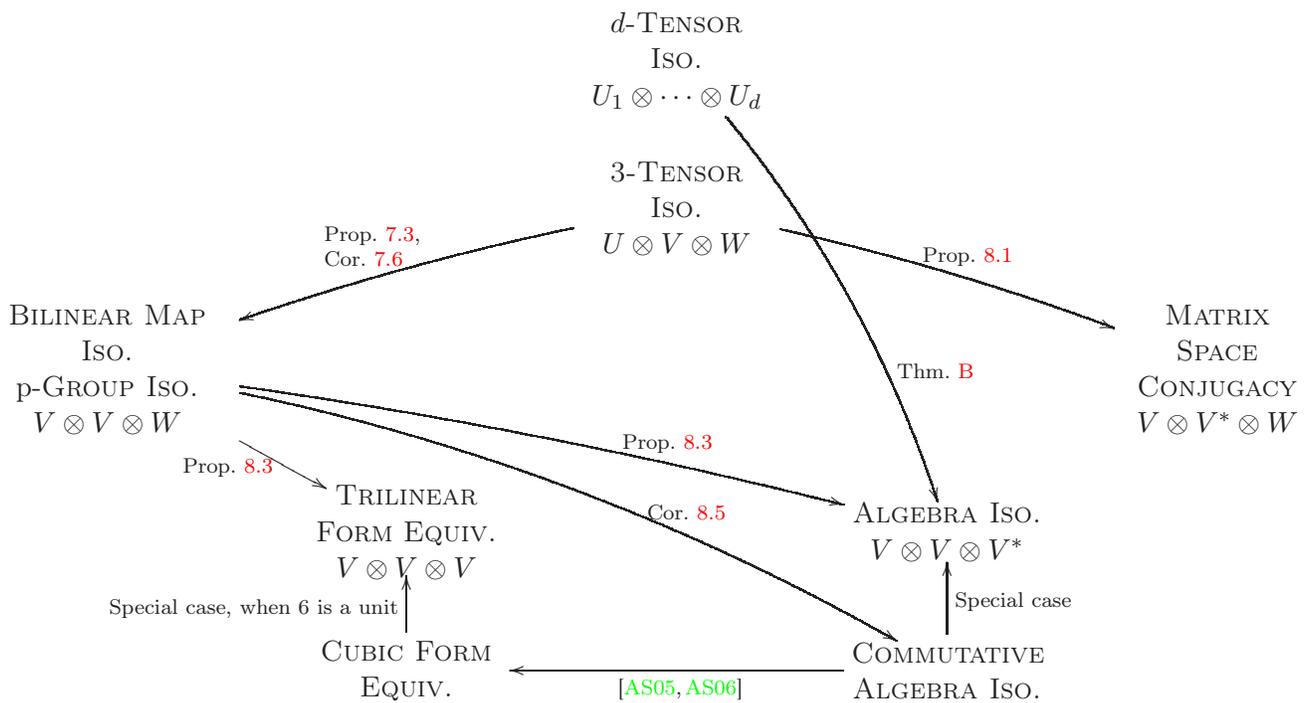

\begin{remark}
Here is a brief summary of what is known about the complexity of some of these 
problems.
Over a finite field $\F_q$, these problems are in $\cc{NP}\cap \cc{coAM}$. For $\ell \times n \times m$ 3-way arrays, the brute-force algorithms run in time $q^{O(\ell^2 + n^2 + m^2)}$, 
as $\GL(n,\F_q)$ can be enumerated in time $q^{\Theta(n^2)}$. Note 
that polynomial-time in the input size here would be $\poly(\ell, n, m, \log q)$.
Over any field $\F$, these problems are in $\cc{NP}_{\F}$ in the Blum--Shub--Smale model. 
When the input arrays are over $\Q$ and we ask for isomorphism over $\C$ or 
$\mathbb{R}$, these problems are in $\cc{PSPACE}$ using quantifier elimination. By 
Koiran's celebrated result on Hilbert's Nullstellensatz, for equivalence over $\C$ 
they are in $\cc{AM}$ assuming the Generalized Riemann Hypothesis \cite{Koi96}. 
When the input is over $\Q$ and we ask for equivalence over $\Q$, it is unknown 
whether these problems are even decidable; classically this is studied under 
\AlgIsolong for associative, unital algebras over 
$\Q$ (see, e.\,g., \cite{AS06, Poonen}), but by 
\Cor{cor:main}, the question of decidability is open for all of these problems.

Over finite fields,
several of these problems can be solved efficiently when one of the side lengths 
of the array is small. For $d$-dimensional spaces of $n \times n$ matrices, \MatSpConjlong and \algprobm{Isometry}
can be solved in 
$q^{O(n^2)}\cdot \poly(d,n,\log q)$ time: once we fix an element of 
$\GL(n,\F_q)$, the isomorphism problem reduces to solving linear systems of equations.
Less trivially, \MatSpConjlongWords 
can be solved in time 
$q^{O(d^2)}\cdot \poly(d,n,\log q)$ and \ThreeTI for $n \times m \times d$ tensors 
in time $q^{O(d^2)}\cdot \poly(d,n,m,\log q)$, since once we fix an element of 
$\GL(d,\F_q)$, the isomorphism problem either becomes an instance of, or reduces to 
\cite{IQ17},  \algprobm{Module Isomorphism}, which admits several polynomial-time 
algorithms \cite{BL08, CIK97, IKS10, Sergeichuk2000}. Finally, one can solve 
\MatSpIsomlongWords in time $q^{O(d^2)}\cdot 
\poly(d,n, \log q)$: once one fixes an element 
of 
$\GL(d,\F_q)$, there is a rather involved algorithm  
\cite{IQ17}, which uses the $*$-algebra technique originated from the study of 
computing with 
$p$-groups \cite{Wil09a,BW12}.
\end{remark}

\subsection{Relations with Graph Isomorphism and Code 
Equivalence}\label{sec:GI_code}
We observe then \GIlong and \CodeEqlong reduce to \ThreeTIlong. In particular, the 
class $\cc{TI}$ contains the classical graph isomorphism class $\cc{GI}$. 

Recall \CodeEqlong asks to decide whether two linear codes are the 
same up to a linear transformation preserving the Hamming weights of codes. Here 
the linear codes are just subspaces of $\F_q^n$ of dimension $d$, represented by 
linear bases. Linear transformations preserving the Hamming weights include 
permutations and monomial transformations. Recall that the latter consists of matrices 
where every row and every column has exactly one non-zero entry. 
Indeed, over many fields this is without loss of generality, as Hamming-weight-preserving 
linear maps are always induced by monomial transformations (first proved over 
finite fields \cite{MacWilliams}, and more recently over much more general 
algebraic objects, e.\,g., \cite{GNW}). 
\CodeEq has long 
been studied in the coding theory community; see e.g. \cite{PR97,Sen00}.

For \CodeEqlong, we observe that previous results already combine to give:
\begin{observation}\label{obs:code_3TI}
\CodeEqlong (under permutations) reduces to \ThreeTIlong.
\end{observation}

\begin{proof}
 \CodeEqlong reduces to \MatLieConjlong \cite{GrochowLie}, a special case of 
 \MatSpConjlong, which in turn reduces to \ThreeTI \cite{FGS19}.
\end{proof}

Using the linear-algebraic coloring gadget, we can extend this to equivalence of 
codes under monomial transformations (see \Sec{sec:technique}). Given two 
$d\times n$ matrices $A, B$ over $\F$ of rank $d$, the \MonCodeEqlong
problem is to decide whether there exist $Q\in \GL(d, \F)$ and a monomial 
matrix $P\in \Mon(n, \F)\leq \GL(n, \F)$ 
(product of a diagonal matrix and a permutation matrix) such that $QAP =B$. 

\begin{proposition} \label{prop:MonCodeEq}
\MonCodeEqlong reduces to \ThreeTIlong.
\end{proposition}

Since \GIlong reduces to \CodeEqlong \cite{Luks} (see \cite{miyazakiCodeEq}) and \cite{PR97} (even over arbitrary fields \cite{GrochowLie}), by 
\Obs{obs:code_3TI} and \Thm{thm:main}, we have the following.
\begin{corollary}
\GIlong reduces to \AltMatSpIsomlong.
\end{corollary}

Using our linear-algebraic gadgets, we also
reprove this result using a much more 
direct reduction (see \Prop{prop:GI}). Besides being a different 
construction, another reason for the additional proof is that the technique leads to the 
search-to-decision reduction, which we discuss below.

\subsection{Application to \GpIlong: reducing the nilpotency class}
For several reasons, the hardest cases of \GpIlong are believed to be $p$-groups 
of class 2 and exponent $p$; recall that these are groups in which every element 
has order $p$, the order of the group is $p^n$, and $G/Z(G)$ is abelian. See 
\hyperref{}{}{sec:problems:nilpotent}{Nilpotent groups} above. While this belief  
has been widely held for many decades, we are not aware of any prior reduction 
from a more general 
class of groups to this class. However, by combining our results with the Lazard 
correspondence, we immediately get such a reduction.

\hyperdef{}{cor:p}{}
\begin{corp}
Let $p$ be an odd prime. 
For groups generated by $m$ matrices of size $n \times n$, \GpIlong for $p$-groups 
of exponent $p$ and class $c < p$ reduces to \GpIlong for $p$-groups of exponent 
$p$ and class $2$ in time $\poly(n, m, \log p)$.
\end{corp}

\begin{proof}
By the Lazard correspondence 
(reproduced as \Thm{thm:lazard} below) 
two $p$-groups of exponent $p$ and class $c < p$ are isomorphic if and only if 
their corresponding $\F_p$-Lie algebras are. 
By \Prop{prop:lazard_matrices}, we can construct a generating set for the 
corresponding Lie algebra by applying the power series for logarithm to the 
generating matrices of $G$. This Lie algebra is thus a subalgebra of $n \times n$ 
matrices, so we can generate the entire Lie algebra (using the linear-algebra 
version of breadth-first search; its dimension is $\leq n^2$) and compute its 
structure constants in time polynomial in $n$, $m$, and $\log p$. 
Then use \cite{FGS19} to reduce isomorphism of Lie 
algebras to \TI, and then apply \Thm{thm:main} (specifically, \Cor{cor:pgp}) to reduce to isomorphism of 
$p$-groups of exponent $p$ and class $2$ given by a matrix generating set.
\end{proof}

The only obstacle to getting this proof to work in the Cayley table model is that 
our reduction from \TI to \AltMatSpIsomlong (\Prop{prop:3-tensor_isometry}) blows 
up the dimension quadratically, which means the size of the group becomes 
$|G|^{O(\log |G|)}$ after the reduction. See 
Question~\ref{question:search_decision}.

\subsection{Search to decision reductions} 
Reducing search problems to their associated decision problems 
is a classical and intriguing topic in complexity 
theory. Aside from the now-standard search-to-decision reduction for SAT, one of the earliest results in this direction was by Valiant in the 1970's \cite{valiant}. A celebrated 
result of Bellare and Goldwasser shows that, assuming $\cc{EE}\neq\cc{NEE}$, there exists a language in $\cc{NP}$ for 
which search does not reduce to decision under polynomial-time reductions 
\cite{BG94}. 
However, as usual for such statements based on complexity-theoretic assumptions, 
the problems constructed by such a proof are considered somewhat unnatural. For 
natural problems, on the one hand, there are search-to-decision reductions for 
$\cc{NP}$-complete problems and for \GI. On the other hand, such is not 
known, nor expected to be the case, for Nash Equilibrium  \cite{CDT09} (for which decision is 
trivial). 

Reducing search to decision is particularly intriguing for testing isomorphism 
of groups. One difficulty is that it is not clear how to 
guess a partial solution, and then make progress by restricting to a 
subgroup. In general, testing isomorphism 
of certain algebraic structures (groups, algebras, etc.) forms a large family of 
problems 
for which search-to-decision reductions are not known.

Because of the close relationship between \ThreeTI and isomorphism of various 
algebraic structures, one might expect similar difficulties in reducing search to 
decision for \ThreeTI, 
and thus for $\cc{TI}$-complete problems as well. 
Nonetheless, for \AltMatSpIsomlong, we are able to use the linear-algebraic 
coloring gadgets to get a non-trivial search-to-decision reduction.

\begin{maintheorem}\label{thm:search_decision}
There is a search-to-decision reduction for \AltMatSpIsomlong which, given 
$n \times n$ alternating matrix spaces $\cA, \cB$ over $\F_q$,  computes an isometry between them if they are 
isometric, in time $q^{\tilde{O}(n)}$. The reduction queries the decision oracle 
with inputs of dimension at most $O(n^2)$.
\end{maintheorem}

As a consequence, a $q^{\tilde{O}(\sqrt{n})}$-time decision algorithm would result 
in a $q^{\tilde{O}(n)}$-time search algorithm, in contrast with the brute-force 
$q^{O(n^2)}$ running time. Note that in this context, the size of the input is 
$\poly(n,\log q)$, so a $q^{\tilde{O}(\sqrt{n})}$ running time is still quite 
generous.

By the connection between \AltMatSpIsomlong and \GpIlong for $p$-groups of class $2$ and exponent $p$, we have 
the following. 
Note that the natural succinct input representation mentioned in the following 
result can have size $\poly(\ell, \log p) = \poly(\log |G|)$. 

\begin{maincorollary} \label{cor:search_decision}
Let $p$ be an odd prime, and let \algprobm{GpIso2Exp($p$)} denote the 
isomorphism problem for $p$-groups of class 2 and exponent $p$ in the model of 
matrix groups over $\F_p$.
For groups of order $p^\ell$, there is a search-to-decision reduction 
for \algprobm{GpIso2Exp($p$)} running in time $|G|^{O(\log 
\log |G|)}=p^{\tilde{O}(\ell)}$.
\end{maincorollary}

\section{Related work} \label{sec:related}
The most closely related work is that of Futorny, Grochow, and Sergeichuk 
\cite{FGS19}. They show that a large family of 
isomorphism problems on 3-way 
arrays---including those involving multiple 3-way arrays simultaneously, or 3-way 
arrays that are partitioned into blocks, or 3-way arrays where some of the blocks 
or sides are acted on by the same group (e.\,g., \MatSpIsomlong)---all reduce to 
\ThreeTI. Our work complements theirs in that all our reductions 
for \Thm{thm:main}
go in the 
opposite direction, reducing \ThreeTI to other problems. Some of our other results 
relate \GI and \CodeEqlong to \ThreeTI; the latter problems were not considered in 
\cite{FGS19}. 
\Thm{thm:d_to_3} considers $d$-tensors for any $d \geq 3$, which were not 
considered in \cite{FGS19}.

In \cite{AS05,AS06}, Agrawal and Saxena considered \algprobm{Cubic Form 
Equivalence} and testing isomorphism of commutative, associative, unital algebras.
They showed that \GI reduces to \AlgIsolong; \algprobm{Commutative Algebra Isomorphism} reduces to 
\CubicFormlong;
and \algprobm{Homogeneous} \algprobm{Degree-$d$} \algprobm{Form} \algprobm{Equivalence} reduces to 
\AlgIsolong
assuming that the underlying field has $d$th root for every field element. 
By combining a reduction from \cite{FGS19}, 
\Prop{prop:3-tensor_isometry}, and \Cor{cor:pseudo_special}, we get a new 
reduction from \CubicFormlong to \algprobm{Algebra Isomorphism} that works over 
any field in which $3!$ is a unit, which is fields of characteristic $0$ or $p > 3$.

There are several other works which consider related isomorphism problems. 
Grigorev \cite{Grigoriev83} showed that \GI is equivalent to 
isomorphism of unital, associative algebras $A$ such that the radical $R(A)$ 
squares to zero and $A/R(A)$ is abelian. Interestingly, we show 
$\cc{TI}$-completeness for conjugacy of 
matrix algebras with the same abstract structure (even when $A/R(A)$ is only 
1-dimensional). Note the latter problem is equivalent to asking whether two 
representations of $A$ are equivalent up to automorphisms of $A$. In the proof of \Thm{thm:d_to_3}, which uses algebras in which $R(A)^d=0$ when reducing from \DeeTI, we use Grigoriev's result.

Brooksbank and Wilson \cite{BW15} showed a reduction from \algprobm{Associative 
Algebra Isomorphism} (when given by structure constants) to \algprobm{Matrix 
Algebra Conjugacy}. 
Grochow \cite{GrochowLie}, among other things, showed that \GI and \CodeEq reduce to \MatLieConjlong, which is a special case of \MatSpConjlong.

In \cite{KS06}, Kayal and Saxena considered testing isomorphism of finite rings 
when the rings are given by structure constants. This problem 
generalizes testing isomorphism of algebras over finite fields. They 
put this problem in $\cc{NP} \cap \cc{coAM}$ \cite[Thm.~4.1]{KS06}, reduce \GI to 
this problem \cite[Thm.~4.4]{KS06}, and 
prove that counting the number of ring automorphism (\#RA) is in 
$\cc{FP}^{\cc{AM} 
\cap \cc{coAM}}$ \cite[Thm.~5.1]{KS06}. They also present a $\cc{ZPP}$ reduction 
from \GI to \#RA, 
and show that the decision version of the ring automorphism problem is in $\cc{P}$.

To summarize this zoo of isomorphism problems and reductions, we include Figure~\ref{fig:summary} for reference.

\begin{figure}[!htbp]
\newcommand{\fieldfootnote}{\hyperref{}{}{fn:field}{\textsuperscript{*}}}
\newcommand{\fieldfootnoteother}{\hyperref{}{}{fn:field2}{\textsuperscript{$\dagger$}}}
\newcommand{\ronyaifootnote}{\hyperref{}{}{fn:ronyai}{\textsuperscript{$\ddagger$}}}
\[
\hspace{-0.5in}
\xymatrix{
\text{\parbox{1.1in}{\centering \algprobm{Symmetric $d$-Tensor Diagonal Iso.}}} \ar@/^/[d] &  \text{\parbox{1.1in}{\centering \algprobm{Matrix $p$-Group Iso.} (class 2, exp. $p$)}}  \ar[d]_{\text{\cite{Bae38}}} & & & \text{$d$-\algprobm{Tensor Iso.}} \ar@/^/[dddl]^(0.3){\text{\Thm{thm:d_to_3}}} 
\\
\text{\parbox{1.1in}{\centering \algprobm{Degree-$d$ Form Eq.}}} \ar[rd]^(0.35){\text{\cite{AS05}\fieldfootnote}} \ar@/^/[u]^{\fieldfootnoteother}  &\text{\parbox{1in}{\centering \algprobm{Alt. Mat. Space Isom.}}} \ar[u] \ar@/^/[r]^{\text{\cite{FGS19}}} \ar[d]^(0.6){\text{\Prop{prop:isometry_algebra}}} & \ar@/^/[l]^{\text{\Thm{thm:main}}} \text{\parbox{1in}{3-\algprobm{Tensor Iso.}}}  \ar@/^1.25pc/[urr]^{\text{\Obs{obs:d}}} \ar@/_/[r]_{\text{\Thm{thm:main}}}& \ar@/_/[l]_{\text{\cite{FGS19}}} \text{\parbox{1in}{\centering\algprobm{Mat. Space Conj.}}} \\
\text{\parbox{1in}{\centering \algprobm{Cubic Form Eq.}}} \ar[u]^{\text{\Prop{prop:cubic_to_d}}} \ar@/^/[r]^{\text{\parbox{1.0in}{\centering \Thm{thm:main}\fieldfootnoteother}}} \ar[ru]^(0.3){\text{\cite{FGS19}\fieldfootnoteother}} & \text{\parbox{1in}{\centering \algprobm{Ring Iso.} (basis)}} \ar@/^2pc/[ld] \ar[drr] \ar@/^/[l]^{\text{\cite{AS05,AS06}}} & & \text{\parbox{1in}{\centering \algprobm{Mat. Assoc. Alg. Conj.}}}\ar[u]  &\text{\parbox{1in}{\centering\algprobm{Mat. Lie Alg. Conj.}}} \ar[ul] \\
\text{\parbox{1in}{\centering \algprobm{Ring Iso.} (gens/rels)}} & & & \text{\parbox{1in}{\centering\algprobm{Unital Assoc. Alg. Iso.}}} \ar[u]^{\text{\cite{BW15}}} & \text{\parbox{1in}{\centering\algprobm{Comm. Mat. Lie Alg. Conj.}}} \ar[u] \\
& & \text{\parbox{1in}{\centering \algprobm{Mon. Code Eq.}}} \ar[uuu]_(0.71){\text{\Prop{prop:MonCodeEq}}} & \ar[ld] \text{\parbox{1in}{\centering\algprobm{Semisimple Mat. Lie Alg. Conj.}}} \ar[ruu] &  \ar[dl] \text{\parbox{1in}{\centering\algprobm{Diag. Mat. Lie Alg. Conj.}}} \ar[u]\\
\text{\parbox{1in}{\centering \algprobm{Factoring Integers}}} \ar[uu]_{\text{\cite{AS05}}} \ar@/^1pc/[rrruu]^(0.3){\text{\parbox{0.55in}{\centering \cite{Ronyai88} (over $\Q$)\ronyaifootnote}}}
& \text{\parbox{1in}{\centering \algprobm{String Isomorphism}}} & \ar[l]^(0.35){\text{\parbox{0.7in}{\centering \cite{Luks82}, cf. \cite{Luks}}}}\text{\GI} \ar[luuu]^(0.6){\text{\cite{AS05,KS06}}} \ar@/_2.5pc/[luuuu]_(0.7){\text{\Prop{prop:GI}}} \ar[u]_(0.6){\text{cf. \cite{GrochowLie}}} \ar[r]_(0.35){\text{\cite{PR97,Luks}}} \ar[ru]_(0.6){\text{\cite{GrochowLie}}}  & \text{\parbox{1in}{\centering\algprobm{Perm. Code Eq.}}} \ar[ru]_(0.6){\text{\cite{GrochowLie}}} \ar[r]_(0.4){\text{\cite{BCGQ}}}& \text{\parbox{1in}{\centering \algprobm{Perm. Group Conj.}}} \\
\text{\parbox{1in}{\centering \algprobm{Alt. Mat. Space Isom.} ($\F_{p^e}$, verbose)}} \ar[r]^{\text{\cite{Bae38}}} & \text{\parbox{1.1in}{\centering \algprobm{$p$-Group Iso.} (class 2, exp. $p$, table)}} \ar[l] \ar[r] & \text{\parbox{1in}{\centering \algprobm{Group Iso.} (table)}} \ar[u]_{\text{(Classical, cf. \cite{ZKT})}} & 
\ar@{--}(18,10);(165,10)^(0.9){\text{\normalsize $\cc{TI}$-complete}} 
\ar@{--}(18,10);(18,-63) 
\ar@{--}(18,-63);(165,-63) 
\ar@{--}(165,10);(165,-63) 
\ar@{--}(18,10);(-15,10)_{\text{\normalsize $\cc{TI}$-complete\fieldfootnote\fieldfootnoteother}} 
\ar@{--}(-15,10);(-15,-44) 
\ar@{--}(-15,-44);(18,-44) 
 }
\]

\caption[Summary of isomorphism problems around \GIlong and 
\ThreeTIlong.]{\label{fig:summary} 
Summary of isomorphism problems around \GIlong and 
\TIlong. $A \to B$ indicates that $A$ reduces to $B$, i.\,e., $A \leq_m^p B$. Unattributed arrows indicate $A$ is clearly a special case of $B$. Note that the 
definition of ring used in \cite{AS05} is commutative, finite, and unital; by 
``algebra'' we mean an algebra (not necessarily associative, let alone commutative 
nor unital) over a field. The reductions between \algprobm{Ring Iso.} (in the 
basis representation) and \algprobm{Degree-$d$ Form Eq.} and \algprobm{Unital 
Associative Algebra Isomorphism} are for rings over a field. The equivalences 
between \AltMatSpIsomlong and $p$-\GpIlong are for matrix spaces over $\F_{p^e}$. 
Some 
\TI-complete problems from \Thm{thm:main} are left 
out for clarity.

\hrulefill

\hyperdef{}{fn:field}{{\color{red}\textsuperscript{*}}} These results only hold over fields where every element has a $d$th root. In particular, 
\algprobm{Degree $d$ Form Equivalence} and \algprobm{Symmetric $d$-Tensor 
Isomorphism} are \ThreeTI-complete over fields with $d$-th roots. 
A finite field $\F_q$ has this property if and only if $d$ is coprime to 
$q-1$. 

\hyperdef{}{fn:filed2}{{\color{red}\textsuperscript{$\dagger$}}} These results only hold over rings where $d!$ is a unit.

\hyperdef{}{fn:ronyai}{{\color{red}\textsuperscript{$\ddagger$}}}Assuming the Generalized Riemann Hypothesis, R\'{o}nyai \cite{Ronyai88} shows a Las Vegas randomized polynomial-time reduction from factoring square-free integers---probably not much easier than the general case---to isomorphism of 4-dimensional algebras over $\Q$. Despite the additional hypotheses, this is notable as the target of the reduction is algebras of \emph{constant} dimension, in contrast to all other reductions in this figure.
 }
\end{figure}

\section{Overview of one new technique, and one full proof}\label{sec:technique}

In this section we describe one of the key new
techniques in this paper: a linear-algebraic 
coloring gadget. We exhibit this gadget by giving the full proof of 
\Prop{prop:MonCodeEq} as an example. A related gadget was used in 
\cite{FGS19} to show reductions \emph{to} \ThreeTI; our reductions all go in the 
opposite direction. Furthermore, whereas the gadgets used in \cite{FGS19} were 
primarily to ensure that two different blocks could not be mixed, our gadgets 
allow us to ensure that certain slices of a tensor can be permuted, while 
disallowing more general linear transformations. 

In the context of \GI, there are many ways to reduce \algprobm{Colored GI} to ordinary \GI; here we give one example, which will serve as an analogy for our linear-algebraic gadget. To individualize a vertex $v \in G$ (give it a unique color), attach to it a large ``star'': if $|V(G)|=n$, add $n+1$ new vertices to $G$ and attach them all to $v$; call the resulting graph $G_v$. This has the effect that any automorphism of $G_v$ must fix $v$, since $v$ has a degree strictly larger than any other vertex. Furthermore, if $H_w$ is obtained by a similar construction, then there is an isomorphism $G \to H$ which sends $v \mapsto w$ if and only if $G_v \cong H_w$. Finally, if we attach stars of size $n+1$ to multiple vertices $v_1, \dotsc, v_k$, then any automorphism of $G$ must permute the $v_i$ amongst themselves, and there is an isomorphism $G \to H$ sending $\{v_1, \dotsc, v_k\} \mapsto \{w_1, \dotsc, w_k \}$ if and only if the corresponding enlarged graphs are isomorphic.

We adapt this idea to the context of 3-way arrays. 
Let 
$\tA$ 
be an $\ell \times n \times m$ 3-way array, with lateral slices $L_1, L_2, \dotsc, 
L_n$ (each an $\ell \times m$ matrix). 
For any vector $v \in \F^n$, we get an associated lateral matrix $L_v$, which is a linear combination of the lateral slices as given, namely $L_v := \sum_{j=1}^n v_j L_j$ (note that when $v=\vec{e_j}$ is the $j$-th standard basis vector, the associated lateral matrix is indeed $L_j$). By analogy with adjacency matrices of graphs, $L_v$ is a natural analogue of the neighborhood of a vertex in a graph. Correspondingly, we get a notion of ``degree,'' which we may define as
\begin{eqnarray*}
\deg_\tA(v) & := & \rk L_v = \rk (\sum_{j=1}^n v_j L_j) 
=  \dim \linspan\{L_v \vec{w} : \vec{w} \in \F^m \} 
= \dim \linspan\{\vec{u}^t L_v : \vec{u} \in \F^\ell \}
\end{eqnarray*}
The last two characterizations are analogous to the fact that the degree of a vertex $v$ in a graph $G$ may be defined as the number of ``in-neighbors'' (nonzero entries the corresponding row of the adjacency matrix) or the number of ``out-neighbors'' (nonzero entries in the corresponding column). 

To ``individualize'' $v$, 
we can enlarge $\tA$ with a gadget to increase $\deg_\tA(v)$, as in the graph case. Note that $\deg_\tA(v) \leq \min\{\ell,m\}$ because the lateral matrices are all of size $\ell \times m$. For notational simplicity, let us individualize $v = \vec{e_1} = (1,0,\dotsc,0)^t$. To individualize $v$, we will increase its degree by $d = \min\{\ell,m\}+1 > \max_{v \in \F^n} \deg_\tA (v)$. Extend $\tA$ to a new 3-way array $\tA_v$ of size $(\ell+d) \times n \times (m+d)$; in the ``first'' $\ell \times n \times m$ ``corner'', we will have the original array $\tA$, and then we will append to it an identity matrix in one slice to increase $\deg(v)$. More specifically, the lateral slices of $\tA_v$ will be 
\[
L_1' = \begin{bmatrix} L_1 & 0 \\ 0 & I_d \end{bmatrix} \qquad \text{ and } \qquad L_j' = \begin{bmatrix} L_j & 0 \\ 0 & 0 \end{bmatrix} \quad ( \text{for } j > 1).
\]
Now we have that $\deg_{\tA_v}(v) \geq d$. This almost does what we want, but now note that any vector $w=(w_1,\dotsc,w_n)$ with $w_1 \neq 0$ has $\deg_{\tA_v}(w) =\rk (w_1 L_1' + \sum_{j \geq 2} w_j L_j) \geq d$. We can nonetheless consider this a sort of linear-algebraic individualization.

Leveraging this trick, we can then individualize an entire basis of $\F^n$ 
simultaneously, so that $d \leq \deg(v) < 2d$ for any vector $v$ in our basis, and 
$\deg(v') \geq 2d$ for any nonzero $v'$ outside the basis   
(not a scalar multiple of one of the basis vectors), as we do in the following 
proof of \Prop{prop:MonCodeEq}. 
This is 
also a 3-dimensional analogue of the reduction from \GI to \CodeEq \cite{Luks,miyazakiCodeEq,PR97} 
(where they use Hamming weight instead of rank).

\begin{proof}[Proof of \Prop{prop:MonCodeEq}]
Without loss of generality we assume $d>1$, as the problem is easily solvable when $d=1$. 
We treat a $d \times n$ matrix $A$ as a 3-way array of size $d \times n \times 1$, 
and then follow the outline proposed above, of individualizing the entire standard 
basis $\vec{e_1}, \dotsc, \vec{e_n}$. Since the third direction only has length 1, 
the maximum degree of any column is 1, so it suffices to use gadgets of rank 2. 
More specifically, 
we build a $(d + 2n) \times n \times (1 + 2n)$ 3-way array $\tA$ whose lateral slices are
\[
L_j = \left[\begin{array}{ccccccc}
a_{1,j} & \vzero_{1 \times 2} & \vzero_{1 \times 2} & \cdots & \vzero_{1 \times 2} & \cdots & \vzero_{1 \times 2} \\
\vdots & \vdots & \vdots & \ddots & \vdots & \ddots & \vdots\\
a_{d,j} & \vzero_{1 \times 2} & \vzero_{1 \times 2} & \cdots & \vzero_{1 \times 2}  & \cdots & \vzero_{1 \times 2} \\
\vzero_{2 \times 1} & \vzero_{2 \times 2} & \vzero_{2 \times 2} & \cdots & \vzero_{2 \times 2}  & \cdots & \vzero_{2 \times 2} \\
\vdots & \vdots & \vdots & \ddots & \vdots & \ddots & \vdots \\
\vzero_{2 \times 1} & \vzero_{2 \times 2} & \vzero_{2 \times 2} & \cdots & I_2 & \cdots & \vzero_{2 \times 2} \\
\vdots & \vdots & \vdots & \ddots & \vdots & \ddots & \vdots \\
\vzero_{2 \times 1} & \vzero_{2 \times 2} & \vzero_{2 \times 2} & \cdots & \vzero_{2 \times 2}  & \cdots & \vzero_{2 \times 2} \\
\end{array}\right]
\]
where the $I_2$ block is in the $j$-th block of size 2 (that is, rows $d + 2(j-1) + \{1,2\}$ and columns $2(j-1) + \{1,2\}$).
It will also be useful to visualize the frontal slices of $\tA$, as follows. Here each entry of the ``matrix'' below is actually a $(1+2n)$-dimensional vector, ``coming out of the page'':
$$
\tA=\begin{bmatrix}
\tilde a_{1,1} & \tilde a_{1,2} & \dots & \tilde a_{1,n} \\
\vdots & \vdots & \ddots & \vdots \\
\tilde a_{d,1} & \tilde a_{d,2} & \dots & \tilde a_{d,n}\\
e_{1,1} & \vzero & \dots & \vzero \\
e_{1,2} & \vzero & \dots & \vzero \\
\vzero & e_{2,1} & \dots & \vzero \\
\vzero & e_{2,2} & \dots & \vzero \\
\vdots & \vdots & \ddots & \vdots \\
\vzero & \vzero & \dots & e_{n,1} \\
\vzero & \vzero & \dots & e_{n,2}
\end{bmatrix}, 
\begin{array}{rcl}
\multicolumn{3}{c}{\text{where}} \\
\tilde a_{i,j} & = & \begin{bmatrix} a_{i,j} \\ \vzero_{2n \times 1} \end{bmatrix} \in \F^{1 + 2n} \\
e_{i,j} & =& \vec{e}_{1+2(i-1)+j} \in \F^{1+2n} \text{ for } i\in[n], j\in[2] \\ \\
\multicolumn{3}{c}{\text{ and the frontal slices are}} \\ \\
A_1 & = & \begin{bmatrix} A \\ \vzero_{2n \times n} \end{bmatrix} \\
A_{1 + 2(i-1) + j} & = & E_{d + 2(i-1) + j, i} \qquad \text{ for } i\in[n], j\in[2]
\end{array}
$$
(In $\tA$ we turn the vectors $\tilde a_{i,j}$ and $e_{i,j}$ ``on their side'' so they become perpendicular to the page. )

We claim that $A$ and $B$ are monomially equivalent as codes if and only if $\tA$ 
and $\tB$ are isomorphic as 3-tensors.

($\Rightarrow$) Suppose $QADP = B$ where $Q \in \GL(n,\F)$, $D = \diag(\alpha_1, \dotsc, \alpha_n)$ and $P \in S_n \leq \GL(n,\F)$. Then by examining the frontal slices it is not hard to see that for $Q' = \begin{bmatrix} Q & 0 \\ 0 & (DP)^{-1} \otimes I_2 \end{bmatrix}$ (where $DP^{-1} \otimes I_2$ denotes a $2n \times 2n$ block matrix, where the pattern of the nonzero blocks and the scalars are governed by $(DP)^{-1}$, and each $2 \times 2$ block is either zero or a scalar multiple of $I_2$) we have $Q' A_1 DP = B_1$ and $Q' A_{1 + 2(i-1) + j} DP = B_{1 + 2(\pi(i)-1) + j}$, where $\pi$ is the permutation corresponding to $P$. Thus $\tA$ and $\tB$ are isomorphic tensors, via the isomorphism $(Q', DP, \diag(I_1, P))$.

($\Leftarrow$) 
Suppose there exist $Q\in \GL(d+2n, \F)$, $P\in \GL(n, \F)$, 
and $R\in \GL(1+2n, \F)$, such that $Q\tA P =\tB^R$. 
First, note that every
lateral slice of $\tA$ is of rank either $2$ or $3$, 
and the actions of $Q$ and $R$ do 
not change the ranks of the lateral slices.
Furthermore, any non-trivial linear 
combination of more than $1$ lateral slice results in a lateral matrix of rank 
$\geq 4$. It follows that $P$ cannot take nontrivial linear combinations of the lateral slices, hence it must be monomial. 

Now consider the frontal 
slices. Note that, as we assume $d>1$, every frontal slice of $Q\tA P$, except the 
first one, is of 
rank $1$. Therefore, $R$ must be of the form 
$\begin{bmatrix} r_{1,1} & \vzero_{1 \times (n-1)} \\ \vec{r'} & R' \end{bmatrix}$ 
where $R'$ is $(n-1) \times (n-1)$. Since $R$ is invertible, we must have $r_{1,1}\neq 0$, and the first frontal slice of $\tB^R$ 
contains all the rows of $B$ scaled by $r_{1,1}$ in its first $d$ rows. The first frontal slice of $Q\tA P$ is a matrix that generates, by definition (and since we've shown $P$ is monomial), a code monomially equivalent to $A$. Since the first frontal slices of $Q \tA P$ and $\tB^R$ are equal, and the latter is just a scalar multiple of $B_1$, we have that $A$ and $B$ are monomially equivalent as codes as well.
\end{proof}

\section{Preliminaries}\label{sec:prel}

\begin{table}[!htbp]
\begin{center}
\begin{tabular}{lll}
Font & Object & Space of objects \\ \hline
$A, B, \dotsc$ & matrix & $\M(n,\F)$ or $\M(\ell \times n, \F)$ \\
$\vA, \vB, \dotsc$ & matrix tuple & $\M(n,\F)^m$ or $\M(\ell \times n, \F)^m$ \\
$\cA, \cB, \dotsc$ & matrix space & [Subspaces of 
$\M(n,\F)$ or $\Lambda(n,\F)$]\\
$\tA, \tB, \dotsc$ & 3-way array & $\T(\ell \times n \times m, \F)$
\end{tabular}
\end{center}
\caption{\label{table:notation} Summary of notation related to 3-way arrays and tensors.}
\end{table}
\paragraph{Vector spaces.}
Let $\F$ be a field. In this paper we only consider 
finite-dimensional vector spaces over $\F$. We use $\F^n$ to denote the 
vector space of 
length-$n$ \emph{column} vectors. The $i$th standard basis 
vector of $\F^n$ is denoted as 
$\vec{e_i}$. Depending on the context, $\vzero$ may denote the zero vector space, 
a zero vector, or an all-zero matrix. Let $S$ be 
a subset of 
vectors. We use $\langle S\rangle$ to denote the subspace 
spanned by elements in $S$. 

\hyperdef{}{sec:prel:groups}{}
\paragraph{Some groups.} The general linear group of degree $n$ over a field $\F$ 
is denoted by $\GL(n, \F)$. The symmetric group of degree $n$ is denoted by 
$\S_n$. The natural embedding of $\S_n$ into $\GL(n, \F)$ is to represent 
permutations by permutation matrices. A monomial matrix in $\M(n, \F)$ is a matrix 
where each row and each column has exactly one non-zero entry. All monomial 
matrices form a subgroup of $\GL(n, \F)$ which we call the monomial 
subgroup, denoted by $\Mon(n, \F)$, which is isomorphic to the semi-direct product 
$\F^n \rtimes S_n$.
The subgroup of $\GL(n, \F)$ consisting of block upper-triangular matrices with a 
fixed block structure is called a (standard) parabolic subgroup.

\paragraph{Matrices.} Let $\M(\ell\times n, \F)$ be the linear space of 
$\ell\times n$ matrices over 
$\F$, and $\M(n, \F):=\M(n\times n, \F)$. 
Given $A\in \M(\ell\times 
n, \F)$, $A^t$ denotes the transpose of $A$. 

A matrix $A\in \M(n, \F)$ is \emph{symmetric}, if for any $u, v\in \F^n$, 
$u^tAv=v^tAu$, or equivalently $A = A^t$. 
That is, $A$ represents a symmetric bilinear form. 
A matrix $A\in \M(n, \F)$ is \emph{alternating}, if for any $u\in \F^n$, 
$u^tAu=0$. That is, $A$ represents an alternating bilinear form. 
Note that in characteristic $\neq 2$, alternating is the same as 
skew-symmetric, but in characteristic 2 they differ (in characteristic 2, 
skew-symmetric=symmetric). The linear space of $n\times n$ 
alternating matrices over $\F$ is denoted by $\Lambda(n, \F)$. 

The $n\times n$ \emph{identity matrix} is 
denoted by $I_n$, and when $n$ is clear from the context, we may just write $I$. 
The \emph{elementary matrix} $E_{i,j}$ is the matrix with the $(i,j)$th entry 
being $1$, 
and other entries being $0$. The \emph{$(i,j)$-th elementary alternating matrix} is the matrix $E_{i,j} - E_{j,i}$.

\paragraph{Matrix tuples.} We use $\M(\ell\times n, \F)^{m}$ to denote the linear 
space of $m$-tuples of $\ell\times n$ matrices. Boldface letters like $\vA$ and 
$\vB$ denote matrix tuples. Let $\vA=(A_1, \dots, A_m), \vB=(B_1, \dots, B_m)\in 
\M(\ell\times n, \F)^m$.  Given $P\in\M(\ell, \F)$ and $Q\in\M(n, \F)$, $P\vA 
Q:=(PA_1Q, \dots, PA_mQ)\in \M(\ell, \F)$. Given $R=(r_{i,j})_{i,j\in[m]}\in \M(m, 
\F)$, $\vA^R:=(A_1', \dots, A_m')\in\M(m,\F)$ where 
$A_i'=\sum_{j\in[m]}r_{j, i}A_j$. 
\begin{remark}
In particular, note that $A_i'$ 
corresponds to the entries in the $i$th \emph{column} of $R$. 
While this choice is immaterial (we could have chosen the opposite 
convention), all of our later calculations are consistent with this 
convention.
\end{remark}

Given $\vA, \vB\in\M(\ell\times n, \F)^m$, we say that $\vA$ and $\vB$ are 
\emph{equivalent}, if there exist $P\in\GL(\ell, \F)$ and 
$Q\in\GL(n, \F)$, such that $P\vA Q=\vB$. Let $\vA, \vB\in\M(n, \F)^m$. Then $\vA$ 
and $\vB$ are \emph{conjugate}, if there exists $P\in\GL(n, \F)$, such that 
$P^{-1}\vA 
P=\vB$. And $\vA$ and $\vB$ are \emph{isometric}, if there exists $P\in\GL(n, 
\F)$, such 
that $P^t\vA P=\vB$. Finally, $\vA$ and $\vB$ are pseudo-isometric, if there exist 
$P\in\GL(n, \F)$ and $R\in\GL(m, \F)$, such that $P^t\vA P=\vB^R$.

\paragraph{Matrix spaces.} Linear subspaces of $\M(\ell\times n, \F)$ are called 
matrix 
spaces. Calligraphic letters like $\cA$ and $\cB$ denote 
matrix spaces. By a slight abuse of notation, 
for $\vA\in 
\M(\ell\times n, \F)^m$, we use $\langle \vA\rangle$ to denote 
the subspace spanned by those matrices in $\vA$.

\paragraph{3-way arrays.} Let $\T(\ell\times n\times m, \F)$ be the linear space 
of $\ell\times n\times m$ 3-way arrays over $\F$. We use the fixed-width teletypefont for 3-way 
arrays, like $\tA$, $\tB$, etc.. 

Given $\tA\in \T(\ell\times n\times m, \F)$, we can 
think of $\tA$ as a 3-dimensional table, where the $(i,j,k)$th entry is denoted as 
$\tA(i,j,k)\in \F$. We can slice $\tA$ along one direction and obtain several 
matrices, which are then called slices. For example, slicing along the first 
coordinate, we obtain the \emph{horizontal} slices, namely $\ell$ matrices $A_1, 
\dots, A_\ell\in \M(n\times m, \F)$, where $A_i(j,k)=\tA(i,j,k)$. Similarly, we 
also obtain the \emph{lateral} slices by slicing along the second coordinate, and 
the \emph{frontal} slices by slicing along the third coordinate. 

We will often represent a 3-way array as a matrix whose entries are vectors. That 
is, given $\tA\in\T(\ell\times n\times m, \F)$, we can write
$$
\tA=\begin{bmatrix}
w_{1,1} & w_{1,2} & \dots & w_{1,n}\\
w_{2,1} & w_{2,2} & \dots & w_{2,n}\\
\vdots & \ddots & \ddots & \vdots \\
w_{\ell,1} & w_{\ell,2} & \dots & w_{\ell, n}
\end{bmatrix},
$$
where $w_{i,j}\in \F^m$, so that $w_{i,j}(k)=\tA(i,j,k)$. Note that, while 
$w_{i,j}\in \F^m$ are column vectors, in the above representation of $\tA$, we 
should think of them as along 
the direction ``orthogonal to the paper.'' Following \cite{KB09}, we call 
$w_{i,j}$ the 
\emph{tube fibers} of $\tA$. Similarly, we can have the \emph{row fibers} 
$v_{i,k}\in\F^n$ such that $v_{i,k}(j)=\tA(i,j,k)$, and the \emph{column fibers} 
$u_{j,k}\in\F^\ell$ such that $u_{j,k}(i)=\tA(i,j,k)$.

Given $P\in \M(\ell, \F)$ and $Q\in \M(n, \F)$, let 
$P\tA Q$ be the $\ell \times 
n\times m$ $3$-way array whose $k$th frontal slice is $P A_k Q$. For 
$R=(r_{i,j})\in 
\GL(m, \F)$, let $\tA^R$ be the $\ell\times n\times m$ $3$-way array whose 
$k$th 
frontal slice is $\sum_{k'\in[m]}r_{k',k}A_{k'}$.  
Note that these notations are consistent with the notations for matrix 
tuples above, when we consider the matrix tuple $\vA = (A_1, \dotsc, A_k)$ 
of frontal slices of $\tA$.

Let $\tA\in \T(\ell\times n\times m, \F)$ be a 3-way array. We say that $\tA$ is 
\emph{non-degenerate} as a 
3-tensor if 
the horizontal slices of $\tA$ are linearly independent, the lateral slices are 
linearly independent, and the frontal slices are linearly independent. 
Let $\vA=(A_1, 
\dots, A_m)\in \M(\ell\times n, \F)^m$ be a matrix tuple consisting of the frontal 
slices of $\tA$. Then it is easy to see that the frontal slices of $\tA$ are 
linearly independent if and only if $\dim(\langle\vA\rangle)=m$. The lateral 
(resp., horizontal) slices of $\tA$ are linearly independent if and only if the 
intersection of the right (resp., left) kernels of $A_i$ is zero.
\begin{observation}\label{obs:nondeg}
Given $3$-way arrays $\tA$ and $\tB$, we can construct non-degenerate 
3-way arrays $\tA'$ and $\tB'$ in polynomial time, such that $\tA$ and $\tB$ are 
isomorphic as 
3-tensors if and only if $\tA'$ and $\tB'$ are isomorphic as 3-tensors.
\end{observation}

\paragraph{Multi-way arrays.} For $d \geq 3$, we use similar notation to 3-way arrays, which we will not belabor. Here we merely observe:

\begin{observation} \label{obs:d}
For any $d' \geq d$, $d$-\TI reduces to $d'$-\TI.
\end{observation}

\begin{proof}
Given an $n_1 \times \dotsb \times n_d$ $d$-way array $\tA$, we embed it as a $d'$-way array $\tilde \tA$ of format $n_1 \times \dotsb \times n_d \times 1 \times 1 \times \dotsb \times 1$. If $\tA \cong \tB$ as $d$-tensors, say via $(P_1, \dotsc, P_d)$, then $\tilde \tA \cong \tilde \tB$ as $d'$-tensors via $(P_1, \dotsc, P_d,1,1,\dotsc,1)$. Conversely, if $\tilde \tA \cong \tilde \tB$ via $(P_1,\dotsc, P_d,\alpha_{d+1},\dotsc,\alpha_{d'})$, then $\tA \cong \tB$ via $(\alpha_{d+1} \alpha_{d+2} \dotsb \alpha_{d'} P_1, \dotsc, P_d)$. That is, all that can ``go wrong'' under this embedding is multiplication by scalars, but those scalars can be absorbed into any one of the $P_i$.
\end{proof}

\paragraph{Algebras and their algorithmic representations.} An algebra $A$ 
consists of a vector space $V$ and a bilinear map $\circ: V\times V\to V$. This 
bilinear map defines the 
product $\circ$ in this algebra. Note that we do not assume $A$ to be unital 
(having an identity), associative, alternating, nor satisfying the Jacobi 
identity. In the literature, an algebra without such properties is sometimes 
called a non-associative algebra (but also, as usual, associative algebras are a special case of non-associative 
algebras).

As in Section~\ref{sec:intro}, after fixing 
an ordered basis $(b_1, \dots, b_n)$ where $b_i\in\F^n$ of $V\cong \F^n$, this 
bilinear map $\circ$ can be represented by an $n\times n\times n$ 3-way array $\tA$, 
such that $b_i\circ b_j=\sum_{k\in[n]}\tA(i,j,k)b_k$. This is the structural 
constant representation of $\tA$. Algorithms for associative algebras and Lie 
algebras have been studied intensively in this model, e.\,g., \cite{IR99,Gra00}.

It is also natural to consider matrix spaces that are closed under multiplication 
or commutator. More specifically, let $\cA\leq \M(n, \F)$ be a matrix space. If 
$\cA$ is closed under multiplication, that is, for any $A, B\in\cA$, $AB\in \cA$, 
then $\cA$ is a matrix (associative) algebra with the product being the matrix 
multiplication. If $\cA$ is closed under commutator, that is, for any $A, B\in 
\cA$, $[A,B]=AB-BA\in\cA$, then $\cA$ is a matrix Lie algebra with the product 
being the commutator. Algorithms for matrix algebras and matrix Lie algebras have 
also been studied, e.\,g., \cite{EG00,Iva00,IR99}.

\paragraph{The Lazard correspondence for $p$-groups.} The Lazard correspondence is 
a correspondence between certain classes of groups and Lie algebras, which 
extends the usual correspondence between Lie groups and Lie algebras (say, over 
$\R$) to some groups and Lie algebras in positive characteristic. Here we state 
just enough to give a sense of it; for further details we refer to Khukhro's book 
\cite{khukhro} and Naik's thesis \cite{naik}. While the thesis is quite long, it also includes a 
reader's guide, and collects many results scattered across the literature or 
well-known to the experts in one place, building the theory from the ground up and 
with many examples.

Recall that a \emph{Lie ring} is an abelian group $L$ equipped with a bilinear map $[,]$, called the Lie bracket, which is (1) alternating ($[x,x]=0$ for all $x \in L$) and (2) satisfies the Jacobi identity $[x,[y,z]] + [y,[z,x]] + [z,[x,y]] = 0$ for all $x,y,z \in L$. Let $L^1 = L$, and $L^{i+1} = [L,L^i]$, which is the subgroup (of the underlying additive group) generated by all elements of the form $[x,y]$ for $x \in L, y \in L^{i}$. Then $L$ is \emph{nilpotent} if $L^{c+1}=0$ for some finite $c$; the smallest such $c$ is the \emph{nilpotency class}. (Lie algebras are just Lie rings over a field.)

The correspondence between Lie algebras and Lie groups over $\R$ uses the Baker--Campbell--Hausdorff (BCH)
formula to convert between a Lie algebra and a Lie group, so we start there. The BCH formula is the solution to the problem that for non-commuting matrices $X,Y$, $e^X e^Y \neq e^{X + Y}$ in general (where the matrix exponential here is defined using the power series for $e^x$). Rather, using commutators $[A,B] = AB - BA$, we have
\[
\exp(X) \exp(Y) = \exp\left(X + Y + \frac12 [X,Y] + \frac{1}{12}\left([X,[X,Y]] - [Y, [X,Y]] \right) - \frac{1}{24} [Y, [X, [X, Y]]] + \dotsb\right),
\]
where the remaining terms are iterated commutators that all involve at least 5 $X$s and $Y$s, and successive terms involve more and more.  Applying the exponential function to a Lie algebra in characteristic zero yields a Lie group. The BCH formula can be inverted, giving the correspondence in the other direction.

In a nilpotent Lie algebra, the BCH formula has only finitely many nonzero terms, so issues of convergence disappear and we may consider applying the correspondence over finite fields or rings; the only remaining obstacle is that the denominators appearing in the formula must be units in the ring. It turns out that the correspondence 
continues to work in characteristic $p$ so long as one does not need to use the 
$p$-th term of the BCH formula (which includes division by $p$), and the latter is 
avoided whenever a nilpotent group has class strictly less than $p$. While the 
correspondence does apply more generally, here we only state the version for 
finite groups. For any fixed nilpotency class $c$, computing the Lazard 
correspondence is efficient in theory; for how to compute it in practice when the groups are given by polycyclic presentations, 
see \cite{cicaloEtAl}.

Let $\mathbf{Grp}_{p,n,c}$ denote the set of finite groups of order $p^n$ and class $c$, and let $\mathbf{Lie}_{p,n,c}$ denote the set of Lie rings of order $p^n$ and class $c$. We note that for nilpotency class 2, the Baer correspondence is the same as the Lazard correspondence.

\begin{theorem}[{Lazard Correspondence for finite groups, see, e.\,g., \cite[Ch.~9 \& 10]{khukhro} or \cite[Ch.~6]{naik}}] \label{thm:lazard}
For any prime $p$ and any $1 \leq c < p$, there are functions $\logbf \colon \mathbf{Grp}_{p,n,c} \leftrightarrow \mathbf{Lie}_{p,n,c} : \expbf$ such that (1) $\logbf$ and $\expbf$ are inverses of one another, (2) two groups $G,H \in \mathbf{Grp}_{p,n,c}$ are isomorphic if and only if $\logbf(G)$ and $\logbf(H)$ are isomorphic, and (3) if $G$ has exponent $p$, then the exponent of the underlying abelian group of $\logbf(G)$ has exponent $p$. More strongly, $\logbf$ is an isomorphism of categories $\mathbf{Grp}_{p,n,c} \cong \mathbf{Lie}_{p,n,c}$. 
\end{theorem}

Part (3) can be found as a special case of \cite[Lemma~6.1.2]{naik}.

For $p$-groups given by $d \times d$ matrices over the finite field $\F_{p^e}$, we will need one additional fact about the correspondence, namely that it also results in a Lie algebra of $d \times d$ matrices. (Being able to bound the dimension of the Lie algebra and work with it in a simple linear-algebraic way 
seems crucial for our reduction to work efficiently.) In fact, the BCH correspondence is \emph{easier} to see for matrix groups using the matrix exponential and matrix logarithm; most of the work for BCH and Lazard is to get the correspondence to work even \emph{without} the matrices. In some sense, this is thus the ``original'' setting of this correspondence. Though it is surely not new, we could not find a convenient reference for this fact about matrix groups over finite fields, so we state it formally here.

\begin{proposition} \label{prop:lazard_matrices}
Let $G \leq \GL(d,\F_{p^e})$ be a finite $p$-subgroup of $d \times d$ matrices 
over a finite field of characteristic $p$. Then $\logbf(G)$ (from the Lazard 
correspondence) can be realized as a finite Lie subalgebra of $d \times d$ 
matrices over $\F_{p^e}$. Given a generating set for $G$ of $m$ matrices, a generating set for $\logbf(G)$ can be constructed in $\poly(d,n,\log p)$ time.
\end{proposition}

\begin{proof}[Proof sketch]
$G$ is conjugate in $\GL(d,\F_{p^e})$ to a group of upper unitriangular matrices (upper triangular with all 1s on the diagonal); this is a standard fact that can be seen in several ways, for example, by noting that the group $U$ of all upper unitriangular matrices in $\GL(d,\F_{p^e})$ is a Sylow $p$-subgroup, and applying Sylow's Theorem. (Note that we do not need to do this conjugation algorithmically, though it is possible to do so; this is only for the proof.) 
Thus we may write every $g \in G$ as $1 + n$, where the sum here is the ordinary sum of matrices, $1$ denotes the identity matrix, and $n$ is strictly upper triangular. In particular, $n^d = 0$ (ordinary exponentiation of matrices). Thus the Taylor series for the logarithm
\[
\log(1 + n) = n - \frac{n^2}{2} + \frac{n^3}{3} - \dotsb
\]
has only finitely many terms, so we may use it even over $\F_{p^e}$. 

In the Lie algebra we would like addition to be ordinary matrix addition; however, it turns out that we can write this addition in terms of a formula involving only commutators of group elements. Deriving this formula---the so-called first BCH inverse formula---for the matrices will be the same, step for step, as deriving the first inverse BCH formula in general. Since the formulae are identical, the additive structures on $\log(G)$ (using the matrix logarithm) and $\logbf(G)$ (from the Lazard correspondence) are identical. Similar considerations apply to the matrix commutator $[\log(g), \log(h)] = \log(g) \log(h) - \log(h)\log(g)$, now using the second BCH inverse formula. Overall, we conclude that $\logbf(G)$ (using Lazard) and $\log(G)$ (using the matrix logarithm) are isomorphic Lie algebras.

Equivalently, we may note that the derivation of the inverse BCH formula in \cite{khukhro,naik} uses a free nilpotent associative algebra as an ambient setting in which both the group and the corresponding Lie algebra live; in our case, we may replace the ambient free nilpotent associative algebra with the algebra of $d \times d$ strictly upper-triangular matrices over $\F_{p^e}$, and all the derivations remain the same, \emph{mutatis mutandis}. See, for example, \cite[p.~105, ``Another remark...'']{khukhro}.
\end{proof}

\subsection{Tensor notation} \label{sec:prelim:tensor}
To see that those problems in Section~\ref{sec:problems} exhaust 
distinct 
isomorphism problems coming from change-of-basis on 3-way arrays (without 
introducing multiple arrays, or block structure, or going to subgroups of $\GL(n, 
\F)$), and to keep track of the relation between all the above problems, we use 
standard mathematical notation for spaces of tensors (however, we won't actually 
need the full abstract definition here; for a formal introduction see, e.\,g., 
\cite{Lan12}).

Much as the three natural equivalence relations on matrices differ by how 
the groups act on the rows and columns, the same is true for tensors, but 
on the rows, columns, and depths (the ``row-like'' sub-arrays which are 
``perpendicular to the page''). There are two aspects to the notation: 
first, we keep track of which group is acting where by introducing names $U,V,W$ for 
the different vector spaces involved (this is also the standard basis-free 
notation, e.\,g., \cite{Lan12}) and the groups acting on them, viz. 
$\GL(U), \GL(V), \GL(W)$, etc. Thus, while it is possible that $\dim U = 
\dim V$ and thus $\GL(U) \cong \GL(V)$, the notation helps make clear 
which group is acting where. Second, to take into account the 
contragradient (``inverse'') action, given a vector space $V$, $V^*$ 
denotes its dual space, consisting of the linear functions $V \to \F$. 
$\GL(V)$ acts on $V^*$ by sending a linear function $\ell \in V^*$ to the 
function $(g \cdot \ell)(v) = \ell(g^{-1}(v))$. 
In this notation, the three different actions on matrices correspond to the notations 
\[
U \otimes V \text{ (left-right action)} \qquad\qquad V \otimes V^* \text{ (conjugacy)} \qquad\qquad V \otimes V \text{ (isometry).} 
\]

When we have a matrix \emph{space} $\cA \subseteq M(n \times m,\F)$ instead of a single matrix $A$, we introduce an additional vector space $W$, which is naturally isomorphic to $\cA$ as a vector space. The action of $\GL(W)$ on $W$ serves to change basis \emph{within} the matrix space, while leaving the space itself unchanged. In this notation, the problems we mention above are listed in Table~\ref{table:problems}.

\begin{table}[!htbp]
\label{table:problems}
\begin{center}
\begin{tabular}{lcc}
Notation & Name & Group Action  \\ \hline
$U \otimes V \otimes W$ &  
\begin{tabular}{@{}c@{}}\MatSpEquivlong \\ \ThreeTIlong\end{tabular}
& $\cA \mapsto g \cA 
h^{-1}$  \\ \hline
$V \otimes V \otimes W$ & 
\begin{tabular}{@{}c@{}}\MatSpIsomlong \\ \algprobm{Bilinear Map 
Pseudo-Isometry}\end{tabular}  & $\cA \mapsto g \cA g^T $ \\ \hline
$V \otimes V^* \otimes W$ & \MatSpConjlong & $\cA \mapsto g \cA g^{-1}$ \\ \hline
$V \otimes V \otimes V$ & \NcCubicFormlong & 
$f(\vec{x}) \mapsto f(g^{-1} \vec{x})$ \\ \hline
$V \otimes V \otimes V^*$ & \AlgIsolong & 
$\mu(\vec{x}, 
\vec{y})\mapsto g 
\mu(g^{-1}\vec{x}, g^{-1}\vec{y})$
\end{tabular}
\end{center}
\caption{\label{table:problems} The cast of isomorphism problems on 3-way arrays. In Section~\ref{sec:prelim:tensor} we show how this exhausts the possibilities. }
\end{table}

To see that the family of problems in Table~\ref{table:problems}
exhausts the possible isomorphism problems on (undecorated) 3-way 
arrays, we note that in this notation there are some 
``different-looking'' isomorphism problems that are trivially 
equivalent. The first is re-ordering the spaces: the isomorphism 
problem for $V \otimes V \otimes W$ is trivially equivalent to that 
for $V \otimes W \otimes V$, simply by permuting indices, viz. 
$\tA'(i,j,k) = \tA(i,k,j)$. The second is about dual vector spaces. 
Although a vector space $V$ and its dual $V^*$ are technically 
different, and the group action differs by an inverse transpose, we 
can choose bases in $V$ and $V^*$ so that there is a linear 
isomorphism $V \to V^*$ which induces a bijection between orbits; 
for example, the orbits of the action $g \cdot A = gAg^t$ are the 
same as the orbits of the action $g \cdot A = g^{-t} A g^{-1}$, even 
though technically the former corresponds to $V \otimes V$ and the 
latter to $V^* \otimes V^*$. This means that if we are considering 
the isomorphism problem in a tensor space such as $V \otimes V 
\otimes W$, we can dualize each of the vector spaces $V,W$ 
separately, so long as when we do so, we dualize all instances of 
that vector space. For example, the isomorphism problem in $V 
\otimes V \otimes W$ is trivially equivalent to that in $V^* \otimes 
V^* \otimes W$, but is not obviously equivalent to that in $V 
\otimes V^* \otimes W$ (though we will show such a reduction in this 
paper). As a consequence, when the action on all three directions 
comes from the same group, there are only two choices: $V \otimes V 
\otimes V$ and $V \otimes V \otimes V^*$; the remaining choices are 
trivially equivalent to one of these two. Using these, we see that 
the Table~\ref{table:problems} in fact covers all possibilities up to these 
trivial equivalences.

\paragraph{Special cases of interest.} As in the case of isometry of 
matrices, 
wherein skew-symmetric and symmetric matrices play a special role, the same is 
true for isometry of matrix spaces. We say a matrix space $\cA$ is symmetric if 
every matrix $A \in \cA$ is symmetric, and similarly for skew-symmetric or alternating. 
\SymMatSpIsomlong is equivalent to asking whether two 
polynomial maps from $\F^n$ to $\F^m$ specified by homogeneous quadratic forms are 
the same under the action of $\GL(n, \F)\times \GL(m, \F)$. This problem has been proposed 
by Patarin \cite{Pat96} as the basis of security for certain identification and 
signature schemes. 
\AltMatSpIsomlong is a particular case of interest, being 
in many ways a linear-algebraic analogue of \GI \cite{LQ17} (in addition to its 
close relation with \GpIlong for $p$-groups of class 2 and exponent $p$).

Among trilinear forms, we can identify commutative cubic forms as those 
for which the coefficient 3-way array $\tA$ is symmetric under all 6 permutations 
of its 3 indices $\tA(i,j,k) = \tA(j,i,k) = \dotsb = 
\tA(k,i,j)$. 
Over rings in which $6$ is a unit, cubic forms embed into trilinear forms via the 
standard map $f \mapsto T$ where $T_{i_1,i_2,i_3} = \frac{1}{3!} \sum_{\pi \in 
S_3} [x_{i_{\pi(1)}} x_{i_{\pi(2)}} x_{i_{\pi(3)}}] f$, where $[x^e] f$ denotes 
the coefficient of $x^e$ in $f$. Thus, over such rings \CubicFormlong, as studied 
by Agrawal and Saxena \cite{AS05, AS06}, is a special case of \NcCubicFormlong.

Special cases of \AlgIsolong that are of interest include those of unital, 
associative algebras (commutative, e.\,g., as studied in \cite{AS05, AS06, KS06}, 
and non-commutative, such as group algebras) and Lie algebras. 

Interesting cases of \MatSpConjlong arise naturally whenever we have an algebra 
$A$ (say, associative or Lie) that is given to us as a subalgebra of the algebra 
$\M(n,\F)$ of $n \times n$ matrices. Two such matrix algebras can be isomorphic as 
abstract algebras, but the more natural notion of ``isomorphism of matrix 
algebras'' is that of conjugacy, which respects both the algebra structure and the 
presentation in terms of matrices. This is the linear-algebraic analogue of 
permutational isomorphism (=conjugacy) of permutation groups, and has been studied 
for matrix Lie algebras \cite{GrochowLie} and associative matrix algebras 
\cite{BW15}. (For those who know what a representation is: it also turns out to be 
equivalent to asking whether two representations of an algebra $A$ are equivalent 
up to automorphisms of $A$, a problem which naturally arises as a subroutine in, 
e.\,g., \GpIlong, where it is often known as \algprobm{Action Compatibility}, e.\,g., 
\cite{GQ17}.)

\subsection{On the type of reduction} \label{sec:reductions}
As these problems arise from several different fields, there are various properties one might hope for in the notion of reduction. Most of our reductions satisfy all of the following properties; see Remark~\ref{rmk:reductions} below for details.

\begin{enumerate}
\item \label{reduction:kernel} Kernel reductions: there is a function $r$ from objects of one type to objects of the other such that $A \sim_1 B$ if and only if $r(A) \sim_2 r(B)$.  See \cite{FortnowGrochowPEq} for some discussion on the relation between kernel reductions and more general reductions.

\item \label{reduction:efficient} Efficiently computable: $r$ as above is computable in polynomial time. In fact, we believe, though have not checked fully, that all of our reductions are computable by uniform constant-depth (algebraic) circuits; over finite fields and algebraic number fields, we believe they are in uniform $\cc{TC}^0$ (the threshold gates are needed to do some simple arithmetic on the indices). That is, there is a small circuit which, given $A,i,j,k$ as input will output the $(i,j,k)$ entry of the output.

\item \label{reduction:projection} Polynomial-size projections (``p-projections'') \cite{valiantProj}: each coordinate of the output is either one of the input variables or a constant, and the dimension of the output is polynomially bounded by the dimension of the input. (In fact, in many cases, the dimension of the output is only linearly larger than that of the input.)

\item \label{reduction:functorial} Functorial. For each type of tensor there is naturally a category of such tensors (see \cite{MacLane} for generalities on categories). For example, for \ThreeTI, $U \otimes V \otimes W$, the objects of the category are three-tensors, and a morphism between $\tA \in U \otimes V \otimes W$ and $\tB \in U' \otimes V' \otimes W'$ is given by linear maps $P: U \to U'$, $Q\colon V \to V'$, and $R\colon W \to W'$ such that $(P,Q,R) \cdot \tA = \tB$. Isomorphism of 3-tensors is the special case when all three of $P,Q,R$ are invertible. Analogous categories can be defined for the other problems we consider, such as $V \otimes V^* \otimes W$. A \emph{functor} between two categories $\mathcal{C}, \mathcal{D}$ is a pair of maps $(r,\overline{r})$ such that (1) $r$ maps objects of $\mathcal{C}$ to objects of $\mathcal{D}$, (2) if $f\colon A \to B$ is a morphism in $\mathcal{C}$, then $\overline{r}(f)\colon r(A) \to r(B)$ is a morphism in $\mathcal{D}$, (3) for any $A \in \mathcal{C}$, $\overline{r}(\id_A) = \id_{r(A)}$, and (4) if $f\colon A \to B$ and $g\colon B \to C$ are morphisms in $\mathcal{C}$, then $\overline{r}(g \circ f) = \overline{r}(g) \circ \overline{r}(f)$.

All our reductions are functorial on the categories in which we only consider 
isomorphisms; we 
suspect that they are also functorial on the entire categories (that is, including 
non-invertible homomorphisms). Furthermore, all our reductions yield another map 
$\overline{s}$ such that for any isomorphism $f'\colon r(A) \to r(B)$, 
$\overline{s}(f)$ is an isomorphism $A \to B$, and $\overline{s}(\overline{r}(f)) 
= f$ for any isomorphism $f\colon A \to B$. If we only consider isomorphisms (and 
not other homomorphisms), we believe all known reductions between isomorphism 
problems have this form, cf. \cite{BabaiSR}.

\item \label{reduction:algebraic} Containment in the sense used in the literature on wildness. There are several definitions in the literature which typically are equivalent when restricted to so-called matrix problems. For a few such definitions, see, e.\,g., \cite[Def.~1.2]{FGS19}, \cite{Sergeichuk2000}, or \cite[Def.~XIX.1.3]{SimsonSkowronski}. For those problems in this paper to which the preceding definitions could apply, our reductions have the defined property. However, since we are working in a slightly more general setting, we would like to suggest the following natural generalization of these notions. 
Given two pairs $(G,V)$ and $(H,W)$ of algebraic groups $G,H$ acting on algebraic varieties $V,W$, an \emph{algebraic containment} is an algebraic map $r\colon V \to W$ (each coordinate of the output is given by a polynomial in the coordinates of the input) that is also a kernel reduction. In our case, all our spaces $V,W$ are affine space $\F^n$ for some $n$, and our maps $r$ are in fact of degree 1. (It might be interesting to consider whether using higher degree allows for more efficient reductions.) We may also require it to be ``functorial,'' in the sense that there is a homomorphism of algebraic groups $\overline{r}\colon G \to H$ (simultaneously an algebraic map and a group homomorphism) such that
\[
\overline{r}(g) \cdot r(v) = r(g \cdot v).
\]
and a section $\overline{s}\colon H \dashrightarrow G$, such that $\overline{s} \circ \overline{r} = \id_G$ and
\[
h \cdot r(v) = r(v') \Longrightarrow \overline{s}(h) \circ v = v',
\]
where the dashed arrow above indicates that $\overline{s}$ need only be defined on a subset of $H$, namely, those $h \in H$ such that there exist $v,v' \in V$ with $h \cdot r(v) = r(v')$ (but on this subset it should still act like a homomorphism, in the sense that $\overline{s}(h h') = \overline{s}(h) \overline{s}(h')$).

\end{enumerate}

\begin{remark} \label{rmk:reductions}
We believe all of our reductions satisfy all of the above properties, with the 
possible exceptions that \Prop{prop:3-tensor_isometry} and 
\Prop{prop:3-tensor_conjugacy} are only projections (\ref{reduction:projection}) 
and algebraic containments (\ref{reduction:algebraic}) on the set of 
\emph{non-degenerate} 3-tensors. These reductions still satisfy the other three 
properties on the set of all tensors: They are kernel reductions by construction; 
non-degeneracy presents no obstacle to polynomial-time computation 
(Observation~\ref{obs:nondeg}); and two tensors are isomorphic iff their 
non-degenerate parts are isomorphic, so they are still functorial. The obstacle to 
being projections or algebraic containments on the set of all 3-tensors here is 
closely related to the fact that the map sending a matrix to its row echelon form 
(or even just zero-ing out a number of rows so that the remaining non-zero rows 
are linearly independent) is neither a projection nor an algebraic map. We would 
find it interesting if there were reductions for these results satisfying all of 
the above properties for all 3-tensors.
\end{remark}

\section{Reductions using the linear algebraic coloring 
gadgets}\label{sec:reduction_gadget}

In this section, we present the remaining reductions that use the linear algebraic 
coloring idea. We first reduce \GIlong to \AltMatSpIsomlong, using a gadget to restrict the full 
general linear group to the monomial matrix 
group, 
similar to that in Section~\ref{sec:technique}. However, unlike in the 
case there, the use here requires slightly more care because of the 
alternating condition. 
We then 
reduce \ThreeTIlong to \AltMatSpIsomlong. 
The gadget there restricts the full general linear group to a 
parabolic subgroup. We note that such a gadget has appeared in \cite{FGS19}, while 
ours is a slight modification of that to be compatible with the alternating 
structure. Finally, we combine the two gadgets to give a search-to-decision 
reduction for \AltMatSpIsomlong over finite fields. 

\subsection{From \GIlong to \AltMatSpIsomlong}\label{sec:graphiso_alternating}

\begin{proposition}\label{prop:GI}
\GIlong reduces to \AltMatSpIsomlong.
\end{proposition}

For this proof we will need the concept of monomial isometry; 
see \hyperref{}{}{sec:prel:groups}{Some Groups} above. Recall that a matrix 
is monomial if, equivalently, it can be written as $DP$ where $D$ is a nonsingular 
diagonal matrix and $P$ is a permutation matrix. 
We say two matrix spaces $\cA, \cB$ are \emph{monomially isometric} if there is some $M \in \Mon(n,\F)$ such that $M^t \cA M = \cB$.

\begin{proof}
For a graph $G=([n], E)$, let $\vA_G$ be the alternating matrix tuple $\vA_G = (A_1, \dotsc, A_{|E|})$ with $A_e = E_{i,j} - E_{j,i}$ where $e=\{i,j\} \in E$, and let $\cA_G = \langle \vA_G \rangle$ be the alternating matrix space spanned by that tuple. If $P$ is a permutation matrix giving an isomorphism between two graphs $G$ and $H$, then it is easy to see that $P^t \cA_G P = \cA_H$, and thus the corresponding matrix spaces are isometric. The converse direction is not clear (and may even be false). Instead, we will first extend the spaces $\cA_G$ and $\cA_H$ by gadgets which enforce that $\cA_G$ and $\cA_H$ are isometric iff they are monomially isometric (Lemma~\ref{lem:gadget}). Given Lemma~\ref{lem:gadget}, it thus suffices to reduce \GI to \AltMatSpMonIsomlong.

Let us establish the latter reduction. We will show that $G \cong H$  if and only 
if $\cA_G$ and $\cA_H$ are monomially isometric. The forward direction was handled 
above. For the converse, suppose $P^t D^t \cA_G DP = \cA_H$ 
where $D$ is diagonal and $P$ is a permutation matrix.  
We claim that in this case, $P$ in fact gives an isomorphism from $G$ to $H$. First let us establish that $P$ alone gives an isometry between $\cA_G$ and $\cA_H$. Note that for any diagonal matrix $D=\diag(\alpha_1, \dotsc, \alpha_n)$ and any elementary alternating matrix $E_{i,j} - E_{j,i}$, we have $D^t (E_{i,j} - E_{j,i}) D = \alpha_i \alpha_j (E_{i,j} - E_{j,i})$. Since $\cA_G$ has a basis of elementary alternating matrices, the action of $D$ on this basis is just to re-scale each basis element, and thus $D^t \cA_G D = \cA_G$. Thus, we have $P^t \cA_G P = \cA_H$. 

Finally, note that $P^t (E_{i,j} - E_{j,i}) P = E_{\pi(i), \pi(j)} - E_{\pi(j), 
\pi(i)} = A_{\pi(e)}$, where $\pi \in \S_n$ is the permutation corresponding to 
$P$, and by abuse of notation we write $\pi(e) = \pi(\{i,j\}) = \{\pi(i), 
\pi(j)\}$ as well. Since the elementary alternating matrices are linearly 
independent, and $\cA_H$ has a basis of elementary alternating matrices, the only 
way for $A_{\pi(e)}$ to be in $\cA_H$ is for it to be equal to one of the basis 
elements (one of the matrices in $\vA_H$). In other words, $\pi(e)$ must be an 
edge of $H$. As $P$ is invertible, we thus have that $P$ gives an isomorphism $G 
\cong H$. 
\end{proof}

\begin{lemma} \label{lem:gadget}
\AltMatSpMonIsomlong reduces to \AltMatSpIsomWords. 

More specifically, there is a $\poly(n,m)$-time algorithm $r$ taking alternating matrix tuples to alternating matrix tuples, such that for $\vA, \vB \in \Lambda(n,\F)^m$, the matrix spaces $\cA = \langle \vA \rangle$ and $\cB = \langle \vB \rangle$ are monomially isometric if and only if the matrix spaces $\langle r(\vA) \rangle$ and $\langle r(\vB) \rangle$ are isometric.
\end{lemma}

\begin{proof}
For $\vA = (A_1, \dotsc, A_m) \in \Lambda(n,\F)^m$, define $r(\vA)$ to be the alternating matrix tuple 
$\tilde\vA=(\tilde A_1, \dots, \tilde A_{m+n^2})\in \Lambda(n + n^2, \F)^{m+n^2}$, 
where
\begin{enumerate}
\item For $k=1, \dots, m$, $\tilde A_k=\begin{bmatrix}
A_k & \vzero \\
\vzero & \vzero \\
\end{bmatrix}$.
\item For $k=m+(i-1)n+j$, $i\in[n]$, $j\in[n]$, $\tilde A_k$ is the elementary 
alternating matrix $E_{i,in+j} - E_{in+j,i}$.
\end{enumerate}
At this point, some readers may wish to look at the large 
matrix in Equation~\ref{eq:tA_tilde} 
and/or at Figure~\ref{fig:gadget}.

It is clear that $r$ can be computed in time $\tilde O((m+n^2)(n^2 + n)) = \poly(n,m)$. Given alternating matrix tuples $\vA, \vB$, let $\cA,\cB$ be the corresponding matrix spaces they span, and let $\tilde\cA = \langle r(\vA) \rangle$ and $\tilde\cB = \langle r(\vB) \rangle$. We claim that $\cA$ and $\cB$ are monomially isometric if and only if $\tilde\cA$ and $\tilde\cB$ are isometric.

To prove this, it will help to think of our matrix tuples $\vA, \tilde \vA$, etc. as (corresponding to) 3-way arrays, and to view these 3-way arrays from two different directions. Towards this end, write the 3-way array corresponding to $\vA$ as
$$
\tA=\begin{bmatrix}
\vzero & a_{1,2} & a_{1,3} & \dots & a_{1,n} \\
-a_{1,2} & \vzero & a_{2,3} & \dots & a_{2,n} \\
-a_{1,3} & -a_{2,3} & \vzero & \dots & a_{3,n}\\
\vdots & \ddots & \ddots & \ddots & \vdots \\
-a_{1,n} & -a_{2,n} & -a_{3,n} & \dots & \vzero
\end{bmatrix}, 
$$
where $a_{i,j}$ are vectors in $\F^m$ (``coming out of the page''), namely $a_{i,j}(k) = A_k(i,j)$. The frontal slices of this array are precisely the matrices $A_1, \dotsc, A_m$.

The 3-way array corresponding to $\tilde \vA = r(\vA)$ is then the  $(n+1)n\times 
(n+1)n\times (m+n^2)$ array:
\begin{equation}\label{eq:tA_tilde}
\tilde\tA=\left[
\begin{array}{ccccc;{2pt/2pt}ccc;{2pt/2pt}ccc;{2pt/2pt}c;{2pt/2pt}ccc}
\vzero & \tilde a_{1,2} & \tilde a_{1,3} & \dots & \tilde a_{1,n} & e_{1,1} & 
\dots & e_{1,n} & \vzero & \dots & \vzero & \ldots & \vzero & \dots & \vzero  \\
-\tilde a_{1,2} & \vzero & \tilde a_{2,3} & \dots & \tilde a_{2,n} & \vzero & 
\dots & \vzero & e_{2,1} & \dots & e_{2,n} & \ldots &  \vzero & \dots & \vzero  \\ 
\vdots & \ddots & \ddots & \ddots & \vdots & \ddots & \ddots & \ddots & \ddots & 
\ddots & \ddots  & \ldots & \ddots & \ddots & \vdots \\
-\tilde a_{1,n} & -\tilde a_{2,n} & -\tilde a_{3,n} & \dots & \vzero & \vzero & 
\dots & \vzero & \vzero & \dots & \vzero &\ldots & e_{n,1} & \dots & e_{n, n} \\ 
\hdashline[2pt/2pt]
-e_{1,1} & \vzero & \vzero & \dots & \vzero & \vzero & \dots & \vzero & \vzero & 
\dots & \vzero & \ldots  & \vzero & \dots & \vzero \\
\vdots & \vdots & \vdots & \dots & \vdots & \vdots & \dots & \vdots & \vdots & 
\dots & \vdots & \ldots  & \vdots & \dots & \vdots \\
-e_{1,n} & \vzero & \vzero & \dots & \vzero & \vzero & \dots & \vzero & \vzero & 
\dots & \vzero & \ldots  & \vzero & \dots & \vzero \\ \hdashline[2pt/2pt]
\vzero & -e_{2,1} & \vzero & \dots & \vzero & \vzero & \dots & \vzero & \vzero & 
\dots & \vzero & \ldots  & \vzero & \dots & \vzero \\
\vdots & \vdots & \vdots & \dots & \vdots & \vdots & \dots & \vdots & \vdots & 
\dots & \vdots & \ldots  & \vdots & \dots & \vdots  \\
\vzero & -e_{2,n} & \vzero & \dots & \vzero & \vzero & \dots & \vzero & \vzero & 
\dots & \vzero & \ldots  & \vzero & \dots & \vzero \\ \hdashline[2pt/2pt]
\vdots & \vdots & \vdots & \dots & \vdots & \vdots & \dots & \vdots & \vdots & 
\dots & \vdots & \ldots  & \vdots & \dots & \vdots  \\ \hdashline[2pt/2pt]
\vzero & \vzero & \vzero & \dots & -e_{n,1} & \vzero & \dots & \vzero & \vzero & 
\dots & \vzero & \ldots  & \vzero & \dots & \vzero \\
\vdots & \vdots & \vdots & \dots & \vdots & \vdots & \dots & \vdots & \vdots & 
\dots & \vdots & \ldots  & \vdots & \dots & \vdots \\
\vzero & \vzero & \vzero & \dots & -e_{n,n} & \vzero & \dots & \vzero & \vzero & 
\dots & \vzero & \ldots  & \vzero & \dots & \vzero \\
\end{array}\right],
\end{equation}
where $\tilde a_{i,j}=\begin{bmatrix} a_{i,j} \\ \vzero\end{bmatrix}\in 
\F^{m+n^2}$ (here think of the vector $a_{i,j}$ as a column vector, \emph{not} 
coming out of the page; in the above array we then lay the column vector $\tilde 
a_{i,j}$ ``on its side'' so that it is coming out of the page), and 
$e_{i,j}:=e_{m+(i-1)n+j}\in \F^{m+n^2}$, which we can equivalently write as 
$\begin{bmatrix} \vzero_m \\ e_i \otimes e_j\end{bmatrix}$, where we think of $e_i 
\otimes e_j$ here as a vector of length $n^2$. Note that all the the nonzero 
blocks besides upper-left ``$\tA$'' block only have nonzero entries that are 
strictly \emph{behind} the nonzero entries in the upper-left block. 

\begin{figure}[!htbp]
\[
\xymatrix@R=8pt@C=6pt{
*{}\ar@{-}'[rrrrr] 
\ar@{-}'[dddd]  
\ar@{}'[rrrrrdddd]|{A}  
& & & & &*{} \ar@{-}'[dddd]
\ar@{-}'[rd] 
& \\ 
 & *{} 
 & & & & & *{} \ar@{-}'[dddd] 
 \ar@{-}'[rrr] & & 
 & *{} \ar@{-}'[rd] 
 \ar@{-}'[d] \\ 
& & & & & & *{} \ar@{-}'[rrr] 
\ar@{-}'[rd] 
\ar@{.}'[rrrrd]|{I_n} 
& & & *{} \ar@{-}'[rd] 
& *{} \ar@{-}'[d] \\ 
& & & & & & *{} \ar@{--}'[rd]'[rdrrr]
& *{} \ar@{-}'[rrr] \ar@{--}'[d] & & & *{}  
\ar@{-}'[rrr] 
\ar@{-}'[d] 
& & & *{} \ar@{-}'[d] 
\ar@{-}'[rd] 
\\
*{} \ar@{-}'[rrrrr]  
\ar@{-}'[rd] 
& & & & &*{} \ar@{-}'[rd] & 
& *{} & & & *{} \ar@{-}'[rrr] 
\ar@{-}'[rd] 
\ar@{.}'[rrrrd]|{I_n} 
& & & *{} \ar@{-}'[rd] 
& *{} \ar@{-}'[d] \\ 
& *{} \ar@{-}'[rrrrr] 
\ar@{-}'[dd]'[dddr]'[dddrr]'[ddr]'[r]'[rdr]'[rdrdd] 
& *{} 
\ar@{.}'[dddr]|{\text{-}I_n} 
& *{} \ar@{--}'[rd]'[rddd]
& & & *{} 
& & & & &*{} \ar@{-}'[rrr] 
& & & *{} \ar@{.}'[rrrd] \\ 
& & & *{} \ar@{--}'[r] & *{} & & & & & & & & & & & & & & & & & & & & & & & \\
& *{} \ar@{-}'[r] & *{} & & & & & & & & & & & & & & & & & & & & & & & & & \\ 
& & *{} 
& *{} \ar@{-}'[r] \ar@{-}'[dd]'[dddr]'[dddrr]'[ddr]'[r]'[rdr]'[rdrdd] & *{} \ar@{.}'[dddr]|{\text{-}I_n}& & & & & & & & & & & & & & & & & & & & & & & \\  
& & & & & *{} & & & & & & & & & & & & & & & & & & & & & & \\ 
& & & *{} \ar@{-}'[r] & *{} & & & & & & & & & & & & & & & & & & & & & & & \\ 
& & & & *{} & *{} \ar@{.}'[dddrr] & & & & & & & & & & & & & & & & & & & & & & \\ 
& & & & & & & & & & & & & & & & & & & & & & & & & & & \\ 
& & & & & & & & & & & & & & & & & & & & & & & & & & & \\ 
& & & & & & & & & & & & & & & & & & & & & & & & & & & 
}
\]
\caption{ \label{fig:gadget} Pictorial representation of the reduction for Lemma~\ref{lem:gadget}.}
\end{figure}
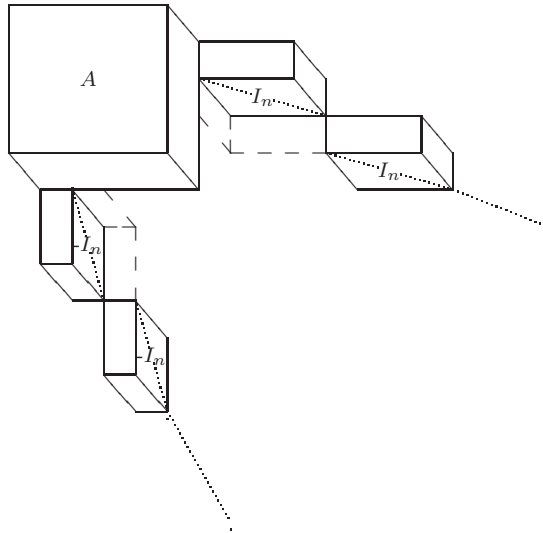

The second viewpoint, which we will also use below, is to consider the lateral 
slices of $\tA$, or equivalently, to view $\tA$ from the side. When viewing $\tA$ 
from the side, we see the $(n+1)n \times (m + n^2) \times (n+1)n$  3-way array: 
\begin{equation} \label{eq:alat}
\tA^{lat}=
\left[\begin{array}{cccc;{2pt/2pt}ccc;{2pt/2pt}c;{2pt/2pt}ccc}
\ell_{1,1} & \ell_{1,2} & \dots & \ell_{1,m} & e_{n+1} & \dots & e_{2n} & 
\dots & 0 & \dots & 0\\
\vdots & \ddots & \ddots & \vdots & \vdots & \ddots & \vdots & \ddots & \vdots & 
\ddots & \vdots \\
\ell_{n, 1} & \ell_{n,2} & \dots & \ell_{n,m} & 0 & \dots & 0 & \dots & 
e_{n^2+1} & 
\dots & e_{n^2+n}\\ \hdashline[2pt/2pt]
0 & 0 & \dots  & 0 & e_{1} & \dots & 0 & \dots & 0 & \dots & 0 \\
\vdots & \vdots & \ddots & \vdots & \vdots & \ddots & \vdots& \dots  & \vdots & 
\ddots & 
\vdots \\
0 & 0 & \dots & 0 & 0 & \dots & e_{1} & \dots & 0 & \dots & 0\\ \hdashline[2pt/2pt]
\vdots & \ddots & \ddots & \vdots & \vdots & \ddots & \vdots & \ddots & \vdots & 
\ddots & \vdots \\ \hdashline[2pt/2pt]
0 & 0 & \dots  & 0 & 0 & \dots & 0 & \dots & e_{n} & \dots & 0 \\
\vdots & \vdots & \ddots & \vdots & \vdots & \ddots & \vdots& \dots  & \vdots & 
\ddots & 
\vdots \\ 
0 & 0 & \dots & 0 & 0 & \dots & 0 & \dots & 0 & \dots & e_{n}\\
\end{array}\right],
\end{equation}
where every $\ell_{i,k}\in \F^{n^2+n}$ has only the first $n$ components 
being possibly non-zero, namely, $\ell_{i,k}(j) = A_k(i,j)$ for $i \in [n], j \in [n], k \in [m]$ and $\ell_{i,k}(j) = 0$ for any $j > n$.

\paragraph{For the only if direction,} suppose there exist $P\in \Mon(n, 
\F)$ and 
$Q\in \GL(m, \F)$, 
such that $P^t\vA P=\vB^Q$. We can construct $\tilde P\in \Mon(n+n^2, \F)$ 
and $\tilde 
Q\in \GL(m+n^2, \F)$ such that $\tilde P^t\tilde\vA\tilde P=\tilde\vB^{\tilde Q}$. In 
fact, we will show that we can take $\tilde P=\begin{bmatrix} P & \vzero \\ \vzero & 
P'\end{bmatrix}$ where $P'\in \Mon(n^2, \F)$, and $\tilde 
Q=\begin{bmatrix} Q & \vzero 
\\ \vzero & Q'\end{bmatrix}$ where $Q'\in \Mon(n^2, \F)$. It is not hard 
to see 
that this form already ensures that the first $m$ matrices in the vector $\tilde 
P^t\tilde\vA\tilde P$ and those of $\tilde \vB^{\tilde Q}$ are the same, 
since when $\tilde P, \tilde Q$ are of this form, those first $m$ matrices are 
controlled entirely by the $P$ (resp., $Q$) in the upper-left block of $\tilde P$ 
(resp., $\tilde Q$).

The remaining question is then how to design appropriate $P'$ and $Q'$ to take 
care of the last $n^2$ matrices in these tuples. This actually boils down to applying
the following simple identity, but ``in 3 dimensions:'' Let $P$ be the permutation matrix 
corresponding to 
$\sigma\in \S_n$, so that $Pe_i=e_{\sigma(i)}$, and $e_i^tP=e_{\sigma^{-1}(i)}^t$. 
Let $D=\diag(\alpha_1, \dots, \alpha_n)$ be a diagonal matrix. Then 
\begin{equation}\label{eq:simple}
P^tDP=\diag(\alpha_{\sigma^{-1}(1)}, \dots, \alpha_{\sigma^{-1}(n)}).
\end{equation}

To see how Equation~\ref{eq:simple} helps in our 
setting, it is easier to focus attention on the lower right $n^2\times n^2$ sub-array of 
$\tA^{lat}$, which can be represented as a symbolic matrix
$$
M=\begin{bmatrix}
x_1 I_n & \vzero & \dots & \vzero \\
\vzero & x_2 I_n & \dots & \vzero \\
\vdots & \ddots & \ddots & \vdots \\
\vzero & \vzero & \dots & x_n I_n
\end{bmatrix}.
$$
Here we think of the $x_i$'s as independent variables, whose indices correspond to ``how far into the page'' they are. That is, $x_i$ corresponds to the vector $\vec{e}_i$ in $\tA^{lat}$, which is coming out of the page and has its only nonzero entry $i$ slices back from the page. 

Then the action of $P$ permutes the $x_i$'s and multiplies them by some scalars, 
the action of $P'$ is on the left-hand 
side, and the action of $Q'$ is on the right-hand side. Let $\sigma$ be 
the permutation supporting $P$. Then $P$ sends $M$ to 
$$
M^P=\begin{bmatrix}
\alpha_{\sigma(1)}x_{\sigma(1)} I_n & \vzero & \dots & \vzero \\
\vzero & \alpha_{\sigma(2)}x_{\sigma(2)} I_n & \dots & \vzero \\
\vdots & \ddots & \ddots & \vdots \\
\vzero & \vzero & \dots & \alpha_{\sigma(n)}x_{\sigma(n)} I_n
\end{bmatrix}.
$$
So setting $P'=\sigma\otimes I_n$, $Q'$ the monomial matrix supported by 
$\sigma\otimes I_n$ with scalars being $1/\alpha_i$'s, we have $P'^tM^P 
Q'=M$ by 
Equation~\ref{eq:simple}.

\paragraph{For the if direction,} suppose there exist $\tilde P\in 
\GL(n+n^2, \F)$ and $\tilde 
Q\in \GL(m+n^2, \F)$, such that $\tilde P^t\tilde\vA\tilde P =\tilde\vB^{\tilde 
Q}$. The key feature of these gadgets now comes into play: consider the lateral 
slices of $\tilde\tA$, which are the frontal slices of $\tA^{lat}$ (which may be 
easier to visualize by looking at Equation~\ref{eq:alat} and Figure~\ref{fig:gadget}). 
The first $n$ lateral slices of 
$\tilde\tA$ and $\tilde\tB$ are of rank $\geq n$ and $<2n$, while the other 
lateral slices are of rank $<n$ (in fact, they are of rank 1; note that without loss of generality we may assume $n > 1$, for the only $1 \times 
1$ alternating matrix space is the zero space). 
Furthermore, left multiplying a lateral slice by $\tilde P^t$ and right
multiplying it by $\tilde Q$ does not change its rank. However, the action of $\tilde P$ here is by $\tilde P^t \tilde \vA \tilde P$, and while the $\tilde P^t$ here corresponds to left multiplication on the lateral slices (=frontal slices of $\tA^{lat}$), the $\tilde P$ on the right here corresponds to taking linear combinations of the lateral slices. In other words, just as $\tA^{lat}$ is the ``side view'' of $\tilde \tA$, $(\tilde P^t \tA^{lat} \tilde Q)^{\tilde P}$ is the side view of $(\tilde P^t \tilde \tA \tilde P)^{\tilde Q}$. Taking linear combinations of the lateral slices could, in principle, alter their rank; we will use the latter possibility to show that $\tilde P$ must be of a constrained form.

Write $\tilde P = \begin{bmatrix}
P_{1,1} & P_{1,2} \\
P_{2,1} & P_{2,2}
\end{bmatrix}$ where $P_{1,1}$ is of size $n\times n$. We first claim that 
$P_{1,2}=\vzero$. For if not, then in $(\tA^{lat})^{\tilde P}$ (the side view), 
one of the last $n^2$ frontal slices receives a nonzero contribution from one of 
the first $n$ frontal slices of $\tA^{lat}$. Looking at the form of these slices 
from Equation~\ref{eq:alat}, we see that any such nonzero combination will have 
rank $\geq n$, but this is a contradiction since the corresponding slice in 
$\tB^{lat}$ has rank $1$. Thus $P_{1,2}=\vzero$, and therefore $P_{1,1}$ must be 
invertible, since $\tilde P$ is.

Finally, we claim that $P_{1,1}$ has to be a monomial matrix. If not, then some frontal slice of $(\tA^{lat})^{\tilde P}$ among the first $n$ would have a contribution from more than one of these $n$ slices. Considering the lower-right $n^2 \times n^2$ sub-matrix of such a slice, we see that it would have rank exactly $kn$ for some $k \geq 2$, which is again a contradiction since the first $n$ slices of $\tB^{lat}$ all have rank $< 2n$. It follows 
that 
$P_{1,1}^t A_i P_{1,1}$, $i\in[m]$, are in $\cB$, and thus
$\cA$ and $\cB$ are monomially isometric via $P_{1,1}$. 
\end{proof}

\subsection{From \ThreeTIlong to \MatSpIsomlong and \algprobm{Matrix Group Isomorphism}}\label{sec:3tensor_alternating}

\begin{proposition} \label{prop:3-tensor_isometry}
\ThreeTIlong reduces to 
\AltMatSpIsomlong. Symbolically, isomorphism in $U \otimes V \otimes W$ reduces to 
isomorphism in $V' \otimes V' \otimes W'$ (or even to $\bigwedge^2 V' \otimes W$), 
where $\ell = \dim U \leq n = \dim V$ and $m = \dim W$, $\dim V' = \ell + 7n + 3$ 
and $\dim W' = m+\ell(2n+1)+n(4n+2)$.
\end{proposition}
\begin{proof}
We will exhibit a function $r$ from 3-way arrays to matrix tuples such that two 
3-way arrays $\tA,\tB \in T(\ell \times n \times m, \F)$ which are non-degenerate 
as 3-tensors, are isomorphic as 3-tensors if and only if the matrix spaces 
$\langle r(\tA) \rangle, \langle r(\tB)\rangle $ are isometric. Note that we can 
assume our input tensors are non-degenerate by Observation~\ref{obs:nondeg}.
The construction is a bit involved, so we will first 
describe the construction in detail, and then prove the desired statement. 

\paragraph{The gadget construction.} 
Given a 3-way array $\tA \in T(\ell \times n \times m, \F)$, let $\vA$ denote the corresponding $m$-tuple of matrices, $\vA \in M(\ell \times n)^m$. The first step is to construct $s(\tA) \in \Lambda(\ell+n, \F)^m$, defined by $s(\tA) = (A_1^{\Lambda}, \dotsc, A_m^{\Lambda})$ where $A_i^{\Lambda}=\begin{bmatrix} 
\vzero & A_i \\-A_i^t & \vzero \end{bmatrix}$. Already, note that if $\tA \cong \tB$, then $s(\tA)$ and $s(\tB)$ are pseudo-isometric matrix tuples (equivalently, $\langle s(\tA) \rangle$ and $\langle s(\tB) \rangle$ are isometric matrix spaces). 

However, it is not clear whether the converse should hold. 
Indeed, suppose $P s(\tA) P^T = s(\tB)^{Q}$ for some $P \in \GL(\ell + n, \F), Q 
\in \GL(m,\F)$. If we write $P$ as a block matrix $\begin{bmatrix} P_{11} & P_{12} 
\\ P_{21} & P_{22} \end{bmatrix}$, where $P_{11} \in M(\ell, \F)$ and $P_{22} \in 
M(n,\F)$, then by considering the (1,2) block we get that $P_{11} A_i 
P_{22}^t - P_{21}^t A_i^t P_{12} = \sum_{j=1}^m q_{ij} B_j$ for all 
$i=1,\dotsc,m$, whereas what we would want is the same equation but without the $P_{21}^t A_i^t 
P_{12}$ term. 
To remedy this, it would suffice if we could extend the tuple $s(\tA)$ 
to $r(\tA)$ so that any pseudo-isometry $(P,Q)$ between $r(\tA)$ and $r(\tB)$ will 
have $P_{21} = 0$. 

To achieve this, we start from $s(\tA)=\vA^{\Lambda} \in\Lambda(n+\ell, \F)^m$, and construct 
$r(\tA) \in \Lambda(\ell+7n+3, \F)^{m+\ell(2n+1)+n(4n+2)}$ as follows. 
Here we write it out symbolically, on the next page is the same thing in matrix 
format, and in Figure~\ref{fig:3-tensor_isometry} is a picture of the construction.
Let 
$s=m+\ell(2n+1)+n(4n+2)$. Write $r(\tA)=(\tilde A_1, \dots, \tilde A_s)$, 
where 
$\tilde A_i\in\Lambda(\ell+7n+3, \F)$ are defined as follows:
\begin{itemize}
\item For $1\leq i\leq m$, $\tilde A_i=\begin{bmatrix} A_i^{\Lambda} & \vzero \\ \vzero & 
\vzero \end{bmatrix}$. Recall that $A_i^{\Lambda} \in \Lambda(\ell+n, \F)$.
\item For the next $\ell(2n+1)$ slices, that is, 
$m+1\leq i\leq m+\ell(2n+1)$, we can naturally represent $i-m$ by $(p, 
q)$ where 
$p\in[\ell]$, $q\in [2n+1]$. We then let $\tilde A_i$ be the elementary 
alternating 
matrix $E_{p,\ell+n+q} - E_{\ell+n+q,p}$.

\item For the next $n(4n+2)$ slices, that is
$m+\ell(2n+1)+1\leq i\leq m+\ell(n+1)+n(4n+2)$, we can naturally 
represent $i-m-\ell(n+1)$ by $(p, q)$ where $p\in[n]$, $q\in[4n+2]$. We then let 
$\tilde A_i$ be the elementary alternating matrix  $E_{\ell+p, n+\ell+2n+1+q} - E_{n+\ell+2n+1+q, \ell+p}$. 
\end{itemize}

We may view the above construction is as follows. 
Write the frontal view of $\tA$ as 
\[
\tA = \left[ \begin{array}{ccc}
a_{1,1}' & \dots & a_{1,n}'\\
\vdots & \ddots & \vdots \\
a_{\ell, 1}' & \dots & a_{\ell, n}'\\ 
\end{array}\right],
\]
where $a_{i,j}'\in \F^m$, 
which we think of as a column vector, but when place in the above array, we think 
of it as coming out of the page.

Let $\tilde\tA$ be the 
$3$-way array whose frontal slices are $\tilde A_i$, so 
$\tilde\tA\in\T((\ell+7n+3)\times
(\ell+7n+3)\times(m+\ell(2n+1)+n(4n+2)), \F)$. Then the frontal view of 
$\tilde\tA$ is 
$$
\tilde\tA=\left[ \begin{array}{ccc;{2pt/2pt}ccc;{2pt/2pt}ccc;{2pt/2pt}ccc}
\vzero & \dots  & \vzero & a_{1,1} & \dots & a_{1,n} & e_{1,1} & \dots & e_{2n+1, 
1} & \vzero & \dots & \vzero\\
\vdots & \ddots & \vdots & \vdots & \ddots & \vdots & \vdots & \ddots & \vdots & 
\vdots & \ddots & \vdots \\
\vzero & \dots  & \vzero & a_{\ell, 1} & \dots & a_{\ell, n} & e_{1, \ell} & 
\dots & e_{2n+1, \ell} & \vzero & \dots & \vzero \\ \hdashline[2pt/2pt]
-a_{1,1} & \dots & -a_{\ell, 1} & \vzero & \dots & \vzero & \vzero & \dots & 
\vzero & f_{1,1} & \dots & f_{4n+2, 1} \\
\vdots & \ddots & \vdots & \vdots & \ddots & \vdots & \vdots & \ddots & \vdots & 
\vdots & \ddots & \vdots \\
-a_{1,n} & \dots & -a_{\ell, n} & \vzero & \dots & \vzero & \vzero & \dots & 
\vzero & f_{1,n} & \dots & f_{4n+2, n} \\  \hdashline[2pt/2pt]
-e_{1,1} & \dots & -e_{1, \ell} & \vzero & \dots & \vzero & \vzero & \dots & 
\vzero 
& \vzero & \dots & \vzero \\
\vdots & \ddots & \vdots & \vdots & \ddots & \vdots & \vdots & \ddots & \vdots & 
\vdots & \ddots & \vdots \\
-e_{2n+1, 1} & \dots & -e_{2n+1, \ell} & \vzero & \dots & \vzero & \vzero & \dots 
& 
\vzero & \vzero & \dots & \vzero \\  \hdashline[2pt/2pt]
\vzero & \dots & \vzero & -f_{1,1} & \dots & -f_{1, n} & \vzero & \dots & \vzero & 
\vzero & \dots & \vzero \\
\vdots & \ddots & \vdots & \vdots & \ddots & \vdots & \vdots & \ddots & \vdots & 
\vdots & \ddots & \vdots \\
\vzero & \dots & \vzero & -f_{4n+2, 1} & \dots & -f_{4n+2, n} & \vzero & \dots & 
\vzero & \vzero & \dots & \vzero 
\end{array}\right], 
$$
where $a_{i,j}=\begin{bmatrix} a_{i,j}'\\ \vzero \end{bmatrix} 
\in\F^{m+\ell(2n+1)+n(4n+2)}$, $e_{i,j}=\vec{e}_{m+(j-1)(2n+1)+i}$, and 
$f_{i,j}=\vec{e}_{m+\ell(2n+1)+(j-1)(4n+2)+i}$. 

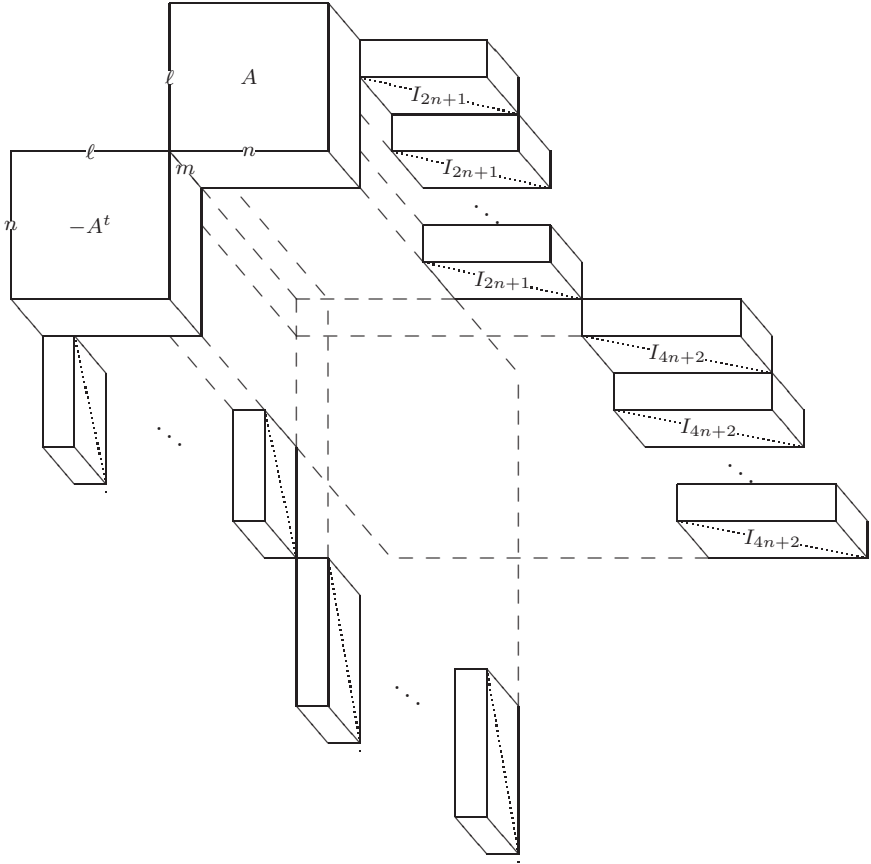
\begin{figure}[!htbp]
\[
\xymatrix@R=8pt@C=6pt{
&&&&& *{} \ar@{-}'[rrrrr] \ar@{-}'[dddd]|{\ell}'[dddddr]|{m}'[dddddrrrrrr]'[ddddrrrrr]'[rrrrr]'[rrrrrdr]'[rrrrrdrdddd] \ar@{}'[ddddrrrrr]|{A} &&&&& *{} & 
 \\ 
&&&&&&&&&&& *{} \ar@{-}'[rrrr]'[rrrrdr]'[rrrrdrd]'[rrrrd]'[d]'[drd]'[drdrrrr] &&&& 
*{} \ar@{-}'[d] &&& 
\\ 
&&&&&&&&&&& *{} \ar@{.}'[rdrrrr]|{I_{2n+1}} &&&& *{} & 
 *{} 
 \\ 
&&&&&&&&&&& *{} \ar@{--}'[rd]  
& *{} \ar@{-}'[d]'[drd]'[drdrrrr]'[drrrrr]'[rrrr]'[rrrrd]'[rrrrddr] &&&& *{}  
\\ 
*{} \ar@{-}'[rrrrr]|{\ell} \ar@{-}'[dddd]|{n} \ar@{}'[ddddrrrrr]|{-A^t}
&&&&& *{} \ar@{-}'[rrrrr]|{n} \ar@{-}'[dddd] &&&&& *{}  
&*{} \ar@{--}'[rdrd]
& *{} \ar@{-}'[rrrr] \ar@{.}'[rrrrrd]|{I_{2n+1}} &&&& *{}  
& *{} 
\\ 
&&&&&& *{} \ar@{-}'[dddd] 
\ar@{--}'[rdrdrd]'[rdrdrdrrrrr]'[rdrdrdrdrrrrrrd][rdrdrdrdrrrrrrdddddddddd] 
& \ar@{--}'[rdrdrd]'[rdrdrdddddddd] 
&&&& *{} \ar@{--}'[rdrd] 
&& *{} \ar@{}'[drrrr]|{\ddots} &&&& *{} 
\\ 
&&&&&& *{} 
 \ar@{--}'[rdrdrd]'[rrrdddrrrrrrrrr] 
&&&&&&& *{} &&&& *{} 
\\ 
&&&&&&&&&&&&& *{} \ar@{-}'[rrrr] 
\ar@{-}'[u]'[urrrr]'[rrrr]'[rrrrrd]'[rrrrr]'[urrrr] 
\ar@{-}'[rd]'[rdrrrr] 
\ar@{.}'[rdrrrr]|{I_{2n+1}} 
& & & & *{} & 
*{} 
\\ 
*{} \ar@{-}'[rrrrr] \ar@{-}'[rd]'[rdrrrrr] 
& & & 
& &  *{} \ar@{-}'[rd] 
& & & & *{} \ar@{--}'[dddd] 
& & && &  *{} & & & &  *{} 
\ar@{-}'[rrrrr]'[rrrrrd]'[d]'[ddr]'[ddrrrrrr]'[rrrrrd]'[rrrrr]'[rrrrrdr]'[rrrrrdrd] 
\ar@{-}'[d] 
& & & & & *{} 
\\ 
& *{} \ar@{-}'[ddd]'[dddrd]'[dddrdr]'[dddr]'[r]'[rrd]'[rrdddd] 
& *{} \ar@{.}'[rdddd] 
& & & *{}  \ar@{--}'[rdrd] 
& *{} 
\ar@{--}'[rdrd] 
&&& *{} 
&&&&&&&&& *{} \ar@{.}'[rdrrrrr]|{I_{4n+2}} &&&&& *{}  
& *{} 
\\ 
&&& *{} \ar@{}'[rrrrddd]|{\ddots} 
&&&&&&&&&&&&&&&& *{}  \ar@{-}'[d]'[drrrrr]'[rrrrr]'[rrrrrdr]'[rrrrrdrd]'[drd]'[d] 
&&&&&*{} 
\\ 
&&&&&&&*{} \ar@{-}'[ddd]'[dddrd]'[dddrdr]'[dddr]'[r]'[rrd]'[rrdddd] \ar@{-}'[r] 
& *{} \ar@{.}'[rdddd] 
&&&&&&&&&&& *{} \ar@{.}'[rrrrrdr]|{I_{4n+2}} 
&&&&&*{} \ar@{-}'[rd] 
&*{} 
\\ 
&*{} \ar@{-}'[r] &*{} 
&&&&&&&*{} 
\ar@{--}'[rdrdrd]'[rdrdrdrrrrrrrrrr] 
&&&&&&&&&&&*{} \ar@{}'[drrrrrr]|{\ddots} 
&&&&&*{} 
\\ 
&&*{}&*{} 
&&&&&&&&&&&&&&&&&&*{} \ar@{-}'[rrrrr]'[rrrrrd]'[d]'[ddr]'[ddrrrrrr]'[rrrrrd]'[rrrrr]'[rrrrrdr]'[rrrrrdrd]  \ar@{-}'[d] 
&&&&&*{} 
\\ 
&&&&&&&*{} \ar@{-}'[r] & *{} 
&&&&&&&&&&&&&*{} \ar@{.}'[rdrrrrr]|{I_{4n+2}}&&&&&*{}
& *{} 
\\ 
&&&&&&&&*{} 
&*{}  \ar@{-}'[dddd]'[ddddrd]'[ddddrdr]'[ddddr]'[r]'[rrd]'[rrddddd] \ar@{-}'[r]  
&*{} \ar@{.}'[rddddd] 
&&&&&&&&&&&&*{} &&&&&*{} 
\\ 
&&&&&&&&&&&*{} 
&&&&&&&&&&&&&&&&&&&&&&&&&&&&&&&&&&&
 \\ 
&&&&&&&&&&& \ar@{}'[rdrdrd]|{\ddots} 
&&&&&&&&&&&&&&&&&&&&&&&&&&&&&&&&&&&&&& 
\\ 
&&&&&&&&&&&&&&*{} \ar@{-}'[dddd]'[ddddrd]'[ddddrdr]'[ddddr]'[r]'[rrd]'[rrddddd] \ar@{-}'[r]  & *{} \ar@{.}'[rddddd] 
&&&&&&&&&&&&&&&&&&&&&&&&&&&&&&
\\ 
&&&&&&&&&*{} \ar@{-}'[r] & *{} 
&&&&&&*{} 
&&&&&&&&&&&&&&&&&&&&&&&&&&&&& 
\\ 
&&&&&&&&&&*{} & *{} 
&&&&&&&&&&&&&&&&&&&&&&&&&&&&&&&&&&&&&&
\\ 
&&&&&&&&&&&&&*{} &*{} 
&&&&&&&&&&&&&&&&&&&&&&&&&&&&&&&
\\ 
&&&&&&&&&&&& &&*{} \ar@{-}'[r] &*{}  
&&&&&&&&&&&&&&&&&&&&&&&&&&&&&&&
\\ 
&&&&&&&&&&&&&&&*{} &*{} 
&&&&&&&&&&&&&&&&&&&&&&&&&&&&&&&
}
\]
\caption{ \label{fig:3-tensor_isometry} Pictorial representation of the reduction for Proposition~\ref{prop:3-tensor_isometry}.}
\end{figure}

We now examine the ranks of the lateral slices $L_i$ of $\tilde \tA$. We claim:
\[
\begin{array}{rccclrcccl}
&&\text{ For $i$...} && &&& \rk(L_i) \\ \hline
1 & \leq & i & \leq & \ell & 2n+1 & \leq & \rk(L_i) & \leq & 3n+1 \\
\ell+1 & \leq & i & \leq & \ell+n & 4n+2 & \leq & \rk(L_i) & \leq & 5n+2 \\
\ell+n+1 & \leq & i & \leq & \ell+n+6n+3 & &&  \rk(L_i) & \leq & n
\end{array}
\]

To see why these hold:
\begin{itemize}
\item For $1\leq i\leq \ell$, the $i$th lateral slice $L_i$ 
is block-diagonal with two non-zero 
blocks. One block is of size $n\times \ell$, and the other is $-I_{2n+1}$. 
Therefore $2n+1\leq \rk(L_i)\leq 3n+1$.
\item For $\ell+1\leq i\leq \ell+n$, the $i$th lateral slice $L_i$ 
 is also block-diagonal with
 two non-zero blocks. One block is of size $\ell\times n$, and the other is 
$-I_{4n+2}$. Therefore $4n+2\leq \rk(L_i)\leq 5n+2$.
\item For $\ell+n+1\leq i\leq \ell+n+6n+3$, after rearranging the columns, the 
$i$th lateral slice $L_i$ has one non-zero block which is 
is $I_\ell$ for the first $2n+1$ slices, and $I_n$ for the next $4n+2$ slices.
Therefore $\rk(L_i)=\ell$ or $n$, and since we have assumed $\ell \leq n$, in either case we have $\rk(L_i)\leq n$. 
\end{itemize}

We then consider the ranks of the linear combinations of the lateral slices.
\begin{itemize}
\item As long as the linear combination involves $L_i$ for $\ell+1\leq i\leq 
\ell+n$, then the resulting matrix has rank at least $4n+2$, because of the matrix 
$-I_{4n+2}$ in the last $4n+2$ rows. 
\item If the linear combination does not involve $L_i$ for $\ell+1\leq i\leq 
\ell+n$, then the resulting matrix has rank at most $4n+1$, because in this case, 
there are at most $\ell+n+2n+1\leq 4n+1$ non-zero rows. 
\item If the linear combination involves $L_i$ for $1\leq i\leq \ell$, then the 
resulting matrix has rank at least $2n+1$, because of the matrix $-I_{2n+1}$ in 
the $(\ell+n+1)$th to the $(\ell+3n+1)$th rows. 
\end{itemize}

We then prove that $\tA$ and $\tB$ are isomorphic as 3-tensors if and only if 
$\langle r(\tA) \rangle$ and $\langle r(\tB) \rangle$ 
are isometric as matrix 
spaces. At first glance, the only if direction seems the easy one, as one expects to 
extend a 3-tensor isomorphism between $\tA$ to $\tB$ to an isometry between 
$\langle r(\tA) \rangle$ and $\langle r(\tB) \rangle$ easily. However, it turns out that this direction 
becomes somewhat technical because of the gadget introduced. This is handled in the following.

\paragraph{For the if direction,} suppose  
$P^t\tilde\tA P=\tilde\tB^Q$, for some $P\in 
\GL(\ell+7n+3, \F)$ and $Q\in \GL(m+\ell(2n+1)+n(4n+2), \F)$. Write $P$ as 
$\begin{bmatrix} P_{1,1} & P_{1,2} & P_{1,3} \\ P_{2,1} & P_{2,2} & P_{2,3} \\ 
P_{3,1} & P_{3,2} & P_{3,3}\end{bmatrix}$, where $P_{1,1}$ is of size $\ell\times 
\ell$, $P_{2,2}$ is of size $n\times n$, and $P_{3,3}$ is of size $(6n+3)\times 
(6n+3)$. By the discussion on the ranks of the linear combinations of the 
lateral slices, we have $P_{2,1}=\vzero$, 
$P_{1,2}=\vzero$, $P_{1,3}=\vzero$, and $P_{2,3}=\vzero$. So $P=\begin{bmatrix} 
P_{1,1} & \vzero & \vzero \\ \vzero & P_{2,2} & \vzero \\ P_{3,1} & P_{3,2} & 
P_{3,3}\end{bmatrix}$, where $P_{1,1}$, $P_{2,2}$, $P_{3,3}$ are invertible. 
Then consider the action of such $P$ on the first $m$ frontal slices of $\tilde\tA$. 
The first $m$ frontal slices of $\tilde\tA$ are of the form 
$\begin{bmatrix}
\vzero & A_i & \vzero \\
-A_i^t & \vzero & \vzero \\
\vzero & \vzero & \vzero 
\end{bmatrix}$, where $A_i$ is of size $\ell\times n$. 
Then we have 
$$
\begin{bmatrix} 
P_{1,1}^t & \vzero & P_{3,1}^t \\ 
\vzero & P_{2,2}^t & P_{3,2}^t \\
\vzero & \vzero & P_{3,3}^t
\end{bmatrix}
\begin{bmatrix}
\vzero & A_i & \vzero \\
-A_i^t & \vzero & \vzero \\
\vzero & \vzero & \vzero 
\end{bmatrix}
\begin{bmatrix} 
P_{1,1} & \vzero & \vzero \\ 
\vzero & P_{2,2} & \vzero \\ 
P_{3,1} & P_{3,2} & P_{3,3}
\end{bmatrix}
=
\begin{bmatrix}
\vzero & P_{1,1}^tA_iP_{2,2} & \vzero \\
-P_{2,2}^tA_iP_{1,1} & \vzero & \vzero \\
\vzero & \vzero & \vzero
\end{bmatrix}.
$$
From the fact that $Q$ is invertible and $P^t \tilde\tA P = \tilde \tB^Q$, by 
considering the $(1,2)$ block, we find that 
every frontal slice of 
$P_{11}^t\tA P_{22}$ lies in $\langle \vB\rangle$ 
(since the gadget does not affect the block-(1,2) position), 
which gives an isomorphism of tensors, as desired.

\paragraph{For the only if direction,} suppose $\tA$ and $\tB$ are isomorphic as 3-tensors, 
that is, $P^t\tA 
Q=\tB^R$, for some $P\in\GL(\ell, \F)$, $Q\in \GL(n, \F)$, and $R\in \GL(m, \F)$. 

We show that there exist $U \in \GL(6n+3, \F)$ and 
$V \in \GL(\ell(2n+1) + n(4n+2), \F)$ 
such that setting 
\[
\begin{array}{rclcl}
\tilde Q & =& \diag(P, Q, U) & \in & \GL(\ell+7n+3, \F) \\
\tilde R & = & \diag(R, V) & \in & \GL(m+\ell(2n+1)+n(4n+2),\F),
\end{array}
\]
 we have 
\[
\tilde Q^t r(\tA) \tilde Q=r(\tB)^{\tilde R},
\]
which will demonstrate that $r(\tA)$ and $r(\tB)$ are pseudo-isometric. 

Since we are claiming that $\tilde R = \diag(R,V) \in \GL(m,\F) \times 
\GL(\ell(2n+1) + n(4n+2), \F)$ works, and $\tilde R$ is block-diagonal, it 
suffices to consider the first $m$ frontal slices separately from the remaining 
slices. For the first $m$ frontal slices, we have:
\[
\tilde Q^t \tilde A_i \tilde Q
=
\begin{bmatrix}
P^t & \vzero & \vzero \\
\vzero & Q^t & \vzero \\
\vzero & \vzero & U^t 
\end{bmatrix}
\begin{bmatrix}
\vzero & A_i & \vzero \\
-A_i^t & \vzero & \vzero \\
\vzero & \vzero & \vzero 
\end{bmatrix}
\begin{bmatrix}
P & \vzero & \vzero \\
\vzero & Q & \vzero \\
\vzero & \vzero & U
\end{bmatrix}
=
\begin{bmatrix}
\vzero & P^t A_i Q & \vzero \\
-Q^t A_i^t P&\vzero & \vzero \\
\vzero & \vzero & \vzero 
\end{bmatrix}.
\]
It follows from the fact that $P^t \tA Q = \tB^R$ that the first $m$ frontal slices of $\tilde 
Q^t r(\tA) \tilde Q$ and of $r(\tB)^{\tilde R}$ are 
the same.

We now consider the remaining frontal slices separately. Towards that end, 
let $\tilde \tA'\in\T((\ell+7n+3)\times(\ell+7n+3)\times (\ell(2n+1)+n(4n+2)), 
\F)$ be the 3-way array obtained by removing the first $m$ frontal slices from 
$\tilde \tA$. That is, the $i$th frontal slice of $\tilde \tA'$ is the $(m+i)$th 
frontal slice of $\tilde \tA$. Similarly construct $\tilde\tB'$ from $\tilde\tB$. 
We are left to show that $\tilde\tA'$ and $\tilde\tB'$ are pseudo-isometric under 
some $\tilde Q=\diag(P, Q, U)$ and $V$. Note that $P$ and 
$Q$ are from the isomorphism between $\tA$ and $\tB$, while $U$ and $V$ are what we still need to design.

We first note that both $\tilde\tA'$ and $\tilde \tB'$ can be viewed as a block 3-way array of size 
$4\times 4\times 2$, whose two frontal slices are the block matrices 
\[\begin{bmatrix}
\tzero & \tzero & \tE & \tzero \\
\tzero & \tzero & \tzero & \tzero \\
-\tE & \tzero & \tzero & \tzero \\
\tzero & \tzero & \tzero & \tzero 
\end{bmatrix} 
\qquad\text{ and }\qquad 
\begin{bmatrix}
\tzero & \tzero & \tzero & \tzero \\
\tzero & \tzero & \tzero & \tF \\
\tzero & \tzero & \tzero & \tzero \\
\tzero & -\tF & \tzero & \tzero 
\end{bmatrix},
\]
where $\tE$ is of size $\ell\times (2n+1)\times \ell(2n+1)$, and $\tF$ is of size 
$n\times (4n+2)\times n(4n+2)$. 
Although these are already identical in $\tA',\tB'$, the issue here is that $P$ 
and $Q$ may alter the slices of $\tilde \tA'$ when they act on $\tA$, so we need a 
way to ``undo'' this action to bring it back to the same slices in $\tB'$.

We now claim that we may further handle these two block slices---the ``$E$'' slices and the ``$F$''-slices---separately, that is, that we may take $U = \diag(U_1, U_2)$ and $V = \diag(V_1, V_2)$ where $U_1\in\GL(2n+1, \F)$, $U_2\in \GL(4n+2, \F)$, $V_1\in\GL(\ell(2n+1), \F)$, and $V_2\in \GL(n(4n+2), \F)$. 

To handle $\tE$, first note that we have
\[
\begin{bmatrix}
P^t \\
 & R^t \\
  & & U_1^t \\
  & & & U_2^t
\end{bmatrix}
\begin{bmatrix}
\vzero & \vzero & E & \vzero \\
\vzero & \vzero & \vzero & \vzero \\
-E^t & \vzero & \vzero & \vzero \\
\vzero & \vzero & \vzero & \vzero
\end{bmatrix}
\begin{bmatrix}
P \\
 & R \\
  & & U_1 \\
  & & & U_2
\end{bmatrix}
=
\begin{bmatrix}
\vzero & \vzero & P^t E U_1 & \vzero \\
\vzero & \vzero & \vzero & \vzero \\
-U_1^t E^t P & \vzero & \vzero & \vzero \\
\vzero & \vzero & \vzero & \vzero 
\end{bmatrix},
\]
where $E \in M(\ell \times(2n+1), \F)$.

Now we examine the lateral slices of $\tE$. 
The $i$th lateral 
slice of $\tE$ 
(up to a suitable permutation)
is 
$$L_i=\begin{bmatrix}\vzero & \dots & \vzero & I_{\ell} & \vzero & \dots & 
\vzero\end{bmatrix},$$ where each $\vzero$ is of size $\ell\times \ell$, 
$I_{\ell}$ is the 
$i$th block, and there are $2n+1$ block matrices in total. The action of $P$ on 
$L_i$ is by left multiplication. So it sends 
$L_i$ to $P^t L_i = \begin{bmatrix}\vzero & \dots & \vzero & P^t & \vzero & \dots & 
\vzero\end{bmatrix}$. 
If we set $U_1$ to be the identity and $V_1 = \diag(P^t, \dotsc, P^t)$, where there are $(2n+1)$ copies of $P^t$ on the diagonal, then we have $L_i V_1 = P^t L_i$, and thus $P^t\tE 
U_1=\tE^{V_1}$. 

It is easy to check that $\tF$ can be handled in the same way, where now $R, U_2, V_2$ play the roles that $P, U_1, V_1$ played before, respectively. This produces the 
desired $U_1$, $U_2$, $V_1$, and $V_2$, and concludes the proof. 
\end{proof}

\begin{corollary} \label{cor:AltSymMatSpIsom}
\ThreeTIlong reduces to \SymMatSpIsomlong.
\end{corollary}

\begin{proof}
In the proof of Proposition~\ref{prop:3-tensor_isometry}, we can easily replace $A_i^{\Lambda}$ with $A_i^s = \begin{bmatrix}\vzero & A_i \\ A_i^t & \vzero  \end{bmatrix}$, and the elementary alternating matrices with the elementary symmetric matrices, and the resulting proof goes through \emph{mutatis mutandis}. 
\end{proof}

Finally, we show how to reduce to \GpIlong for matrix groups. We begin with a lemma that we also need for the search-to-decision reduction below. We believe this lemma to be classical, but have not found a reference stating it in quite the form we need.

\begin{lemma}[Constructive version of Baer's correspondence for matrix groups] \label{lem:baer_matrix}
Let $p$ be an odd prime. Over the finite field $\F=\F_{p^e}$, \AltMatSpIsomlong is 
equivalent to \GpIlong for matrix groups over $\F$ that are $p$-groups of class $2$ and exponent $p$. More precisely, there are functions computable in time $\poly(n,m,\log|\F|)$:
\begin{itemize}
\item $G\colon \Lambda(n,\F)^m \to \M(n+m+1,\F)^{n+m}$ and
\item $\text{Alt} \colon \M(n,\F)^m \to \Lambda(m, \F)^{O(m^2)}$ 
\end{itemize}
such that: (1) for an alternating bilinear map $\vA$, the group generated by $G(\vA)$ is the Baer group corresponding to $\vA$, (2) $G$ and $\text{Alt}$ are mutually inverse, in the sense that the group generated by $G(\text{Alt}(M_1, \dotsc, M_m))$ is isomorphic to the group generated by $M_1, \dotsc, M_m$, and conversely $\text{Alt}(G(\vA))$ is pseudo-isometric to $\vA$.
\end{lemma}

\begin{proof}
First, let $G$ be a $p$-group of class $2$ and exponent 
$p$ given by $m$ generating matrices of size $n \times n$ over $\F$. Then from 
the generating matrices of $G$, we first compute a generating 
set of $[G,G]$, by just computing all 
the commutators of the given generators. We can then remove those 
redundant elements from this generating set in time 
$\poly(\log|[G,G]|, \log|\F|)$, using Luks' result on computing with 
solvable matrix groups\cite{Luk92}.
 We then compute a set of  
representatives of a non-redundant generating 
set of $G/[G,G]$, again using Luks's aforementioned result. From 
these data we 
can compute an 
alternating bilinear map representing the commutator map of $G$ 
in time $\poly(n,m, \log|F|)$. 

Conversely, let an alternating bilinear map be given by 
$\vA=(A_1, \dots, A_m)\in \Lambda(n, \F)^m$. From $\vA$, for 
$i\in[n]$, construct 
$B_i=[A_1\vec{e_i}, \dots,
A_m\vec{e_i}]\in\M(n\times m, \F)$. That is, the $j$th column of 
$B_i$ 
is the
$i$th
column of $A_j$. Then for $i\in[n]$, construct
$$\tilde B_i
=
\begin{bmatrix}
1 & e_i^t & 0 \\
0 & I_n & B_i \\
0 & 0 &  I_m
\end{bmatrix}\in \GL(1+n+m, \F),
$$
and for $j\in[m]$, construct
$$
\tilde C_j
=
\begin{bmatrix}
1 & 0 & e_j^t \\
0 & I_n & 0 \\
0 & 0 &  I_m
\end{bmatrix}\in \GL(1+n+m, \F).
$$
Let $G(\vA)$ be the matrix group generated by $\tilde B_i$ and 
$\tilde C_j$. Then
it can be verified easily that,
$G(\vA)$ is isomorphic to the Baer group
corresponding to the alternating bilinear map defined by $\vA$. 
In particular, $[G, G]\cong \F^m \cong
\Z_p^{em}$ (isomorphism of abelian groups), and $G/[G, 
G]\cong \F^n \cong \Z_p^{en}$. This construction can be done in time 
$\poly(n, m, \log |\F|)$.
\end{proof}

\begin{corollary} \label{cor:pgp}
Let $p$ be an odd prime. \ThreeTIlong over $\F=\F_{p^e}$ reduces to \GpIlong for $p$-groups of class 2 and exponent $p$ given by matrices over $\F$, in time $\poly(n, \log |\F|)$ (where $n$ is the max of the dimensions of the 3-tensor).
\end{corollary}

\begin{proof}
Combine Proposition~\ref{prop:3-tensor_isometry} with Lemma~\ref{lem:baer_matrix}. Note that for this direction of the reduction, we only need the function $G$ from Lemma~\ref{lem:baer_matrix}, which can be computed in time $\poly(n,\log p)$.
\end{proof}

\subsection{Search to decision reduction for $p$-\GpIlong and \AltMatSpIsomlong}

\begin{thmsearch}
Given an oracle deciding \AltMatSpIsomlong, there is a 
$q^{O(n)}\cdot n!=q^{\tilde{O}(n)}$-time algorithm to find an 
isometry between two alternating matrix spaces $\cA, \cB \in 
\Lambda(n, \F_q)$, if it exists, using at most $q^{O(n)}$ 
oracle queries each of size at most $O(n^2)$. 
\end{thmsearch}

In particular, if \AltMatSpIsomlong can be decided in 
$q^{\tilde{O}(\sqrt{n})}$ time, then isometries between such 
spaces can be found in $q^{\tilde{O}(n)}$ time. See Question~\ref{question:search_decision}.

\begin{proof}
As before, we first present the gadget construction, which is a combination of the 
two gadgets introduced in 
Sections~\ref{sec:graphiso_alternating} 
and~\ref{sec:3tensor_alternating}, respectively. Then based on this gadget, we 
present the search-to-decision reduction. 

\paragraph{Gadget construction.} 
Let $\vA=(A_1, \dots, A_m)$ be an ordered 
linear basis of $\cA$, 
and let  
$\tA \in \M(n\times n\times m, \F_q)$ be the 3-way array constructed from 
$\vA$, so we can write
$$
\tA=\begin{bmatrix}
\vzero & a_{1,2} & a_{1,3} & \dots & a_{1,n} \\
-a_{1,2} & \vzero & a_{2,3} & \dots & a_{2,n} \\
-a_{1,3} & -a_{2,3} & \vzero & \dots & a_{3,n}\\
\vdots & \ddots & \ddots & \ddots & \vdots \\
-a_{1,n} & -a_{2,n} & -a_{3,n} & \dots & \vzero
\end{bmatrix},
$$
where $a_{i,j} \in \F^m$, $1\leq i<j\leq n$ thought of as a vector coming out of the page.

We first consider a $3$-tensor $\tilde\tA_i$ constructed from $\tA$, for any 
$1\leq i\leq n-1$, as $\tilde \tA_i=$
$$
{\footnotesize
\setlength\arraycolsep{1pt}
\left[\begin{array}{cccc;{2pt/2pt}ccc;{2pt/2pt}ccc;{2pt/2pt}ccc;{2pt/2pt}ccc;{2pt/2pt}ccc}
\vzero & a_{1,2} & \dots & a_{1,i} & a_{1, i+1} & \dots & a_{1,n} & -e_{1,1} & 
\dots & -e_{1,2n} & \vzero & \dots & \vzero & \vzero & \dots & \vzero & \vzero & 
\dots & \vzero \\
-a_{1,2} & \vzero & \dots & a_{2,i} & a_{2, i+1} & \dots & a_{2,n} & \vzero & 
\dots & \vzero & -e_{2,1} & \dots & -e_{2,2n} & \vzero & \dots & \vzero & \vzero & 
\dots & \vzero \\
\vdots & \vdots & \ddots & \vdots & \vdots & \ddots & \vdots & \vdots & \ddots & 
\vdots & \vdots & \ddots & \vdots & \vdots & \ddots & \vdots & \vdots & \ddots & 
\vdots \\
-a_{1,i} & -a_{2,i} & \dots & \vzero & a_{i,i+1} & \dots & a_{i, n} & \vzero & 
\dots & \vzero& \vzero & \dots & \vzero & -e_{i,1} & \dots & -e_{i, 2n} & \vzero & 
\dots & \vzero \\ \hdashline[2pt/2pt]
-a_{1,i+1} & -a_{2,i+1} & \dots & -a_{i,i+1} & \vzero & \dots & a_{i+1, n} & 
\vzero & \dots & \vzero& \vzero & \dots & \vzero& \vzero & \dots & \vzero & 
-f_{1,1} & \dots & -f_{1,n} \\
\vdots & \vdots & \ddots & \vdots & \vdots & \ddots & \vdots & \vdots & \ddots & 
\vdots & \vdots & \ddots & \vdots & \vdots & \ddots & \vdots & \vdots & \ddots & 
\vdots \\
-a_{1,n} & -a_{2,n} & \dots & -a_{i,n} & -a_{i+1, n} & \dots & \vzero& \vzero & 
\dots & \vzero& \vzero & \dots & \vzero & \vzero & \dots & \vzero& -f_{n-i, 1} & 
\dots & -f_{n-i, n} \\ \hdashline[2pt/2pt]
e_{1,1} & \vzero & \dots & \vzero& \vzero & \dots & \vzero& \vzero & \dots & 
\vzero& \vzero & \dots & \vzero& \vzero & \dots & \vzero& \vzero & \dots & \vzero\\
\vdots & \vdots & \ddots & \vdots & \vdots & \ddots & \vdots & \vdots & \ddots & 
\vdots & \vdots & \ddots & \vdots & \vdots & \ddots & \vdots & \vdots & \ddots & 
\vdots \\
e_{1,2n} & \vzero & \dots & \vzero& \vzero & \dots & \vzero& \vzero & \dots & 
\vzero& \vzero & \dots & \vzero& \vzero & \dots & \vzero& \vzero & \dots & \vzero\\ \hdashline[2pt/2pt]
\vzero & e_{2,1} & \dots & \vzero & \vzero & \dots & \vzero& \vzero & \dots & 
\vzero& \vzero & \dots & \vzero& \vzero & \dots & \vzero& \vzero & \dots & \vzero\\
\vdots & \vdots & \ddots & \vdots & \vdots & \ddots & \vdots & \vdots & \ddots & 
\vdots & \vdots & \ddots & \vdots & \vdots & \ddots & \vdots & \vdots & \ddots & 
\vdots \\
\vzero & e_{2,2n} & \dots & \vzero & \vzero & \dots & \vzero& \vzero & \dots & 
\vzero& \vzero & \dots & \vzero& \vzero & \dots & \vzero& \vzero & \dots & \vzero\\ \hdashline[2pt/2pt]
\vzero & \vzero & \dots & e_{i,1} & \vzero & \dots & \vzero& \vzero & \dots & 
\vzero& \vzero & \dots & \vzero& \vzero & \dots & \vzero& \vzero & \dots & \vzero\\
\vdots & \vdots & \ddots & \vdots & \vdots & \ddots & \vdots & \vdots & \ddots & 
\vdots & \vdots & \ddots & \vdots & \vdots & \ddots & \vdots & \vdots & \ddots & 
\vdots \\
\vzero & \vzero & \dots & e_{i,2n} & \vzero & \dots & \vzero& \vzero & \dots & 
\vzero& \vzero & \dots & \vzero& \vzero & \dots & \vzero& \vzero & \dots & \vzero\\ \hdashline[2pt/2pt]
\vzero & \vzero & \dots & \vzero & f_{1,1} & \dots & f_{n-i,1}& \vzero & \dots & 
\vzero& \vzero & \dots & \vzero& \vzero & \dots & \vzero& \vzero & \dots & \vzero\\
\vdots & \vdots & \ddots & \vdots & \vdots & \ddots & \vdots & \vdots & \ddots & 
\vdots & \vdots & \ddots & \vdots & \vdots & \ddots & \vdots & \vdots & \ddots & 
\vdots \\
\vzero & \vzero & \dots & \vzero & f_{1,n} & \dots & f_{n-i,n}& \vzero & \dots & 
\vzero& \vzero & \dots & \vzero& \vzero & \dots & \vzero& \vzero & \dots & \vzero
\end{array}\right].
}
$$

Consider the lateral slices of $\tilde\tA_i$. 
\begin{itemize}
\item The first $i$ lateral slices have rank in $[2n, 3n)$. Note that the rank is \emph{strictly} less than $3n$ because some tube fibers (coming out of the page) are $\vzero$ in the upper-left 
$n\times n$ sub-array.
\item The next $n-i$ lateral slices have rank in $[n, 2n)$. 
\item The remaining $2ni+n$ lateral slices have rank in $[1,n)$ (since $i \geq 1$.)
\end{itemize}

By combining the arguments for the two gadgets introduced in 
Sections~\ref{sec:graphiso_alternating} 
and~\ref{sec:3tensor_alternating} respectively, we have the following. From \Sec{sec:3tensor_alternating}, for 
invertible matrices $P$ and $Q$ to satisfy $P^t\tilde\tA_i 
P=\tilde\tB_i^Q$, $P$ has to be of the form
$
\begin{bmatrix}
P_{1,1} & \vzero & \vzero \\
\vzero & P_{2,2} & \vzero \\
P_{3,1} & P_{3,2} & P_{3,3}
\end{bmatrix}$,
where $P_{1,1}$ is of size $i\times i$, $P_{2,2}$ is of size $(n-i)\times (n-i)$, 
and $P_{3,3}$ is of size $(2ni+n)\times (2ni+n)$. Furthermore, from \Sec{sec:graphiso_alternating}, $P_{1,1}$ is a 
monomial matrix. In particular, if such $P$ and $Q$ exist, then it implies that 
$\tA$ and $\tB$ are isometric by a matrix of the form $\begin{bmatrix} P_{1,1} & 
\vzero \\ \vzero & P_{2,2}\end{bmatrix}$ where $P_{1,1}$ is a monomial matrix of 
size $i\times i$. 
Note that the presence of $P_{3,i}$, $i=1,2,3$, does not interfere here, because 
of the argument in the if direction in the proof of 
Proposition~\ref{prop:3-tensor_isometry}.
On the other hand, if $\tA$ and $\tB$ are isometric by a matrix 
of such form, then $\tilde\tA_i$ and $\tilde\tB_i$ are also isometric.

\paragraph{The search-to-decision reduction.} Given these preparations, we now 
present the search-to-decision reduction for 
\AltMatSpIsomlong. Recall that this requires us to use 
the decision oracle $\oracle$ to compute an explicit isometry transformation $P\in 
\GL(n, q)$, if $\cA$ and $\cB$ are indeed isometric. Think of $P$ as sending the 
standard basis $(\vec{e_1}, \dots, \vec{e_n})$ to another basis $(v_1, \dots, 
v_n)$, where $e_i$ and $v_i$ are in $\F_q^n$.

In the first step, we guess $v_1$, the image of $e_1$, and a complement subspace 
of $\langle v_1\rangle$, at the cost of $q^{O(n)}$. For each such guess, 
let $P_1$ be the matrix which sends $e_1 \mapsto v_1$ and sends $\langle e_2, 
\dotsc, e_n \rangle$ to the chosen complementary subspace in some fashion. We 
apply $P_1$ to 
$\tA$, and call the resulting $3$-way array 
$\tA$ in the following. Then construct $\tilde\tA_1$ and $\tilde\tB_1$, and feed 
these two 
instances to the oracle $\oracle$. 
Note that, since $P_{1,1}$ (using notation as above) must be monomial, any 
equivalence between $\tilde \tA_1$ and $\tilde \tB_1$ must preserve our choice of 
$v_1$ up to scale. Thus, clearly, if $\tA$ and $\tB$ are indeed 
isometric and we guess the correct image of $e_1$, then the oracle $\oracle$ will 
return yes (and conversely).

In the second step, we guess $v_2$, the image of $e_2$, and a complement subspace 
of $\langle v_2\rangle$ within $\langle e_2, \dots, e_n\rangle$, at the cost of 
$q^{O(n)}$. 
Note here that the previous step guarantees that there is an isometry respecting 
the direct sum decomposition $\langle v_1 \rangle \oplus \langle e_2, \dotsc, e_n 
\rangle$, so we need only search for a complement of $v_2$ within $\langle e_2, 
\dotsc, e_n \rangle$, and \emph{not} a more general complement of $\langle v_1, 
v_2 \rangle$ in all of $\F_q^n$. This is crucial for the runtime, as at the $n/2$ 
step, the latter strategy would result in searching through $q^{\Theta(n^2)}$ 
possibilities.

For each such guess, we apply the corresponding transformation to 
$\tA$ (and again call the resulting $3$-way array $\tA$). Then construct 
$\tilde\tA_2$ 
and $\tilde\tB_2$, and feed these two instances to the oracle $\oracle$. Clearly, 
if $\cA$ and $\cB$ are indeed isometric and we guess the correct image of $e_2$ 
(and $e_1$ from the previous step), 
then the oracle $\oracle$ will return yes. However, there is a small caveat here, 
namely we may guess some image of $e_2$, such that $\cA$ and $\cB$ are actually 
isometric by some matrix $P$ of the form $\begin{bmatrix} P_{1,1} & \vzero \\ 
\vzero & P_{2,2} \end{bmatrix}$ where $P_{1,1}$ is a monomial matrix of size $2$. 
But this is fine, as it still means that our choices of $\{v_1, v_2\}$ is correct 
as a set up to scaling. So we proceed.

In general, in the $i$th step, we know that $\cA$ and $\cB$ are isometric by some 
$P=\begin{bmatrix} P_{1,1} & \vzero \\ \vzero & P_{2,2} \end{bmatrix}$ where 
$P_{1,1}$ is a monomial matrix of size $(i-1)\times (i-1)$. We guess $v_i$, the 
image of $e_i$ in $\langle e_i, \dots, e_n\rangle$, and a complement subspace of 
$\langle v_i\rangle$ within $\langle e_i, \dots, e_n\rangle$. 
This cost is 
$q^{O(n)}$. For each such guess, we apply the corresponding transformation to 
$\tA$ (and call the resulting $3$-way array $\tA$). Then construct 
$\tilde\tA_i$ 
and $\tilde\tB_i$, and feed these two instances to the oracle $\oracle$. Once we 
guess correctly, we ensure that $\cA$ and $\cB$ are isometric by 
$P=\begin{bmatrix} P_{1,1} & \vzero \\ \vzero & P_{2,2} \end{bmatrix}$ where 
$P_{1,1}$ is a monomial matrix of size $i\times i$.

So after the $(n-1)$th step, we know that $\cA$ and $\cB$ are isometric by a 
monomial transformation. The number of all monomial transformations is by $(q-1)^n\cdot n!\leq q^n\cdot 2^{n\log 
n}=q^{\tilde{O}(n)}$. Therefore we can enumerate all 
monomial transformations and check correspondingly. 

Note that all the instances we feed into the oracle $\oracle$ are of size 
$O(n^2)$. This concludes the proof. 
\end{proof}

\begin{corsearch}
[Search to decision for testing isomorphism of 
a class of 
$p$-groups]
Let $p$ be an odd prime. Given an oracle deciding isomorphism 
of $p$-groups of class 2 and exponent $p$ given by generating 
matrices over $\F_p$
of size $\poly(n)$, there is a $|G|^{O(\log \log |G|)}$-time 
algorithm to find an isomorphism between such groups, using at 
most $\poly(|G|)$ oracle queries each of size at most 
$\poly(n)$.
\end{corsearch}

\begin{proof}
The result follows from 
Theorem~\ref{thm:search_decision} with 
the constructive version of Baer's 
Correspondence in the model of matrix groups over finite 
fields (Lemma~\ref{lem:baer_matrix}). 

In more detail, given Lemma~\ref{lem:baer_matrix} 
we can follow the procedure 
in the proof of Theorem~\ref{thm:search_decision}. For the 
given $p$-groups, we compute their commutator maps. Then 
whenever 
we need to feed the decision oracle, we transform from the 
alternating bilinear map to a generating set of a $p$-group of 
class $2$ and exponent $p$ with this bilinear map as the 
commutator map. After getting the desired pseudo-isometry for 
the alternating bilinear maps, we can easily recover an 
isomorphism between the originally given $p$-groups. This 
concludes the proof.
\end{proof}

\section{Other reductions for the Main Theorem~\ref{thm:main}}\label{sec:reduction_other}

In this section, we present other reductions to finish the proof of 
Theorem~\ref{thm:main}. The reductions here are based on the constructions 
which may be summarized as ``putting the given 3-way array to an appropriate 
corner of a larger 3-way array.'' Such an idea is quite classical in the 
context of matrix problems and wildness \cite{gel-pon}; 
here we use the same idea for problems on 3-way arrays. 

\subsection{From 3-Tensor Isomorphism to Matrix Space Conjugacy}

\begin{proposition}\label{prop:3-tensor_conjugacy}
\ThreeTIlong reduces to \MatSpConjlong. Symbolically,  $U \otimes V \otimes W$  reduces to $V' \otimes V'^* \otimes W$, where $\dim V'  = \dim U + \dim V$. 
\end{proposition}
\begin{proof} 
\textbf{The construction.} 
For a 3-way array $\tA \in \T(\ell\times n\times m, \F)$, let $\vA=(A_1, \dots, A_m)\in \M(\ell\times n, \F)^m$ 
be the matrix tuple consisting of frontal slices of $\tA$. Construct $\tilde \vA=(\tilde A_1, \dots, \tilde A_m)\in \M(\ell+n, \F)^m$ from $\vA$, where $\tilde A_i=\begin{bmatrix}
\vzero & A_i \\
\vzero & \vzero 
\end{bmatrix}$. See Figure~\ref{fig:3-tensor_conjugacy}.

\begin{figure}[!htbp]
\[
\xymatrix@R=8pt@C=6pt{
& *{} \ar@{-}'[rrrrr]'[rrrrrdddd]'[dddd]'[] \ar@{-}'[rd]'[rdrrrrr]
&& *{} \ar@{-}'[dd]'[rrrdd]'[rrrddrd] \ar@{-}'[rd] \ar@{}'[rrrdd]|(0.7){\tA} 
&&& *{} \ar@{-}'[rd]'[rddddd]'[dddd] \\
& &*{} 
&& *{} \ar@{-}'[dd] &&&*{} \\
& &&*{} \ar@{-}'[rd]'[rdrrr] &&&*{} 
&& \\
\tilde \tA = & &&&*{} 
&&& *{} 
\\
& *{} \ar@{-}'[rd]'[rdrrrrr] &&&&& *{} && \\ 
& &*{} \ar@{-}'[uuuu] &&&&&*{} &  
}
\]
\caption{ \label{fig:3-tensor_conjugacy} Pictorial representation of the reduction for Proposition~\ref{prop:3-tensor_conjugacy}.}
\end{figure}
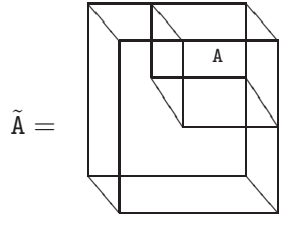

Given two non-degenerate 3-way arrays $\tA,\tB$ which we wish to test for isomorphism (we can assume non-degeneracy without loss of generality, see Observation~\ref{obs:nondeg}), we claim that $\tA \cong \tB$ as 3-tensors if and only if the matrix spaces $\langle 
\tilde\vA\rangle$ and $\langle\tilde\vB\rangle$ are conjugate. 

\textbf{For the only if direction,} since $\tA$ and $\tB$ are isomorphic as 3-tensors, 
there exist $P\in \GL(\ell, \F)$, $Q\in\GL(n, \F)$, and 
$R\in\GL(m, \F)$, such that $P\vA Q=\vB^R=(B_1', \dots, B_m')\in\M(\ell\times n, 
\F)^m$. Let 
$\tilde P=\begin{bmatrix} 
P^{-1} & \vzero \\
\vzero & Q
\end{bmatrix}$. Then $\tilde P^{-1}\tilde A_i \tilde P=
\begin{bmatrix}
P & \vzero \\
\vzero & Q^{-1}
\end{bmatrix}\cdot 
\begin{bmatrix}
\vzero & A_i \\
\vzero & \vzero 
\end{bmatrix}\cdot 
\begin{bmatrix}
P^{-1} & \vzero \\
\vzero & Q
\end{bmatrix}
=
\begin{bmatrix}
\vzero & P A_i Q \\
\vzero & \vzero 
\end{bmatrix}
=
\begin{bmatrix}
\vzero & B_i'\\
\vzero & \vzero
\end{bmatrix}
$.
It follows that, $\tilde P^{-1}\tilde\vA \tilde P=\tilde\vB^R$, which just says that 
$\tilde P^{-1} \langle\tilde\vA\rangle \tilde P =\langle\tilde\vB\rangle$.

\textbf{For the if direction,} since $\langle\tilde\vA\rangle$ and $\langle\tilde\vB\rangle$ are conjugate, 
there exist $\tilde P\in\GL(\ell+n, \F)$, and $\tilde R\in\GL(m, \F)$, such that 
$\tilde P^{-1}\tilde\vA \tilde P=\tilde\vB^{\tilde R}$. Write $\tilde\vB^{\tilde R}:=\tilde\vB'=(\tilde B_1', 
\dots, \tilde 
B_m')$, where $\tilde B_i'=\begin{bmatrix}
\vzero & B_i'\\
\vzero & \vzero 
\end{bmatrix}$, $B_i'\in\M(\ell\times n, \F)$. Let $\tilde P=\begin{bmatrix}
P_{1,1} & P_{1,2}\\
P_{2,1} & P_{2,2}
\end{bmatrix}$, where $P_{1,1}\in \M(\ell, \F)$. Then as $\tilde\vA 
\tilde P=\tilde P\tilde\vB'$, we have for every $i\in[m]$,
\begin{equation}\label{eq:3TI_conjugacy}
\begin{bmatrix}
P_{1,1} & P_{1,2}\\
P_{2,1} & P_{2,2}
\end{bmatrix}
\begin{bmatrix}
\vzero & A_i\\
\vzero & \vzero 
\end{bmatrix}
=
\begin{bmatrix}
\vzero & P_{1,1} A_i \\
\vzero & P_{2,1} A_i
\end{bmatrix}
=
\begin{bmatrix}
B_i'P_{2,1} & B_i'P_{2,2} \\
\vzero & \vzero 
\end{bmatrix}
=
\begin{bmatrix}
\vzero & B_i'\\
\vzero & \vzero 
\end{bmatrix}
\begin{bmatrix}
P_{1,1} & P_{1,2}\\
P_{2,1} & P_{2,2}
\end{bmatrix}.
\end{equation} This in particular implies that for every $i\in[m]$, 
$P_{2,1}A_i=\vzero$. In 
other words, every row of $P_{2,1}$ lies in the common left kernel of $A_i$ with 
$i\in[m]$. Since $\vA$ is non-degenerate, $P_{2,1}$ must be the zero matrix. It 
follows that $\tilde P=\begin{bmatrix}
P_{1,1} & P_{1,2} \\
\vzero & P_{2,2}
\end{bmatrix}\in\GL(\ell+n, \F)$, so $P_{1,1}$ and $P_{2,2}$ are both invertible 
matrices. By Equation~\ref{eq:3TI_conjugacy}, we have $P_{1,1}\vA=\vB^{\tilde R} P_{2,2}$, 
where $P_{1,1}\in\GL(\ell, \F)$, $P_{2,2}\in\GL(n, \F)$, and $\tilde R\in\GL(m, \F)$, 
which just says that $\tA$ and $\tB$ are isomorphic as 3-tensors.
\end{proof}

\begin{corollary}
\ThreeTIlong reduces to
\begin{enumerate}
\item \algprobm{Matrix Lie Algebra Conjugacy}, where $L$ is commutative; 

\item \algprobm{Associative Matrix Algebra Conjugacy}, where $A$ is commutative (and in fact has the property that $ab=0$ for all $a,b \in A$; note that $A$ is not unital); 

\item \algprobm{Matrix Lie Algebra Conjugacy}, where $L$ is solvable of derived length 2, and $L / [L, L] \cong \F$; and, 

\item \algprobm{Associative Matrix Algebra Conjugacy}, where the Jacobson radical $J(A)$ squares to zero, and $A / J(A) \cong \F$.
\end{enumerate}
\end{corollary}

\begin{proof}
We use the notation from the proof of Proposition~\ref{prop:3-tensor_conjugacy}. Note that the matrix spaces constructed there, e.\,g., $\tilde \vA$, are all subspaces of the $(\ell + n) \times (\ell + n)$ matrix space $\mathcal{U} := \begin{bmatrix} \vzero & M(\ell \times n,\F) \\ \vzero & \vzero \end{bmatrix}$.

For (1) and (2), observe that for any two matrices $A,A' \in \mathcal{U}$, we have $AA' = 0$, and thus $[A,A'] = AA' - A'A = 0$ as well. Thus any matrix subspace of $\mathcal{U}$ is both a commutative matrix Lie algebra and a commutative associative matrix algebra with zero product.

For (3) and (4), we note that we can alter the construction of Proposition~\ref{prop:3-tensor_conjugacy} by including the matrix $M_0 = \begin{bmatrix} I_\ell & \vzero \\ \vzero & \vzero \end{bmatrix}$ in both matrix spaces $\tilde \cA$ and $\tilde \cB$ without disrupting the reduction. Indeed, for the forward direction we have that (again, following notation as above) $\tilde P^{-1} \begin{bmatrix} I_\ell & \vzero \\ \vzero & \vzero \end{bmatrix} \tilde P = \begin{bmatrix} P & \vzero \\ \vzero & Q^{-1} \end{bmatrix}\begin{bmatrix} I_\ell & \vzero \\ \vzero & \vzero \end{bmatrix}\begin{bmatrix} P^{-1} & \vzero \\ \vzero & Q\end{bmatrix} = \begin{bmatrix} I_\ell & \vzero \\ \vzero & \vzero \end{bmatrix}$. 

For the reverse direction, we then have that for $\tilde \vB' = \tilde \vB^{\tilde R}$, we have $\tilde B'_i = \begin{bmatrix}\alpha I_d & B_i' \\ \vzero & \vzero \end{bmatrix}$. Let $\tilde P=\begin{bmatrix}
P_{1,1} & P_{1,2}\\
P_{2,1} & P_{2,2}
\end{bmatrix}$, where $P_{1,1}\in \M(\ell, \F)$. Then as $\tilde\vA 
\tilde P=\tilde P\tilde\vB'$, we have for every $i\in[m]$,
\begin{equation}\label{eq:3TI_conjugacy}
\begin{bmatrix}
P_{1,1} & P_{1,2}\\
P_{2,1} & P_{2,2}
\end{bmatrix}
\begin{bmatrix}
\vzero & A_i\\
\vzero & \vzero 
\end{bmatrix}
=
\begin{bmatrix}
\vzero & P_{1,1} A_i \\
\vzero & P_{2,1} A_i
\end{bmatrix}
=
\begin{bmatrix}
\alpha P_{1,1} + B_i' P_{2,1} & B_i' P_{2,2} \\
\alpha P_{2,1} & \vzero 
\end{bmatrix}
=
\begin{bmatrix}
\alpha I_d & B_i'\\
\vzero & \vzero 
\end{bmatrix}
\begin{bmatrix}
P_{1,1} & P_{1,2}\\
P_{2,1} & P_{2,2}
\end{bmatrix}.
\end{equation} 
Considering the (2,1) block of this equation, we find that if $\alpha \neq 0$, then immediately $P_{2,1} = \vzero$. But even if $\alpha = 0$, then we are back to the same argument as in Proposition~\ref{prop:3-tensor_conjugacy}, namely that by the non-degeneracy of $\vA$, we still get $P_{2,1} = \vzero$ by considering the (2,2) block. The remainder of the argument only depended on the (1,2) block of the preceding equation, which is the same as before.

Finally, to see the structure of the corresponding algebras, we must consider how our new element $M_0$ interacts with the others. Easy calculations reveal:
\[
M_0^2 = M_0 \qquad
M_0 \tilde A_i = \tilde A_i \qquad 
\tilde A_i M_0 = \vzero \qquad 
[M_0, \tilde A_i] = M_0 \tilde A_i - \tilde A_i M_0 = \tilde A_i
\]

(3) For the structure of the Lie algebra, we have from the above equations that any commutator is either 0 or lands in $\mathcal{U}$. And since $[M_0, \tilde A_i] = \tilde A_i$, we have that $[L,L]$ is the subspace of $\mathcal{U}$ that we started with before including $M_0$. Since everything in that subspace commutes, we get that $[[L,L], [L,L]] = 0$, and thus the Lie algebra is solvable of derived length 2. Moreover, $L / [L, L]$ is spanned by the image of $M_0$, whence it is isomorphic to $\F$.

(4) Recall that for rings without an identity, the Jacobson radical can be characterized as $J(A) = \{ a \in A | (\forall b \in A)(\exists c \in A)[c + ba = cba]\}$ \cite[p.~63]{lam}. Note that the only nontrivial cases to check are those for which $b = M_0$, since otherwise $ba = 0$ and then we may take $c=0$ as well. So we have $J(A) = \{a \in A | (\exists c \in A)[c + M_0 a = c M_0 a]\}$. But since $M_0$ is a left identity, this latter equation is just $c + a = ca$. For any $a \in \mathcal{U}$, we may take $c = -a$, since then both sides of the equation are zero, and thus $J(A)$ includes all the matrices in the original space from Proposition~\ref{prop:3-tensor_conjugacy}. However, $M_0 \notin J(A)$, for there is no $c$ such that $c + M_0 = c M_0$: any element of $A$ can be written $\alpha M_0 + u$ for some $u \in \mathcal{U}$. Writing $c$ this way, we are trying to solve the equation $\alpha M_0 + u + M_0 = (\alpha M_0 + u)M_0 = \alpha M_0$. Thus we conclude $u=0$, and then we get that $\alpha+1 = \alpha$, a contradiction. So $M_0 \notin J(A)$, and thus $A / J(A)$ is spanned by the image of $M_0$, whence it is isomorphic to $\F$.
\end{proof}

\subsection{From Matrix Space Isometry to Algebra Isomorphism and Trilinear Form Equivalence}

\begin{proposition} \label{prop:isometry_algebra}
\MatSpIsomlong reduces to \algprobm{Algebra Isomorphism} and \NcCubicFormlong. Symbolically, $V \otimes V \otimes W$ reduces to $V' \otimes V' \otimes V'^*$ and to $V' \otimes V' \otimes V'$, where $\dim V' = \dim V + \dim W$.
\end{proposition}

\newcommand{\Alg}{\text{Alg}}
\begin{proof}
\textbf{The construction.} Given a matrix space $\cA$ by an ordered linear basis $\vA = (A_1, \dotsc, A_m)$, construct the 3-way array $\tA' \in T((n+m) \times (n+m) \times (n+m), \F)$ whose frontal slices are:
\[
A_i' = \vzero \quad (\text{for } i \in [n]) \qquad A_{n+i}' = \begin{bmatrix} A_i & \vzero \\ \vzero & \vzero \end{bmatrix} \quad (\text{for } i \in [m]).
\]
Let $\Alg(\tA')$ denote the algebra whose structure constants are defined by $\tA'$, and let $f_{\tA'}$ denote the trilinear form whose coefficients are given by $\tA'$. 

Given two matrix spaces $\cA, \cB$, we claim that $\cA$ and $\cB$ are isometric if and only if $\Alg(\tA') \cong \Alg(\tB')$ (isomorphism of algebras) if and only if $f_{\tA'}$ and $f_{\tA'}$ are equivalent as trilinear forms. The proofs are broken into the following two lemmas, which then complete the proof of the proposition.
\end{proof}

\begin{lemma}\label{lem:VVW_VVVstar}
Let notation be as above. The matrix spaces $\cA,\cB$ are isometric if and only if $\Alg(\tA')$ and $\Alg(\tB')$ are isomorphic.
\end{lemma}

\begin{proof}
Let $\vA,\vB$ be the ordered bases of $\cA, \cB$, respectively. Recall that $\cA, \cB$ are isometric if and only if there exist 
$(P, R) \in\GL(n, \F) \times \GL(m, \F)$ such that $P^t\vA P=\vB^R$. Also recall that 
$\Alg(\tA')$ and 
$\Alg(\tB')$ are isomorphic as algebras if and only if there exists $\tilde P\in\GL(n+m, \F)$ 
such that $\tilde P^t\vA' \tilde P=\vB'^{\tilde P}$. Since $A_i$ (resp. $B_i$) form a linear basis of 
$\cA$ (resp. $\cB$), we have that $A_i$ (resp. 
$B_i$) are linearly independent.

\paragraph{The only if direction} is easy to verify. Given an isometry $(P, R)$ between 
$\cA$ and $\cB$, let $\tilde P=\begin{bmatrix}
P & \vzero \\
\vzero & R
\end{bmatrix}$. Let $\tilde P^t \vA' \tilde P=(A_1'', \dots, A_{n+m}'')$. Then for $i\in[n]$, 
$A_i''=\vzero$. For $n+1\leq i\leq n+m$, $A_i''=\begin{bmatrix} P^t A_i P & \vzero \\
\vzero & \vzero \end{bmatrix}$. Let $\vB'^{\tilde P}=(B_1'', \dots, B_{n+m}'')$. Then for 
$i\in[n]$, $B_i''=\vzero$. For $n+1\leq i\leq n+m$, $B_i''$ is the $(i-n)$th 
matrix in $\vB^R$, which in turn equals $P^t A_i P$ by the assumption on $P$ and 
$R$. This proves the only if direction.

\paragraph{For the if direction,} let $\tilde P=\begin{bmatrix}
P & X \\
Y & R
\end{bmatrix}\in\GL(n+m, \F)$ be an algebra isomorphism, where $P$ is of size 
$n\times n$. Let 
$\tilde P\vA' \tilde P^t=(A_1'', 
\dots, A_{n+m}'')$, and $\vB'^{\tilde P}=(B_1'',\dots, B_{n+m}'')$. Since for $i\in[n]$, 
$A_i'=\vzero$, we have $A_i''=\vzero=B_i''$. Therefore $Y$ has to be $\vzero$, 
because 
$B_i$'s 
are linearly independent. It follows that $\tilde P=
\begin{bmatrix}
P & X \\
\vzero & R
\end{bmatrix}$, where $P$ and $R$ are invertible.
So for $1\leq i\leq 
m$, we have $\tilde P^t A_{i+n}' \tilde P=\begin{bmatrix}
P^t & \vzero \\
X^t & R^t
\end{bmatrix}\begin{bmatrix}
A_i & \vzero \\
\vzero & \vzero 
\end{bmatrix} 
\begin{bmatrix}
P & X \\
\vzero & R
\end{bmatrix}
=
\begin{bmatrix}
P^t A_i P& P^t A_i X\\
X^t A_i P & X^t A_i X
\end{bmatrix}
$. Also the last $m$ matrices in $\vB'^{\tilde P}$ are $\begin{bmatrix} B_i'' & \vzero 
\\\vzero & \vzero \end{bmatrix}$, where $B_i''$ is the $i$th matrix in $\vB^R$.
This implies that $P\in\GL(n, 
\F)$ and $R\in \GL(m, \F)$ 
together form an isometry between $\cA$ and $\cB$.
\end{proof}

\begin{corollary} \label{cor:pseudo_special}
\MatSpIsomlong reduces to
\begin{enumerate}
\item \label{cor:pseudo_special:assoc_comm_unital} \algprobm{Associative Algebra Isomorphism}, for algebras that are commutative and unital; 

\item \algprobm{Associative Algebra Isomorphism}, for algebras that are commutative and 3-nilpotent ($abc=0$ for all $a,b,c \in A$); and, 

\item \algprobm{Lie Algebra Isomorphism}, for Lie algebras that are 2-step nilpotent ($[u,[v,w]]=0$ for all $u,v,w \in L$).
\end{enumerate}
\end{corollary}

\begin{proof}
We follow the notation from the proof of Lemma~\ref{lem:VVW_VVVstar}. We begin by observing that $\Alg(\tA')$ is a 3-nilpotent algebra, and therefore is automatically associative. Let $V' = V \oplus W$, where $\dim V = n$, $\dim W = m$, and, as a subspace of $V' \cong \F^{n+m}$, $V$ has a basis given by $e_1, \dotsc, e_n$ and $W$ has a basis given by $e_{n+1}, \dotsc, e_{n+m}$. Let $\circ$ denote the product in $\Alg(\tA')$, so that $x_i \circ x_j = \sum_{k} \tA'(i,j,k) x_k$. Note that because the lower $m$ rows and the rightmost $m$ columns of each frontal slice of $\tA'$ are zero, we have that $w \circ x = x \circ w = 0$ for any $w \in W$ and any $x \in V'$. Thus only way to get a nonzero product is of the form $v \circ v'$ where $v,v' \in V$, and here the product ends up in $W$, since the only nonzero frontal slices are $n+1, \dotsc, n+m$. Since any nonzero product ends up in $W$, and anything in $W$ times anything at all is zero, we have that $abc=0$ for all $a,b,c \in \Alg(\tA')$, that is, $\Alg(\tA')$ is 3-nilpotent. Any 3-nilpotent algebra is automatically associative, since the associativity condition only depends on products of three elements.

(2) If instead of general \MatSpIsomlong, we start from \SymMatSpIsomlong (which is also $\cc{3TI}$-complete by Corollary~\ref{cor:AltSymMatSpIsom}), then we see that the algebra is commutative, for we then have $\tA'(i,j,k) = \tA'(j,i,k)$, which corresponds to $x_i \circ x_j = x_j \circ x_i$.

(1) As is standard, from the algebra $A = \Alg(\tA')$, we may adjoin a unit by considering $A' = A[e] / (e \circ x = x \circ e = x | x \in A')$. In terms of vector spaces, we have $A' \cong A \oplus \F$, where the new $\F$ summand is spanned by the identity $e$. This standard algebraic construction has the property that two such algebras $A,B$ are isomorphic if and only if their corresponding unit-adjoined algebras $A', B'$ are (see, e.\,g., \cite{dorroh1932concerning, wikipediaAdjoin}). 

(3) By starting from an alternating matrix space $\cA$ (and noting that \AltMatSpIsomWords is still $\cc{3TI}$-complete, by Corollary~\ref{cor:AltSymMatSpIsom}), we get that $\Alg(\tA')$ is alternating, that is, $v \circ v = 0$. Since we still have that it is 3-nilpotent, $a \circ b \circ c = 0$, we find that $\circ$ automatically satisfies the Jacobi identity. An alternating product satisfying the Jacobi identity is, by definition, a Lie bracket (that is, we can define $[v,w] := v \circ w$), and thus we get a Lie algebra with structure constants $\tA'$. Translating the 3-nilpotency condition $a \circ b \circ c = 0$ into the Lie bracket notation, we get $[a, [b,c]] = 0$, or in other words that the Lie algebra is nilpotent of class 2.
\end{proof}

\begin{corollary} \label{cor:comm_cubic_form}
\ThreeTIlong reduces to \CubicFormlong. 
\end{corollary}

\begin{proof}
Agrawal and Saxena \cite{AS06} show that \algprobm{Commutative Algebra 
Isomorphism} reduces to \CubicFormlong. Combine with Corollary~\ref{cor:pseudo_special}(\ref{cor:pseudo_special:assoc_comm_unital}). 
\end{proof}

The reduction from $V\otimes V\otimes W$ to $V' \otimes V' \otimes V'$ is achieved by 
the same construction.
\begin{lemma}\label{lem:VVW_VVV}
Let $\vA, \vB, \vA'$, and $\vB'$ be as above. Then $\vA$ and $\vB$ 
are pseudo-isometric if and only if $\vA'$ and $\vB'$ are isomorphic as 
trilinear forms.
\end{lemma}
\begin{proof}
Recall that $\vA$ and $\vB$ are pseudo-isometric if there exist 
$P\in\GL(n, \F), R\in\GL(m, \F)$ such that $P^t\vA P=\vB^R$. Also recall that 
$\vA'$ and 
$\vB'$ are equivalent as trilinear forms if there exists 
$\tilde P\in\GL(n+m, \F)$ 
such that $\tilde P^t\vA'^{\tilde P} \tilde P=\vB'$. Since $A_i$ (resp. $B_i$) form a linear basis of 
$\cA$, we have that $A_i$ (resp. 
$B_i$) are linearly independent.

\paragraph{The only if direction} is easy to verify. Given an pseudo-isometry $P, R$ between 
$\vA$ and $\vB$, let $\tilde P=\begin{bmatrix}
P & \vzero \\
\vzero & R^{-1}
\end{bmatrix}$. Then it can be verified easily that $\tilde P$ is a trilinear form equivalence between $\vA'$ and $\vB'$, following the same approach in 
the proof of  
Lemma~\ref{lem:VVW_VVVstar}.

\paragraph{For the if direction,} write $\tilde P=\begin{bmatrix}
P & X \\
Y & R
\end{bmatrix}\in\GL(n+m, \F)$ be a trilinear form equivalence between 
$\vA'$ and 
$\vB'$. We first observe that the last $m$ matrices in $\tilde P^t\vA' \tilde P$ are still 
linearly independent. Then, because of the first $n$ matrices in $\vB'$ are all 
zero matrices, $Y$ has to be the zero matrix. It follows that 
$\tilde P=\begin{bmatrix}
P & X \\
\vzero & R
\end{bmatrix}$, where $P$ and $R$ are invertible. Then it can be 
verified easily that $P$ and $R^{-1}$ form an pseudo-isometry between $\vA$ and 
$\vB$, 
following the same approach in 
the proof of  
Lemma~\ref{lem:VVW_VVVstar}.
\end{proof}

\section{Reducing \DeeTIlong to \ThreeTIlong} \label{sec:dto3}
\newcommand{\Path}{\text{Path}}

\begin{thmdto3}
$d$-\TIlong reduces to \AlgIsolong. If the input tensor has size $n_1 \times n_2 \times \dotsb \times n_d$, then the output algebra has dimension $O(d^2 n^{d-1})$ where $n = \max\{n_i\}$.
\end{thmdto3}

\begin{remark} \label{rmk:dto3_dim}
One cannot do too much better in terms of size of the output, as the following argument suggests. Over finite fields, we may count the number of orbits, which provides a rigorous lower bound on the size blow-up of any kernel reduction (see, e.\,g., \cite[Sec.~4.2.4]{FortnowGrochowPEq}). Over infinite fields, if we consider algebraic reductions, they must preserve dimension, so we can make a similar (albeit more heuristic) argument by considering the ``dimension'' of the set of orbits. Here we have put ``dimension'' in quotes because the set of orbits is not a well-behaved topological space (it is typically not even $T_1$), but even in this case the same argument should essentially hold. The space of $n \times n \times \dotsb \times n$ $d$-tensors has dimension $n^d$, and the group $\GL_n \times \dotsb \times \GL_n$ has dimension $dn^2$, so the ``dimension'' of the set of orbits is at least $n^d - dn^2 \sim n^d$ ($d \geq 3$); over $\F_q$, the number of orbits is at least $q^{n^d - dn^2}$. For algebras of dimension $N$, the space of such algebras is $\leq N^3$-dimensional, so the ``dimension'' of the set of orbits is at most $N^3$; over $\F_q$, the number of orbits is at most $q^{N^3}$. Thus we need $N^3 \gtrsim n^d$, whence $N \gtrsim n^{d/3}$. 
\end{remark}

\begin{proof}[Proof idea]
The idea here is similar to the reduction from \ThreeTI to \AlgIsolong: we want to 
create an algebra in which all products eventually land in an ideal, and 
multiplication of algebra elements by elements in the ideal is described by the 
tensor we started with. For a 3-tensor this was very natural, as the structure 
constants of any algebra form a 3-tensor. In that case, we are using it to say how 
to write the product of 2 elements as a linear combination (the third factor of 
the tensor) of basis elements. With a $d$-tensor for $d \geq 3$, we now want to 
use it to describe how to write the product of $d-1$ elements as a linear 
combination of basis elements. The tricky part here is that in an algebra we must 
still describe the product of any \emph{two} elements. The idea is to create a set 
of generators, let them freely generate monomials up to degree $d-2$, and then 
when we get a product of $d-1$ generators, rewrite it as a linear combination of 
generators according to the given tensor. This idea almost provides one direction 
of the reduction: if two $d$-tensors $\tA, \tB$ are isomorphic, then the 
corresponding algebras $\cA,\cB$ are isomorphic. 
However, there is an issue with implementing this, namely that monomials are 
commutative, but our tensors $\tA,\tB$ need not be symmetric, and moreover, they 
need not even be ``square'' (have all side lengths equal). 
In \cite[Thm.~5]{AS05} they reduce \DFormlong to \algprobm{Commutative Algebra 
Isomorphism} along similar lines, but there the starting objects are themselves 
commutative, so this was not an issue. 
In our case, we will get around this 
using a certain noncommutative algebra where the only nonzero products are those 
which come ``in the right order.''

Another potentially tricky aspect of the reduction is the converse: suppose we build our algebras $\cA,\cB$ as above from two $d$-tensors, and $\cA,\cB$ are isomorphic; how can we guarantee that $\tA$ and $\tB$ are isomorphic? For this, we would like to be able to identify certain subsets of the algebras as characteristic (invariant under any automorphism), so that those characteristic subsets force the isomorphism to take a particular form, which we can then massage into an isomorphism between the tensors $\tA, \tB$. Our way of doing this is to encode the ``degree'' structure into the path algebra of a graph, as described in the next section. If the graph has no automorphisms, then the path algebra has the advantage that for any two vertices $i,j$, the subset of $\cA$ spanned by the paths from $i$ to $j$ is nearly characteristic in a way we make precise below. 
\end{proof}

\subsection{Preliminaries for Theorem~\ref{thm:d_to_3}} \label{sec:dto3_prelim}
To make the above proof idea precise, we will need a little background on Leavitt 
path algebras (a.k.a. quiver algebras) and their quotients. For a textbook reference on these algebras, see
\cite[Ch.~II]{AssemSimsonSkowronski}, and for a textbook treatment of 
Wedderburn--Artin theory and the Jacobson radical, see \cite{lam}. Aside from the 
definition of path algebra, most of this section will end up being used as a black 
box; we include it mostly for ease of reference.

We start with some important, classical results on the structure of associative algebras. The 
\emph{Jacobson radical} of an associative algebra $A$, here denoted $R(A)$, is the 
intersection of all maximal right ideals. Equivalently, $R(A) = \{a \in A : 
\text{every element of } 1 + AxA \text{ is invertible}\}$. A unital algebra $A$ 
over a field $\F$ is \emph{semisimple} if $R(A) = 0$; in this case, by 
Wedderburn's Theorem (see below), $A$ is isomorphic to a direct sum of matrix 
algebras over finite-degree division rings extending $\F$. An algebra $A$ is 
called \emph{separable} if it is semisimple over every field extending $\F$, that 
is, $A \otimes_{\F} \mathbb{K}$ is semisimple for all fields $\mathbb{K}$ 
extending $\F$. Equivalently, $A$ is separable if it is isomorphic to 
$\bigoplus_{i=1}^d \M(d_i,\F_i)$, where each $\F_i$ is a division ring extending 
$\F$ such that the center $Z(\F_i)$ is a separable field extension of $\F$. If 
$\F$ has characteristic zero or is perfect (which includes all finite fields), 
then all its extensions are separable. For the algebra we construct, it will 
simply be a direct sum of copies of $\F$, so it is automatically separable over any field.

An element $a \in A$ is \emph{idempotent} if $a^2 = a$. An idempotent $e$ is \emph{primitive} if it cannot be written as the sum of two nonzero idempotents. Two idempotents $e,f$ are \emph{orthogonal} if $ef = fe = 0$. A \emph{complete set of primitive orthogonal idempotents} of $A$ is a set $\{e_1, \dotsc, e_n\}$ of primitive idempotents which are pairwise orthogonal, and such that the set is maximal subject to this condition. 

\begin{theorem}[{Wedderburn--Mal'cev, see, e.\,g., \cite{farnsteiner}}] \label{thm:WM}
Let $A$ be an finite-dimensional, associative, unital algebra over a field $\F$. Then 
\begin{enumerate}
\item $A / R(A) \cong \bigoplus_{i=1}^d \M(d_i,\F_i)$ (as algebras), where each 
$\F_i$ is a division ring of finite degree over $\F$. 
\item If $A/R(A)$ is separable, then there exists a subalgebra $S \subseteq A$ such that $A = S \oplus R(A)$ (as $\F$-vector spaces).
\item \label{thm:WM:conj}If $T \subseteq A$ is any separable subalgebra, then there exists $r \in R(A)$ such that $(1+r) T (1+r)^{-1} \subseteq S$.
\end{enumerate}
\end{theorem}
The last part of the preceding theorem is what we will use to show that the set of paths $i \to j$ in our graph is ``nearly characteristic;'' that is, it is not characteristic, but it is characteristic up to conjugacy (=inner automorphisms).

\begin{definition}[Leavitt path algebra]
Given a directed multigraph $G$ (possibly with parallel edges and self-loops, a.k.a. quiver), its \emph{Leavitt path algebra} $\Path(G)$ is the algebra of paths in $G$, where multiplication is given by concatenation of paths (when this is well-defined), and zero otherwise. That is, $\Path(G)$ is generated by $\{e_v : v \in V(G)\} \cup \{x_a : a \in E(G)\}$, where the generators $e_v$ are thought of as the ``path of length $0$'' at vertex $v$. The defining relations in $\Path(G)$ are that the product of two paths is their concatenation if the end of the first equals the start of the second, and zero otherwise. More formally, the relations are:
\begin{eqnarray*}
e_v e_w & = & \delta_{v,w} e_v \\
e_v x_{a} & = & \delta_{v,\text{start}(a)} x_{a} \\
x_{a} e_{v} & = & \delta_{v,\text{end}(a)} x_{a} \\
x_a x_b & = & 0 \text{ if } \text{start}(b) \neq \text{end}(a),
\end{eqnarray*}
where $\delta_{x,y}$ is the Kronecker delta: it is $1$ if $x=y$ and $0$ otherwise.
\end{definition}
Note that we are allowed to take formal linear combinations of paths in this algebra, as it is an $\F$-algebra (so in particular, it is an $\F$-vector space). The \emph{arrow ideal} of $\Path(G)$ is the two-sided ideal generated by the arrows, and has a basis consisting of all paths of length $\geq 1$; it is denoted $R_G$.

\begin{lemma}[{See \cite[Cor.~II.1.11]{AssemSimsonSkowronski}}] \label{lem:idempotents}
If $G$ is finite, connected, and acyclic, then $R(\Path(G))$ is the arrow ideal $R_G$, and has a basis consisting of all paths of length $\geq 1$, and the set $\{e_v : v \in V(G)\}$ is a complete set of primitive orthogonal idempotents.
\end{lemma}

\begin{corollary}
Let $G$ be a finite, connected, acyclic graph, and $I$ an ideal of $\Path(G)$ contained in $R_G$; let $A = \Path(G) / I$. Then (1) $R(A) = R_G / I$, (2) $A/R(A) \cong \F^{\oplus |V(G)|}$, whence $A/R(A)$ is separable, and (3) $\{\overline{e}_v : v \in V(G)\}$ is a complete set of primitive orthogonal idempotents, where $\overline{e}_v$ is the image of $e_v$ under the quotient map $\Path(G) \to \Path(G) / I = A$.
\end{corollary}

\begin{proof}
(1) This holds for any ideal contained in the radical of any finite-dimensional associative unital algebra \cite[Prop.~4.6]{lam}. 

(2) It is clear that as vector spaces, $\Path(G) = \langle e_1, \dotsc, e_n 
\rangle \oplus R_G$ (where $n=|V(G)|$), and the span of the $e_i$ is easily seen 
to be an algebra isomorphic to $\F^n$, where the $i$-th copy of $\F$ is spanned by 
$\pi(e_i)$, where $\pi\colon \Path(G) \to \Path(G) / R_G$ is the natural 
projection. Thus $\Path(G) / R_G \cong \F^n$. Since $R(A) = R_G / I$, we have $A / 
R(A) = (\Path(G) / I) / (R_G / I) \cong \Path(G) / R_G \cong \F^n$. As a 
semisimple algebra, we thus have that $A/R(A) \cong \bigoplus \M(1,\F)$, and as 
$\F$ is always a separable extension over itself, $A/R(A)$ is separable.

(3) The property of being a set of primitive orthogonal idempotents is preserved 
by homomorphisms, so there are only two things to check here: first, that none of 
the $\overline{e}_v$ is zero modulo $I$, and second, that there are no additional 
primitive idempotents in $A$ that are mutually orthogonal with every $\overline{e}_v$. To see that none of the $\overline{e}_v$ are zero, 
note that $\pi\colon \Path(G) \to \Path(G) / R_G$ factors through $A$; then since 
$\pi(e_v) \neq 0$ for any $v$ (from the previous paragraph), it must be the case 
that $\overline{e}_v \neq 0$ as well. Finally, we must show this is a complete set 
of primitive orthogonal idempotents. Suppose not; that is, suppose there is some 
$e \notin \{\overline{e}_v : v \in V(G)\}$ such that $e$ is a primitive idempotent 
that is orthogonal in $A$ to every $\overline{e}_v$. First, we claim that $e 
\notin R(A) = R_G / I$. For, since $G$ is a finite acyclic graph, its arrow ideal 
$R_G$ is nilpotent: there are no paths longer than $n-1=|V(G)-1|$, so we must have 
$R_G^n = 0$, whence $R_G$ cannot contain any idempotents. Since $R_G$ is 
nilpotent, the same must be true of $R_G / I$, whence $R_G$ cannot contain any 
idempotents, so $e$ cannot be in $R_G$. But then the image of $e$ in $A/R_G$ is 
nonzero (since $e \notin R_G$), so $e$ is another primitive idempotent orthogonal 
to every $\pi(e_v)$ in $\Path(G) / R_G = A / R(A)$. But this is a contradiction, 
since $\{\pi(e_v)\}$ is already a complete set of primitive orthogonal idempotents 
for $A/R(A)$.
\end{proof}

Finally, in the course of the proof, we will use the following construction of Grigoriev:

\begin{theorem}[{Grigoriev \cite[Theorem~1]{Grigoriev83}}] \label{thm:grigoriev}
\GIlong is equivalent to \AlgIsolong for algebras $A$ such that the radical squares to zero and $A/R(A)$ is abelian.
\end{theorem}

In our proof, all we will need aside from Grigoriev's result is to see the construction itself, which we recall here in language consistent with ours.

\begin{proof}[Construction \cite{Grigoriev83}]
Given a graph $G$, construct an algebra $\cA_G$ as follows: it is generated by $\{e_i : i \in V(G)\} \cup \{e_{ij} : (i,j) \in E(G)\}$ subject to the following relations: $e_i e_j = \delta_{ij} e_i$, $e_i e_{kj} = \delta_{ik} e_{kj}$, $e_{kj}e_i = \delta{ij} e_{kj}$, $e_{ij} e_{kl} = 0$ when $j \neq k$, $R(\cA_G)$ is generated by $\{e_{ij}\}$, and the radical squares to zero. It is immediate that this is just $\Path(G) / R_G^2$. From any such algebra $\cA$, Grigoriev recovers a corresponding weighted graph, where the weight on $(i,j)$ is $\dim e_i \cA e_j$. In our setting we use multiple parallel edges rather than weight, but the proof goes through \emph{mutatis mutandis}.
\end{proof}

\subsection{Proof of Theorem~\ref{thm:d_to_3}}
\begin{proof}
Let $\tA$ be an $n_1 \times n_2 \times \dotsb \times n_d$ $d$-tensor. Let $G$ be the following directed multigraph (see Figure~\ref{fig:graph}): it has $d$ vertices, labeled $1,\dotsc,d$, and for $i=1,\dotsc,d-1$, it has $n_i$ parallel arrows from vertex $i$ to vertex $i+1$, and $n_d$ parallel arrows from $1$ to $d$. 

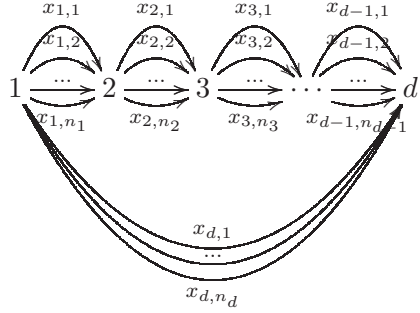
\begin{figure}[!htbp]
\[
\xymatrix{
1 \ar[r] \ar@/^2pc/[r]^{x_{1,1}} \ar@/^1pc/[r]^{x_{1,2}} \ar[r]^{...} \ar@/_/[r]_{x_{1,n_1}}  \ar@/_5pc/[rrrr]^{x_{d,1}} \ar@/_5.5pc/[rrrr]^{...} \ar@/_6pc/[rrrr]_{x_{d,n_d}} & 
2 \ar[r] \ar@/^2pc/[r]^{x_{2,1}} \ar@/^1pc/[r]^{x_{2,2}} \ar[r]^{...} \ar@/_/[r]_{x_{2,n_2}} & 
3 \ar[r] \ar@/^2pc/[r]^{x_{3,1}} \ar@/^1pc/[r]^{x_{3,2}} \ar[r]^{...} \ar@/_/[r]_{x_{3,n_3}} & 
\dotsb \ar[r] \ar@/^2pc/[r]^{x_{d-1,1}} \ar@/^1pc/[r]^{x_{d-1,2}} \ar[r]^{...} \ar@/_/[r]_{x_{d-1,n_{d-1}}} & 
d \\
}
\]
\caption{ \label{fig:graph} The graph $G$ whose path algebra we take a quotient of to construct the reduction for Theorem~\ref{thm:d_to_3}.}
\end{figure}

Because of the structure of this graph, we can index the generators of $\Path(G)$ a little more mnemonically than in the preliminaries above: let the generators corresponding to the $n_i$ arrows from $i \to (i+1)$ be $x_{i,a}$ for $a=1,\dotsc,n_i$, and let the generators corresponding to the $n_d$ arrows $1 \to d$ be $x_{d,a}$ for $a=1,\dotsc,n_d$. Let $\cA$ be the quotient of $\Path(G)$ by the relation\footnote{For those familiar with quiver algebras, we note that this ideal is \emph{not} admissible, as it is not contained in $R_G^2$. It can probably be made admissible by inserting new vertices in the middle of each edge $1 \to d$. However, when we tried to do that in a naive way, we ran into problems verifying the reduction, as what should be a linear transformation either ends up being incorrect or ends up being quadratic, either of which caused issues.}
\begin{equation}
x_{1,i_1} x_{2,i_2} \dotsb x_{d-1,i_{d-1}} = \sum_{j=1}^{n_d} \tA(i_1, i_2, \dotsc, i_{d-1}, j) x_{d,j} \label{eq:rel_tensor}
\end{equation}

At the moment, we only have $\cA$ in terms of generators and relations; however, 
it will be easy to turn it into its basis representation. The key is to bound its 
dimension, which we do now.
Except for paths of length $d-1$ (because of the nontrivial relations (\ref{eq:rel_tensor})), this is just counting the number of paths in the graph described above. The only nonzero monomials of degree $k+1$ are those of the form $x_{i,a_i} x_{i+1,a_{i+1}} x_{i+2, a_{i+2}} \dotsb x_{i+k,a_{i+k}}$. For a given choice of $i \in \{1,\dotsc,d-1-k\}$, there are exactly $n_i n_{i+1} \dotsb n_{i+k}$ such monomials, so we have 
\begin{eqnarray*}
\dim \cA & = & \#\{e_i\} + n_d + \sum_{k < d-1} \sum_{i=1}^{d-1-k} \#\{\text{paths } i \to (i+k)\}\\
& = & d + n_d + \sum_{k=0}^{d-2} \sum_{i=1}^{d-1-k} \prod_{j=i}^{i+k} n_j \\
& \leq & 2n + \sum_{k=0}^{d-2} \sum_{i=1}^{d-1-k} n^{k+1} \\
& \leq & O(d^2 n^{d-1}).
\end{eqnarray*}
Note that in the first line we can exactly specify $\dim \cA$, independent of $\tA$ itself (depending only on its dimensions).
For any fixed $d$, this dimension is polynomial in $n$.
By the linear-algebraic analogue of breadth-first search, we may thus list a basis 
for $\cA$ and its structure constants with respect to that basis.

We claim that the map $\tA \mapsto \cA$ is a reduction. Suppose $\tB$ is another tensor of the same dimension, and let $\cB$ be the associated algebra as above. We claim that $\tA \cong \tB$ as $d$-tensors if and only if $\cA \cong \cB$ as algebras.

\textbf{For the only if direction,} suppose $\tA \cong \tB$ via $(P_1, P_2, \dotsc, P_d) \in \GL(n_1, \F) \times \dotsb \times \GL(n_d, \F)$, that is
\[
\tA(i_1, \dotsc, i_d) = \sum_{j_1, \dotsc, j_d} (P_1)_{i_1,j_1} \dotsb (P_d)_{i_d,j_d} \tB(j_1, \dotsc, j_d)
\]
for all $i_1, \dotsc, i_d$. Then we claim that the block-diagonal matrix $P = 
\diag(P_1, P_2, \dotsc, P_{d-1}, P_d^{-1}) \in \GL(n,\F)$ (where $n=\sum_{i=1}^d 
n_i$), together with mapping $e_i$ to $e_i$, induces an isomorphism from $\cA$ to 
$\cB$. Note that $P$ itself is not an isomorphism, as $\dim \cA \approx n^d$, but 
$P$ specifies a linear map on the generators of $\cA$, which we may then 
exend to all of $\cA$.

First let us see that $P$ indeed gives a well-defined homomorphism $\cA \to \cB$. Since $P$ is only defined on the generators and is, by definition, extended by distributivity, the only thing to check here is that $P$ sends the relations of $\cA$ into the relations of $\cB$. Let $y_{1,1}, \dotsc, y_{1,n_1}, \dotsc, y_{d,n_d}, e_1, \dotsc, e_{d}$ denote the basis of $\cB$ as above. The map $P$ is defined by $P(e_i) = e_i$, 
\[
P(x_{i,a}) = \sum_{a'=1}^{n_i} (P_i)_{aa'} y_{i,a'} \qquad \text{ for } i=1,\dotsc,d-1
\]
and
\[
P(x_{d,a}) = \sum_{a'=1}^{n_d} (P_d^{-t})_{aa'} y_{d,a'}.
\]
By left multiplying by $P_d^t$, we may rewrite this last equation as
\[
y_{d,a} = \sum_{a'=1}^{n_d} (P_d)_{a',a} P(x_{d,a'}),
\]
note the transpose.

To check the relations, let us write out the Leavitt path algebra relations explicitly for our graph, in our notation. The generators of $\cA$ are $x_{1,1}, x_{1,2}, \dotsc, x_{1,n_1}, x_{2,1}, x_{2,2}, \dotsc, x_{2,n_2}, \dotsc, x_{d,n_d}, e_1, \dotsc, e_d$, and the relations are (\ref{eq:rel_tensor}) and the quiver relations:
\begin{eqnarray}
e_i e_j & = & \delta_{i,j} e_i \nonumber \\
e_i x_{j,a} & = & (\delta_{i,j} + \delta_{i,1}\delta_{j,d}) x_{j,a} \nonumber \\
x_{j,a} e_i & = & (\delta_{j+1,i} + \delta_{j,d}\delta_{i,d}) x_{j,a} \nonumber \\
x_{i,a} x_{d,b} & = & 0 \label{eq:rel_d_annihilate} \\
x_{d,b} x_{i,a} & = & 0 \quad (i < d) \nonumber \\
x_{i,a} x_{j,b} & = & 0 \quad \text{ if } j \neq i+1 \nonumber 
\end{eqnarray}
Note that the set $e_i \cA e_j$ is linearly spanned by the paths $i \to j$ in this graph.

The relations involving the $e_i$ are easy to verify, since they only depend on the first subscript of $x_{i,a}$ (resp., $y_{j,b}$), and $P$ does not alter this subscript. 

For relation (\ref{eq:rel_d_annihilate}), we have:
\begin{eqnarray*}
P(x_{i,a} x_{d,b}) & = & P(x_{i,a}) P(x_{d,b}) \\
& = & \left(\sum_{a'=1}^{n_i} (P_i)_{aa'} y_{i,a'}\right)\left(\sum_{b'=1}^{n_d} (P_d)_{bb'} y_{d,b'}\right) \\
& = & \sum_{a'=1}^{n_i} \sum_{b'=1}^{n_d} (P_i)_{aa'} (P_d)_{bb'} y_{i,a'} y_{d,b'} = 0,
\end{eqnarray*}
where the final inequality comes from the defining relations $y_{i,a'} y_{d,b'} = 0$ in $\cB$.

The verification for remaining quiver relations is similar, since $P$ does not alter the start and end vertices of any arrow (though it may send a single arrow $i \to j$ in $\cA$ to a linear combination of arrows $i \to j$ in $\cB$).

We now verify the relation (\ref{eq:rel_tensor}). We have
\begin{align*}
& P(x_{1,i_1} x_{2,i_2} \dotsb x_{d-1,i_{d-1}}) \\
& =  \sum_{j_1=1}^{n_1} \sum_{j_2=1}^{n_2} \dotsb \sum_{j_{d-1}=1}^{n_{d-1}} (P_1)_{i_1,j_1} (P_2)_{i_2,j_2} \dotsb (P_{d-1})_{i_{d-1}, j_{d-1}} y_{1,j_1} y_{2,j_2} \dotsb y_{d-1,j_{d-1}} \\
& =  \sum_{j_1,j_2,\dotsb,j_{d-1}} (P_1)_{i_1,j_1} (P_2)_{i_2,j_2} \dotsb (P_{d-1})_{i_{d-1}, j_{d-1}} \sum_{j_d=1}^{n_d} \tB(j_1,j_2,\dotsc,j_d) y_{d,j_d}  \\
& =  \sum_{j_1, \dotsb, j_{d-1}} (P_1)_{i_1,j_1} (P_2)_{i_2,j_2} \dotsb (P_{d-1})_{i_{d-1}, j_{d-1}} \sum_{j_d=1}^{n_d} \tB(j_1,j_2,\dotsc,j_d) \sum_{i_d=1}^{n_d} (P_d)_{i_d,j_d} P(x_{d,i_d}) \\
& =  \sum_{i_d=1}^{n_d} \left( \sum_{j_1, \dotsb, j_{d-1},j_d} (P_1)_{i_1,j_1} \dotsb (P_d)_{i_d,j_d} \tB(j_1,\dotsc,j_d) \right) P(x_{d,i_d}) \\
& =  \sum_{i_d=1}^{n_d} \tA(i_1,\dotsc,i_d) P(x_{d,i_d}),
\end{align*}
as desired. Thus the map $\cA \to \cB$ induced by $P$ is an algebra homomorphism.

Next, since $P$ is an isomorphism of $(d+n)$-dimensional vector spaces, the map it induces $\cA \to \cB$ is surjective on the generators of $\cB$, whence it is surjective onto all of $\cB$. Finally, since $\dim \cA = \dim \cB < \infty$, any linear surjective map $\cA \to \cB$ is automatically bijective, so this map is indeed an isomorphism of algebras.

\textbf{For the if direction,} suppose that $f\colon \cA \to \cB$ is an isomorphism of algebras. Since the Jacobson radical is characteristic, we have $f(R(\cA)) = R(\cB)$. Then $\{f(e_v) : v \in V\}$ is a set of primitive orthogonal idempotents in $\cB$, and their span $T = \langle f(e_v) : v \in V \rangle$ is a separable subalgebra (isomorphic to $\F^n$) such that $\cB = T \oplus R(\cB)$. By the Wedderburn--Mal'cev Theorem (Theorem~\ref{thm:WM}(\ref{thm:WM:conj})), there is some $r \in R(\cB)$ such that $(1+r)T(1+r)^{-1} = \langle e_1, \dotsc, e_n \rangle =: S$. Since the $e_i$ are the only primitive idempotents in $S$, we must have that $(1+r)f(e_i)(1+r)^{-1} = e_{\pi(i)}$ for all $i$ and some permutation $\pi \in S_n$. 

Next we will show that this permutation is in fact the identity, so that $(1+r)f(e_i)(1+r)^{-1} = e_i$ for all $i$. For this, consider $\cA' = \cA / R(\cA)^2$ and similarly $\cB'$. These are precisely the algebras considered by Grigoriev \cite{Grigoriev83} (reproduced as Theorem~\ref{thm:grigoriev} above). Since $R(\cA)$ is characteristic, so is its square, and thus $f$ induces an isomorphism $\cA' \stackrel{\cong}{\to} \cB'$. By Theorem~1 of Grigoriev \cite{Grigoriev83}, any isomorphism $\cA' \to \cB'$ induces an isomorphism of the corresponding graphs, so this isomorphism must map $e_i$ to $e_i$ for each $i$ (since our graph $G$ has no automorphisms). Thus $\pi$ must be the identity, and $(1+r)f(e_i)(1+r)^{-1} = e_i$ for all $i$.

Since conjugation is an automorphism, let $f'\colon \cA \to \cB$ be $c_{1+r} \circ f$, where $c_{1+r}(b) = (1+r)b(1+r)^{-1}$. By the above, $f'(e_i) = e_i$ for all $i$. 
Thus $f'(e_i \cA e_j) = e_i \cB e_j$. In particular, define $P_i$ to be the restriction of $f'$ to $e_i \cA e_{i+1}$ for $i=1,\dotsc,d-1$ and $P_d$ to be the restriction of $f'$ to $e_1 \cA e_d$. Then we have that $P_i$ is a linear bijection from the span of $x_{i,1},\dotsc,x_{i,n_i}$ to the span of $y_{i,1},\dotsc,y_{i,n_i}$ for all $i$. We claim that $P = (P_1,\dotsc,P_{d-1},P_d^{-t})$ is a tensor isomorphism $\tA \to \tB$, that is, 
\[
\tA(i_1, \dotsc, i_d) = \sum_{j_1, \dotsc, j_d} (P_1)_{i_1,j_1} \dotsb (P_d^{-t})_{i_d,j_d} \tB(j_1, \dotsc, j_d).
\]
From the fact that $f'$ is an isomorphism, we have
\begin{eqnarray*}
\sum_{i_d=1}^{n_d} \tA(i_1, \dotsc, i_d) f'(x_{d,i_d}) & = & f'(x_{1,i_1} x_{2,i_2} \dotsb x_{d-1,i_{d-1}}) \\
\sum_{i_d=1}^{n_d} \tA(i_1, \dotsc, i_d) \sum_{j_d=1}^{n_d} (P_d)_{i_d,j_d} y_{d,j_d} & = & f'(x_{1,i_1}) f'(x_{2,i_2}) \dotsb f'(x_{d-1,i_{d-1}}) \\
 & = & \sum_{j_1,\dotsc,j_{d-1}} (P_1)_{i_1,j_1} (P_2)_{i_2,j_2} \dotsb (P_{d-1})_{i_{d-1}, j_{d-1}} y_{1,j_1} y_{2,j_2} \dotsb y_{d-1,j_{d-1}} \\
 & = & \sum_{j_1, \dotsc, j_{d-1}} (P_1)_{i_1,j_1} (P_2)_{i_2,j_2} \dotsb (P_{d-1})_{i_{d-1}, j_{d-1}} \sum_{j_d=1}^{n_d} \tB(j_1,\dotsc,j_d) y_{d,j_d}
\end{eqnarray*}
For each $j_d \in \{1,\dotsc,n_d\}$, equating the coefficient of $y_{d,j_d}$ gives
\begin{eqnarray*}
\sum_{i_d=1}^{n_d} \tA(i_1, \dotsc, i_d)  (P_d)_{i_d,j_d} & = & \sum_{j_1, \dotsc, j_{d-1}} (P_1)_{i_1,j_1} (P_2)_{i_2,j_2} \dotsb (P_{d-1})_{i_{d-1}, j_{d-1}} \tB(j_1,\dotsc,j_d)
\end{eqnarray*}
Let $\tA(i_1, \dotsc, i_{d-1}, -)$ be the natural row vector of length $n_d$, and similarly for $\tB(j_1, \dotsc, j_{d-1}, -)$. Then we may rewrite the preceding set of $n_d$ equations (one for each choice of $j_d$) in matrix notation as
\[
\tA(i_1, \dotsc, i_{d-1}, -) \cdot P_d = \sum_{j_1, \dotsc, j_{d-1}} (P_1)_{i_1,j_1} (P_2)_{i_2,j_2} \dotsb (P_{d-1})_{i_{d-1}, j_{d-1}} \tB(j_1,\dotsc, j_{d-1},-)
\]
Right multiplying by $P_d^{-1}$, we then get
\begin{eqnarray*}
\tA(i_1, \dotsc, i_{d-1},-) & = & \sum_{j_1, \dotsc, j_{d-1}} (P_1)_{i_1,j_1} (P_2)_{i_2,j_2} \dotsb (P_{d-1})_{i_{d-1}, j_{d-1}} \tB(j_1,\dotsc, -) P_d^{-1} \\
\tA(i_1,\dotsc,i_d) & = & \sum_{j_1, \dotsc, j_{d-1},j_d} (P_1)_{i_1,j_1} (P_2)_{i_2,j_2} \dotsb (P_{d-1})_{i_{d-1}, j_{d-1}} \tB(j_1,\dotsc, j_d) (P_d^{-1})_{j_d,i_d} \\
& = & \sum_{j_1,\dotsc,j_d} (P_1)_{i_1,j_1} (P_2)_{i_2,j_2} \dotsb (P_{d-1})_{i_{d-1}, j_{d-1}} (P_d^{-t})_{i_d,j_d} \tB(j_1,\dotsc, j_d),
\end{eqnarray*}
as claimed.
\end{proof}

\section{Conclusion: universality and open questions} \label{sec:conclusion}

\subsection{Towards universality for basis-explicit linear structures} \label{sec:universality}
A classic result is that \GI is complete for isomorphism problems of explicitly given structures (see, e.\,g., \cite[Section~15]{ZKT}). Here we formally state the linear-algebraic analogue of this result, and observe trivially that the results of \cite{FGS19} already show that 3-Tensor Isomorphism is universal among what we call ``basis-explicit'' (multi)linear structures of degree 2.

First let us recall the statement of the result for \GI, so we can develop the 
appropriate analogue for tensor isomorphism. A \emph{first-order signature} is a 
list of positive integers $(r_1, r_2, \dotsc, r_k; f_1, \dotsc, f_\ell)$; a 
\emph{model} of this signature consists of a set $V$ (colloquially referred to as 
``vertices''), $k$ relations $R_i\subseteq V^{r_i}$, 
and $\ell$ functions $F_i \colon V^{f_i} 
\to V$. The numbers $r_i$ are thus the arities of the relations $R_i$, and the 
$f_i$ are the arities of the functions $F_i$.\footnote{Sometimes one also includes 
constants in the definition, but these can be handled as relations of arity 1. 
While we could have done the same for functions, treating a function of arity $f$ 
as its graph, which is a relation of arity $f+1$, distinguishing between relations 
and functions will be useful when we come to our linear-algebraic analogue.}
Two such models $(V; R_1, \dotsc, R_k; F_1, \dotsc, F_\ell)$ and $(V'; R_1', \dotsc, R_k'; F_1', \dotsc, F_\ell')$ are isomorphic if there is a bijection $\varphi\colon V \to V'$ that sends $R_i$ to $R_i'$ for all $i$ and $F_i$ to $F_i'$ for all $i$. In symbols, $\varphi$ is an isomorphism if $(v_1, \dotsc, v_{r_i}) \in R_i \Leftrightarrow (\varphi(v_1), \dotsc, \varphi(v_{r_i})) \in R_i'$ for all $i$ and all $v_* \in V$, and similarly if $\varphi(F_i(v_1, \dotsc, v_{f_i})) = F_i'(\varphi(v_1), \dotsc, \varphi(v_{f_i}))$ for all $i$ and all $v_* \in V$. By an ``explicitly given structure'' or ``explicit model'' we mean a model where each relation $R_i$ is given by a list of its elements and each function is given by listing all of its input-output pairs. Fixing a signature, the isomorphism problem for that signature is to decide, given two explicit models of that signature, whether they are isomorphic. This isomorphism problem is directly encoded into the isomorphism problem for edge-colored hypergraphs, which can then be reduced to \GI using standard gadgets.

For example, the signature for directed graphs (possibly with self-loops) is simply $\sigma=(2; )$---its models are simply binary relations. If one wants to consider graphs without self-loops, this is a special case of the isomorphism problem for the signature $\sigma$, namely, those explicit models in which $(v,v) \notin R_1$ for any $v$. Note that a graph without self-loops is never isomorphic to a graph with self-loops, and two directed graphs without self-loops are isomorphic as directed graphs if and only if they are isomorphic as models of the signature $\sigma$. In other words, the isomorphism problem for simple directed graphs really is just a special case. The same holds for undirected graphs without self-loops, which are simply models of the signature $\sigma$ in which $(v,v) \notin R_1$ and $R_1$ is symmetric. As another example, the signature for finite groups is $\gamma = (1; 1,2)$: the first relation $R_1$ will be a singleton, indicating which element is the identity, the function $F_1$ is the inverse function $F_1(g) = g^{-1}$, and the second function $F_2$ is the group multiplication $F_2(g,h) = gh$. Of course, models of the signature $\gamma$ can include many non-groups as well, but, as was the case with directed graphs, a group will never be isomorphic to a non-group, and two groups are isomorphic as models of $\gamma$ iff they are isomorphic as groups.

A natural linear-algebraic analogue of the above is as follows. One additional feature we add here for purposes of generality is that we need to make room for dual vector spaces. A \emph{linear signature} is then a list of pairs of nonnegative integers $((r_1, r_1^*), \dotsc, (r_k, r_k^*); (f_1, f_1^*), \dotsc, (f_\ell, f_\ell^*))$ with the property that $r_i + r_i^* > 0$ and $f_i + f_i^* > 0$ for all $i$. By the arity of the $i$-th relation (resp., function) we mean the sum $r_i + r_i^*$ (resp., $f_i + f_i^*$).

\begin{definition}[Linear signature, basis-explicit]
Given a linear signature 
$$\sigma = ((r_1, r_1^*), \dotsc, (r_k, r_k^*); (f_1, f_1^*), \dotsc, (f_\ell, 
f_\ell^*)),$$ a \emph{linear model} for $\sigma$ over a field $\F$ consists of an 
$\F$-vector space $V$, and linear subspaces $R_i \leq V^{\otimes r_i} \otimes 
(V^*)^{\otimes r_i^*}$ for $1 \leq i \leq k$ and linear maps $F_i \colon 
V^{\otimes f_i} \otimes (V^*)^{\otimes f_i^*} \to V$ for $1 \leq i \leq \ell$. Two 
such linear models $(V; R_1, \dotsc, R_k; F_1, \dotsc, F_\ell), (V'; R_1', \dotsc, 
R_k'; F_1', \dotsc, F_\ell')$ are \emph{isomorphic} if there is a linear bijection 
$\varphi\colon V \to V'$ that sends $R_i$ to $R_i'$ for all $i$ and $F_i$ to 
$F_i'$ for all $i$ (details below). 

A \emph{basis-explicit linear model} is given by a basis for each $R_i$, and, for each element of a basis of the domain of $F_i$, the value of $F_i$ on that element. Vectors here are written out in their usual dense coordinate representation.
\end{definition}

In particular, this means that an element of $V^{\otimes r}$---say, a basis element of $R_1$---is written out as a vector of length $(\dim V)^r$. We will only be concerned with finite-dimensional linear models.

Given $\varphi \colon V \to V'$, let $\varphi^{\otimes r_i \otimes r_i^*}$ denote the linear map $\varphi^{\otimes r_i \otimes r_i^*}\colon V^{\otimes r_i} \otimes (V^*)^{\otimes r_i^*} \to V'^{\otimes r_i} \otimes (V'^*)^{\otimes r_i^*}$ which is defined on basis vectors factor-wise: $\varphi^{\otimes r_i \otimes r_i^*}(v_1 \otimes \dotsb \otimes v_{r_i} \otimes \ell_1 \otimes \dotsb \otimes \ell_{r_i^*}) = \varphi(v_1) \otimes \dotsb \otimes \varphi(v_{r_i}) \otimes \varphi^*(\ell_1) \otimes \dotsb \otimes \varphi^*(\ell_{r_i^*})$, and then extended to the whole space by linearity. (Recall that $V^* = \hom(V, \F)$, so elements of $V^*$ are linear maps $\ell\colon V \to \F$, and thus $\varphi^*(\ell) := \ell \circ \varphi^{-1}$ is a map from $V' \to V \to \F$, i.\,e., an element of $V'^*$, as desired). Similarly, when we say that $\varphi$ sends $F_i$ to $F_i'$, we mean that $\varphi(F_i(v_1 \otimes \dotsb \otimes v_{f_i} \otimes \ell_1 \otimes \dotsb \otimes \ell_{f_i^*})) = F_i'(\varphi^{\otimes f_i \otimes f_i^*}(v_1 \otimes \dotsb \otimes v_{f_i} \otimes \ell_1 \otimes \dotsb \otimes \ell_{f_i^*}))$.

\begin{remark}
We use the term ``basis-explicit'' rather than just ``explicit,'' because over a \emph{finite} field, one may also consider a linear model of $\sigma$ as an explicit model of a different signature (where the different signature additionally encodes the structure of a vector space on $V$, namely, the addition and scalar multiplication), and then one may talk of a single mathematical object having explicit representations---where everything is listed out---and basis-explicit representations---where things are described in terms of bases. An example of this distinction arises when considering isomorphism of $p$-groups of class $2$: the ``explicit'' version is when they are given by their full multiplication table (which reduces to \GI), while the ``basis-explicit'' version is when they are given by a generating set of matrices or a polycyclic presentation (which \GI reduces \emph{to}). 
\end{remark}

\begin{theorem}[{Futorny--Grochow--Sergeichuk \cite{FGS19}}]
Given any linear signature $\sigma$ where all relationship arities are at most 3 and all function arities are at most 2, the isomorphism problem for finite-dimensional basis-explicit linear models of $\sigma$ reduces to \ThreeTIlong in polynomial time.
\end{theorem}

Because of the equivalence between \DeeTIlong and \ThreeTIlong (Theorem~\ref{thm:d_to_3} + \cite{FGS19}), we expect the analogous result to hold for arbitrary $d$. Thus an analogue of the results of \cite{FGS19} for $d$-tensors would yield the full analogue of the universality result for \GI.

\begin{question}
Is \DeeTIlong universal for isomorphism problems on $d$-way arrays? That is, prove the analogue of the results of \cite{FGS19} for $d$-way arrays for any $d \geq 3$.
\end{question}

\subsection{Other open questions}
Our search-to-decision reduction (Theorem~\ref{thm:search_decision}) produces instances of dimension $O(n^2)$ from instances of dimension $n$. As stated, this means that a simply-exponential ($q^{\tilde O(n)}$-time) decision algorithm would result only in a $q^{\tilde O(n^2)}$ search algorithm, but the latter runtime is trivial. We note that it may be possible to alleviate this blow-up by attempting to generalize the logarithmic-size ``coloring palette'' construction for reducing \algprobm{Colored GI} to \GI  from the graph case to the linear-algebraic case.

\begin{question} \label{question:search_decision}
Is there a search-to-decision reduction for \AltMatSpIsomWords (and, consequently, isomorphism of $p$-groups of class 2 and exponent $p$, given in their natural succinct encoding) that runs in time $q^{\tilde O(n)}$, and produces instances of quasi-linear ($\tilde O(n)$) dimension?
\end{question}

In Section~\ref{sec:GI_code} we gave several different reductions from \GI to \AltMatSpIsomWords. To summarize, they are: 
\begin{enumerate}
\item A direct reduction from \GI to \AltMatSpIsomlong (\Prop{prop:GI})
\item \GI $\leq$ \MatLieConjlong \cite{GrochowLie}, which in turn reduces to \ThreeTI \cite{FGS19}, and then to \AltMatSpIsomlong (\Thm{thm:main}); 
\item \GI $\leq$ \CodeEq \cite{PR97, Luks, miyazakiCodeEq}, \CodeEq $\leq$ \MatLieConjlong \cite{GrochowLie}, and then follow the same reductions as in (1);
\item \GI $\leq$ \MonCodeEqlong (the same reduction from \cite{PR97} works for monomial equivalence of codes, see \cite{GrochowLie}), which in turn reduces to \ThreeTI (\Prop{prop:MonCodeEq}), and thence to \AltMatSpIsomlong (\Thm{thm:main})
\item \GI $\leq$ \algprobm{Algebra Isomorphism} \cite{Grigoriev83, AS05}, which reduces to \ThreeTI \cite{FGS19}, and then to \AltMatSpIsomWords (\Thm{thm:main}).
\end{enumerate}

Can one prove that these reductions are all distinct? Are some of them equivalent in some natural sense, e.\,g., up to a change of basis?

Next, most of our results hold for arbitrary fields, or arbitrary fields with 
minor restrictions. However, in all of our reductions, we reduce one problem over 
$\F$ to another problem over the same field $\F$. 

\begin{question}
What is the relationship between $\cc{TI}$ over different fields? In particular, 
what is the relationship between $\cc{TI}_{\F_p}$ and $\cc{TI}_{\F_{p^e}}$, between $\cc{TI}_{\F_p}$ and $\cc{TI}_{\F_q}$ for coprime 
$p,q$, or between $\cc{TI}_{\F_p}$ and $\cc{TI}_{\Q}$?
\end{question}

We note that even the relationship between $\cc{TI}_{\F_p}$ and 
$\cc{TI}_{\F_{p^e}}$ is not particularly clear. For matrix \emph{tuples} (rather than 
spaces; equivalently, representations of finitely generated algebras) it is the 
case that for any extension field $\mathbb{K} \supseteq \F$, two matrix tuples 
over $\F$ are $\F$-equivalent (resp., conjugate) if and only if they are 
$\mathbb{K}$-equivalent \cite{KL86} (see \cite{dSPFields} for a simplified proof). 
However, for equivalence of tensors this need not be the case. This seems closely related to the so-called ``problem of forms'' for various algebras, namely the existence of algebras that are not isomorphic over $\F$, but which become isomorphic over an extension field.

\begin{example}[Non-isomorphic tensors isomorphic over an extension field] 
Over $\R$, let $M_1 = I_4$ and let $M_2 = \diag(1,1,1,-1)$. Since these two matrices have different signatures, they are not isometric over $\R$; since they have the same rank, they \emph{are} isometric over $\C$. To turn this into an example of 3-tensors, first we consider the corresponding instance of \algprobm{Matrix Space Isometry} given by $\mathcal{M}_1 = \langle M_1 \rangle$ and $\mathcal{M}_2 = \langle M_2 \rangle$. Note that $\mathcal{M}_1 = \{\lambda I_4 : \lambda \in \R\}$, so the signatures of all matrices in $\mathcal{M}_1$ are $(4,0)$, $(0,0)$, or $(0,4)$. Similarly, the signatures appearing in $\mathcal{M}_2$ are $(3,1)$, $(0,0)$, and $(1,3)$, so these two matrix spaces are not isometric over $\R$, though they are isometric over $\C$ since $M_1$ and $M_2$ are. Finally, apply the reduction from \algprobm{Matrix Space Isometry} to \ThreeTI \cite{FGS19} to get two 3-tensors $\tA_1, \tA_2$. Since the reduction itself is independent of field, if we consider it over $\R$ we find that $\tA_1$ and $\tA_2$ must not be isomorphic 3-tensors over $\R$, but if we consider the reduction over $\C$ we find that they are isomorphic as 3-tensors over $\C$.

Similar examples can be constructed over finite fields $\F$ of odd characteristic, taking $M_1 = I_2$ and $M_2 = \diag(1,\alpha)$ where $\alpha$ is a non-square in $\F$ (and replacing the role of $\C$ with that of $\mathbb{K} = \F[x] / (x^2 - \alpha)$). Instead of signature, isometry types of matrices over $\F$ are characterized by their rank and whether their determinant is a square or not. In this case, since our matrices are even-dimensional diagonal matrices, scaling them multiplies their determinant by a square. Thus every matrix in $\mathcal{M}_1$ will have its determinant being a square in $\F$, and every nonzero matrix in $\mathcal{M}_2$ will not, but in $\mathbb{K}$ they are all squares.
\end{example}

It would also be interesting to study the complexity of other group actions on 
tensors and how they relate to the problems here. For example, the action of 
unitary groups $U(\C^{n_1}) \times \dotsb \times U(\C^{n_d})$ on $\C^{n_1} \otimes 
\dotsb \otimes \C^{n_d}$ classifies pure quantum states up to ``local unitary 
operations,'' and the action of $\SL(U_1) \times \dotsb \times \SL(U_d)$ on $U_1 
\otimes \dotsb \otimes U_d$, over $\C$, is the well-studied action by stochastic 
local operations with classical communication (SLOCC) on quantum states (e.\,g., 
\cite{gourWallach, miyake, SLOCC4}). Isomorphism of $m$-dimensional lattices in 
$n$-dimensional space can be seen as the natural action of $O_n(\R) \times 
\GL_m(\Z)$ by left and right multiplication on $n \times m$ matrices. As another 
example, orbits for several of the natural actions of $\GL_n(\Z) \times \GL_m(\Z) 
\times \GL_r(\Z)$ on 3-tensors over $\Z$, even for small values of $n,m,r$, are 
the fundamental objects in Bhargava's seminal work on higher composition laws 
\cite{bhargava,bhargava2,bhargava3,bhargava4}. We note that while the orthogonal 
group $O(V)$ is the stabilizer of a 2-form on $V$ (that is, an element of $V 
\otimes V$) and $\SL(V)$ is the stabilizer of the induced action on 
$\bigwedge^{\dim V} V$ (by the determinant)---so gadgets similar to those in this 
paper might be useful---$\GL_n(\Z)$ is not the stabilizer of any such structure. 

In Remark~\ref{rmk:dto3_dim} we observed that any reduction (in the sense of 
\Sec{sec:reductions}) from \DeeTI to \ThreeTI must have a blow-up in dimension 
which is asymptotically $n^{d/3}$, while our construction uses dimension $O(d^2 
n^{d-1})$.

\begin{question}
Is there a reduction from \DeeTI to \ThreeTI (as in \Sec{sec:reductions}) such that the dimension of the output is $\poly(d) \cdot n^{d/3(1 + o(1))}$? 
\end{question}

Finally, in terms of practical algorithms, we wonder how well modern \algprobm{SAT} solvers would do on instances of \ThreeTIlong over $\F_2$ (or over other finite fields, encoded into bit-strings). 

\section*{Acknowledgments}
The authors would like to thanks James B. Wilson for related discussions, and Uriya First, Lek-Heng Lim, and J. M. Landsberg for help searching for references asking whether $\DeeTI$ could reduce to $\ThreeTI$.
J. A. G. would like to thank V. Futorny and V. V. Sergeichuk for their collaboration on the related work \cite{FGS19}.
Ideas leading to this work originated from the workshop ``Wildness
in computer science, physics, and mathematics'' at the Santa Fe Institute. 
Both authors were supported by NSF grant DMS-1750319. Y. Q. was partly 
supported by Australian Research Council DECRA DE150100720.

\appendix

\section{Reducing \CubicFormlong to \DFormlong} \label{app:cubic}

\begin{proposition} \label{prop:cubic_to_d}
\CubicFormlong reduces to \DFormlong, for any $d \geq 3$.
\end{proposition}

We suspect that a similar construction would give a reduction from 
\algprobm{Degree-$d'$ Form Equivalence} to \DFormlong for any $d' \leq d$, but our 
argument relies on a case analysis that is somewhat specific to $d'=3$. Our 
argument might be adaptable to any fixed value of $d'$ the prover desires, with a 
consequently more complicated case analysis, but to prove it for all $d'$ 
simultaneously seems to require a different argument.

\begin{proof}
The reduction itself is quite simple: $f \mapsto z^{d-3} f$, where $z$ is a new variable not appearing in $f$. If $A$ is an equivalence between $f$ and $g$---that is, $f(x) = g(Ax)$---then $\diag(A, 1_z)$ is an equivalence from $z^{d-3} f$ to $z^{d-3}g$. Conversely, suppose $\tilde f = z^{d-3} f$ is equivalent to $\tilde g = z^{d-3} g$ via $\tilde f(x) = \tilde g(Bx)$. We split the proof into several cases. 

\paragraph{If $d=3$,} then $z$ is not present so we already have that $f$ and $g$ are equivalent.

\paragraph{If $f$ is not divisible by $\ell^{d-3}$ for some linear form $\ell$,} then $z^{d-3}$ is the unique factor in both $z^{d-3} f$ and $z^{d-3} g$ which is raised do the $d-3$ power. Thus any equivalence $B$ between these two must map $z$ to itself, hence has the form
\[
B = \left(\begin{array}{ccc|c}
* & \dotsc & * & 0 \\
\vdots & \ddots & \vdots & \vdots \\
* & \dotsc & * & 0 \\ \hline
* & \dotsc & * & 1
\end{array}\right),
\]
(if we put $z$ last in our basis, and think of the matrix as acting on the left of the column vectors corresponding to the variables). However, since both $f$ and $g$ do not depend on $z$, it must be the case that whatever contributions $z$ makes to $g(Bx)$, they all cancel. More precisely, all monomials involving $z$ in $g(Bx)$ must cancel, so if we alter $B$ into $\tilde B$ that $\tilde Bx_i$ never includes $z$ (that is, if we make the stars in the last row above all zero), then $g(\tilde B x) = g(Bx)$, hence $f(x) = g(\tilde Bx)$, so $f$ and $g$ are equivalent.

The preceding case always applies when $d > 6$, for then $d-3 > 3$, but $\deg f = 3$. We are left to handle the following cases:
\begin{enumerate}
\item $d \leq 6$ and $f$ is a product of linear forms;
\item $d=4$, $f$ is a product of a linear form and an irreducible quadratic form.
\end{enumerate}

\paragraph{Suppose $f$ is a product of linear forms,} then let us define $\rk(f)$ as the number of linearly independent linear forms appearing in the factorization of $f$. Note that if $\rk(f)=1$, then $f = \alpha \ell^3$ for some $\alpha \in \F$, if $\rk(f)=2$, then $f = \ell_1^2 \ell_2$ (now we can absorb any constant into $\ell_2$), and if $\rk(f)=3$ then $f = \ell_1 \ell_2 \ell_3$ with all $\ell_i$ linearly independent. Then we have that $f \sim g$ if and only if $g$ is also a product of linear forms of the same rank. For $\GL_n$ acts transitively on $k$-tuples of linearly independent vectors for all $k \leq n$, and in order to have $\rk(f)$ linearly independent forms, we must have $n \geq \rk(f)$. Since we have supposed $z^{d-3} f \sim z^{d-3} g$, by uniqueness of factorization $g$ must be a product of linear forms of the same rank as $f$, and thus indeed $f \sim g$.

\paragraph{If $d=4$ and $f = \ell \varphi$ where $\ell$ is linear and $\varphi$ is an irreducible quadratic,} then to understand the situation we begin by first doing a change of basis on $f$ to put $\varphi$ into a form in which its kernel is evident. Note that none of these simplifications are part of the reduction, but rather they are to help us prove that the reduction works. Thinking of $\varphi$ as given by its matrix $M_{\varphi}$ such that $\varphi(x) = x^t M_{\varphi} x$, we can always change basis to get $M_{\varphi}$ into the form
\[
\begin{bmatrix}
M' & 0 \\
0& 0_{n-r}
\end{bmatrix}
\]
where $r = \rk(M_{\varphi}) = \rk(M')$. Since $\varphi$ does not depend on $z$, if we think of $\varphi$ as a quadratic form on $\{x_1, \dotsc, x_n, z\}$, then the matrices are the same, but larger by one additional zero row and column.

Next we will try to simplify $\ell$ as much as possible while maintaining the (new) form of $M_\varphi = \diag(M',\vzero)$. For this we first compute the stabilizer of the new form of $M_{\varphi}$. We can compute the stabilizer as the set of invertible matrices $A$ such that:
\begin{eqnarray*}
\begin{bmatrix}
A_{11}^t & A_{21}^t \\
A_{12}^t & A_{22}^t
\end{bmatrix}
\begin{bmatrix}
M' & 0 \\
0 & 0_{n-r+1}
\end{bmatrix}
\begin{bmatrix}
A_{11} & A_{12} \\
A_{21} & A_{22} 
\end{bmatrix}
=
\begin{bmatrix}
M' & 0 \\
0 & 0_{n-r+1}
\end{bmatrix}. 
\end{eqnarray*}
This turns into the following equations on the blocks of $X$:
\[
\begin{array}{rclcrcl}
A_{11}^t M' A_{11} & = & M' & \qquad & 
A_{12}^t M' A_{11} & = & 0 \\
A_{12}^t M' A_{12} & = & 0 & \qquad & 
A_{11}^t M' A_{12} & = & 0
\end{array}
\]
From the first equation and the fact that $M'$ is full rank, we find that $A_{11}$ must be an invertible $r \times r$ matrix. From the next equation and the fact that both $M$ and $A_{11}$ are full rank, we then find that $A_{12} = 0$. Thus the stabilizer of $M_{\varphi}$ is:
\[
S := \left\{ \begin{bmatrix} A_{11} & 0 \\ A_{21} & A_{22} \end{bmatrix} : A_{11}^t M' A_{11} = M' \text{ and } A_{22} \text{ is invertible} \right\}.
\]

Now we simplify $\ell$. Note that $S$ acts on $\ell$ as a column vector. Consider $\ell = \sum_{i=1}^n \ell_i x_i$, with $\ell_i \in \F$; we will say ``$\ell$ contains $x_i$'' if and only if $\ell_i \neq 0$. If $\ell$ contains some $x_{r+k}$ with $k \geq 1$, then by setting $A_{11} = I_r$ and $A_{21} = 0$, we may choose $A_{22}$ to be any invertible matrix which sends $(\ell_{r+1}, \dotsc, \ell_n, \ell_{n+1})$ (recall the trailing $\ell_{n+1}$ for the $z$ coordinate) to $(1,0,\dotsc,0)$, and thus without loss of generality we may assume that $\ell$ only contains $x_i$ with $1 \leq i \leq r+1$. 

\renewcommand{\theenumi}{\alph{enumi}}
Next, note that if $\ell$ contains some $x_i$ for $1 \leq i \leq r$ \emph{and} $x_{r+1}$, then we may use the action of $S$ to eliminate the $x_{r+1}$. Namely, by taking $A_{11} = I_r$, $A_{22} = I_{n+1}$, and $A_{21} = (-\ell_{r+1} / \ell_i) E_{1i}$. This makes $\ell_i x_i$ in $\ell$ contribute $-\ell_{r+1}$ to the $x_{r+1}$ coordinate, eliminating $x_{r+1}$. Thus, under the action of $S$, we need only consider two cases for linear forms under the action of $S$: a linear form is equivalent to either 
\begin{enumerate}
\item one which contains some $x_i$ with $1 \leq i \leq r$, in which case we can bring it to a form in which it contains \emph{no} $x_{r+j}$ with $j \geq 1$ (and no $z$), \emph{or} 
\item it contains no $x_i$ with $1 \leq i \leq r$, in which case we can use the action of $S$ to bring it to the form $\ell = x_{r+1}$. 
\end{enumerate}
Let us call the corresponding linear forms ``type (a)'' and ``type (b).'' Note that the linear form $z$ is of type (b).

Now, write $f = \ell \varphi$ and $g = \ell' \varphi'$, and assume that we have applied the preceding change of basis to bring $f$ to the form specified above. Recall that we are assuming $\tilde f \sim \tilde g$, and need to show that $f \sim g$. If, after applying the same change of basis to $g$, we do not have $M_{\varphi'} = M_{\varphi}$, then $f \not\sim g$ and also $\tilde f \not\sim \tilde g$---contrary to our assumption---since $\varphi$ (resp., $\varphi'$) is the unique irreducible quadratic factor of $\tilde f$ (resp., $\tilde g$). So we may assume that, after this change of basis, $\varphi = \varphi'$, both of which have $M_{\varphi} = \diag(M', 0_{n-r+1})$ with $r = \rank(M_{\varphi})$. 

Next, since we are assuming $\tilde f \sim \tilde g$, and $z$ itself is of type (b), so it must be the case that the types of $\ell,\ell'$ are the same. Thus we have two cases to consider: either they are both of type (a), or both of type (b). 

\paragraph{Suppose both $\ell,\ell'$ are of type (a).} In this case, the equivalence between $\tilde f$ and $\tilde g$ cannot send $z$ to $\ell'$ and $\ell$ to $z$, for both $\ell,\ell'$ are of type (a), whereas $z$ is of type (b). Thus the equivalence between $\tilde f$ and $\tilde g$ must restrict to an equivalence between $f$ and $g$ (when we ignore $z$, or set its contribution to the other variables to zero, as in the above case where $f$ was not divisible by $\ell^{d-3}$).

\paragraph{Suppose both $\ell,\ell'$ are of type (b).} In this case, it is possible that the equivalence from $\tilde f$ to $\tilde g$ \emph{could} send $z$ to $\ell'$ and $\ell$ to $z$ (since all three of $\ell,\ell',z$ are in case (b)); however, we will see that in this case, even such a situation will not cause an issue. Without loss of generality, by the change of bases described above, we have $\tilde f = z x_{r+1} \varphi$ and $\tilde g = z \ell' \varphi$ (the same $\varphi$), where $\ell'$ contains no $x_i$ with $1 \leq i \leq r$. Using elements of $S$ with $A_{11} = I_r$, and $A_{21} = 0$, we then get an action of $\GL_{n-r+1}$ (via $A_{22}$) on linear forms in the variables $x_{r+1}, \dotsc, x_n, z$. Since $\ell'$ is linearly independent from $z$ (in particular, it does not contain $z$) and the action of $\GL$ is transitive on pairs of linearly independent vectors, we may use $S$ to fix $\varphi$ and $z$, and send $x_{r+1}$ to $\ell'$, giving the desired equivalence $f \sim g$.
\end{proof}

\bibliographystyle{alphaurl}
\bibliography{references}

\end{document}